\documentclass[11pt, a4paper]{amsart}

\usepackage{xcolor}
\definecolor{MyBlue}{cmyk}{1,0.13,0,0.63}
\definecolor{MyGreen}{cmyk}{0.91,0,0.88,0.52}
\newcommand{\mylinkcolor}{MyBlue}
\newcommand{\mycitecolor}{MyGreen}
\newcommand{\myurlcolor}{black}

\usepackage{hyperref}
\hypersetup{%
  bookmarksnumbered=true,bookmarksopen=false,%
  plainpages=false,
  linktocpage=true,%
  colorlinks=true,breaklinks=true,%
  linkcolor=\mylinkcolor,citecolor=\mycitecolor,urlcolor=\myurlcolor,%
  pdfpagelayout=OneColumn,%
  pageanchor=true,%
}

\usepackage{graphicx}
\usepackage{amsmath, amsthm,  amsfonts, amscd}
\usepackage{amssymb}
\usepackage{url}
\usepackage{fullpage}
\usepackage{mathtools}
\usepackage[all]{xy}
\usepackage{slashed}
\usepackage{multirow}

\newtheorem{thm}{Theorem}[section]
\newtheorem*{thm*}{Theorem}
\newtheorem{cor}[thm]{Corollary}
\newtheorem{lemma}[thm]{Lemma}
\newtheorem{prop}[thm]{Proposition}

\theoremstyle{definition}
\newtheorem{defn}[thm]{Definition}

\theoremstyle{remark}
\newtheorem{remark}[thm]{Remark}

\newtheorem{example}[thm]{Example}
\newtheorem{remarks}[thm]{Remarks}

\newcommand{\End}{\ensuremath{\mathrm{End}}}

\newcommand{\wt}{\ensuremath{\widetilde}}

\newcommand{\R}{\ensuremath{\mathbb{R}}}
\newcommand{\N}{\ensuremath{\mathbb{N}}}
\newcommand{\Z}{\ensuremath{\mathbb{Z}}}

\newcommand{\C}{\ensuremath{\mathbb{C}}}
\newcommand{\T}{\ensuremath{\mathbb{T}}}

\def\calT{\mathcal{T}}
\def\calC{\mathcal{C}}
\def\calL{\mathcal{L}}

\def\calK{\mathcal{K}}
\def\calB{\mathcal{B}}
\def\calH{\mathcal{H}}

\def\calA{\mathcal{A}}
\def\calN{\mathcal{N}}
\def\calP{\mathcal{P}}
\def\calM{\mathcal{M}}

\def\calE{\mathcal{E}}
\def\calQ{\mathcal{Q}}

\def\calW{\mathcal{W}}

\def\bP{\mathbf{P}}

\newcommand{\ol}{\overline}

\theoremstyle{definition}

\DeclareMathOperator{\Dom}{Dom}

\DeclareMathOperator{\Index}{Index}

\DeclareMathOperator{\Ker}{Ker}

\DeclareMathOperator{\Tr}{Tr}
\DeclareMathOperator*{\res}{res}

\newcommand{\A}{\mathcal{A}}

\newcommand{\rst}[1]{\ensuremath{{\mathbin\upharpoonright}%
\raise-.5ex\hbox{$#1$}}}

\makeatletter

\newcommand{\Rmnum}[1]{\expandafter\@sl217--242owromancap\romannumeral #1@}
\makeatother

\author{C. Bourne}
\address{Department Mathematik, Friedrich-Alexander-Universit\"{a}t Erlangen-N\"{u}rnberg, 
Cauerstr. 11, 91058 Erlangen, Germany; \emph{and}
WPI-Advanced Institute for Materials Research (WPI-AIMR), Tohoku University,
2-1-1 Katahira, Aoba-ku, Sendai, 980-8577, Japan}
\email{chris.bourne@tohoku.ac.jp}

\author{A. Rennie}
\address{School of Mathematics and Applied Statistics, University of Wollongong, Wollongong, NSW 2522, Australia}
\email{renniea@uow.edu.au}

\date{\today}

\begin{document}

\begin{abstract}
In order to study continuous models of 
disordered topological phases,
we construct an unbounded 
Kasparov module and a semifinite 
spectral triple for the crossed product 
of a separable $C^*$-algebra 
by a twisted $\R^d$-action. 
The spectral triple allows us to employ
the non-unital local index formula to obtain
the higher Chern numbers in the 
continuous setting with complex observable algebra. 
In the case of the crossed product of a compact disorder space, 
the pairing can be extended to a larger algebra closely related to dynamical localisation,  
as in the 
tight-binding approximation. 
The Kasparov module allows us to exploit the Wiener--Hopf 
extension and the Kasparov product to obtain a bulk-boundary
correspondence for continuous models of disordered topological phases.
\end{abstract}

\title{Chern numbers, localisation and the bulk-edge correspondence for continuous
models of topological phases}
\maketitle

 

\tableofcontents


\section{Introduction}
This paper examines the noncommutative index theory of twisted crossed 
products of a separable $C^*$-algebra $B$ by 
$\R^d$. Our motivation for studying such algebras 
comes from its application to continuous models of 
disordered quantum systems, where the algebra of observables 
can be described by the
twisted crossed product $C(\Omega)\rtimes_\theta\R^d$
\cite{BelGapLabel, Bellissard94}.  Numerous results in condensed
matter physics which can be proved in the tight-binding approximation have
not been addressed for continuum models. Here we study 
higher Chern numbers, the bulk-edge
correspondence and stability of phases in the strongly disordered/dynamically localised 
regime for continuum models. Because of the anti-linear symmetries 
that appear in topological insulator systems, we will consider  
both complex and real $C^*$-algebras and crossed products. 

The key to our approach is the construction of a Kasparov module and 
a semifinite spectral triple modelling the geometry of
the noncommutative disordered Brillouin zone. The spectral triple satisfies the 
strongest summability conditions
of \cite{CGRS2}, allowing us to employ the local index formula for complex 
algebras. 

The local index formula yields the higher Chern numbers directly, 
in complete analogy with the formula for the higher Chern
numbers in the tight-binding approximation~\cite{PLB13, PSB14,PSBbook,PSBKK}.

In Section \ref{sec:localisation_general}
we extend the formulae for the higher Chern numbers to a larger 
Sobolev algebra that is constructed using the non-commutative $L^p$-spaces and 
closely related to regions of dynamical localisation.

Kellendonk and Richard~\cite{KR06} use the Wiener--Hopf extension
to model the relationship between bulk and edge observables,
\begin{equation} 
\label{eq:bulkedge_SES_intro}
  0 \to \calK \otimes (B\rtimes_\theta\R^{d-1}) 
  \to \big(C_0(\R\cup \{+\infty\})\otimes (B\rtimes_\theta\R^{d-1}) \big)\rtimes \R 
     \to (B\rtimes_\theta\R^{d-1})\rtimes\R\to 0.
\end{equation}
We prove, in Section \ref{sec:bulkedge}, that our Kasparov module for a twisted $\R^d$-action 
factorises (up to a sign) into the product of a Kasparov module for a twisted $\R^{d-1}$-action 
with the extension class from Equation \eqref{eq:bulkedge_SES_intro} linking the 
bulk and edge algebras.
This factorisation implies a bulk-edge correspondence for the semifinite index pairing as well 
as more general pairings of our Kasparov module with 
real or complex $K$-theory classes.

We return to our initial motivation in Section \ref{sec:top_phases_application} 
and include a
case-study of how our theory applies to disordered quantum systems and their 
topological properties. 
The example of disordered magnetic Schr\"odinger operators on 
$L^2(\R^d)$ also allows us to consider the connection of our Sobolev algebra to the 
localised states studied in~\cite{AENSS,GK13, GT13}. 
We compare our results and those in~\cite{AENSS}, 
where we show that if the Fermi energy lies in a region of dynamical localisation 
and the disorder space has an ergodic 
probability measure, then 
our $\Z$ or $\Z_2$-valued bulk indices are still well-defined. Furthermore, in the complex case, the 
Chern number formulas are constant throughout the mobility gap. 
We are also able to extend our 
results on the bulk-edge correspondence for strong complex topological phases 
and show that non-trivial 
bulk invariants imply delocalised edge states on the boundary, analogous 
to the discrete case in~\cite[Section 6.6]{PSBbook}.

Finally, Appendix \ref{Sec:Non_unital_prelims} gives a brief summary of non-unital 
index theory and the tools from Kasparov theory we require.


\section{Kasparov modules for twisted crossed products by $\R^d$} 
\label{sec:Crossed_prod_Kas_mod}
In this section we construct a Kasparov module for twisted crossed products
$B\rtimes_\theta\R^d$ where $B$ is a real or complex separable $C^*$-algebra; 
see Appendix \ref{Sec:Non_unital_prelims} for the definition of an 
unbounded Kasparov module.
This Kasparov module is closely related to the Connes--Thom class in Kasparov theory
when the crossed product is untwisted. The inverse class was 
studied in \cite{AnderssonThesis,Andersson14} for a different class of twisted crossed products.

\subsection{Preliminaries on twisted dynamical systems}
\label{Sec:twisted_systems_prelim}

Let $B$ be a  $C^*$-algebra with $(B,\R^d,\alpha,\theta)$ 
a twisted dynamical system~\cite{PR89}. We consider the $\ast$-algebra 
$C_c(\R^d, B)$ with operations,
\begin{align*}
 &(f_1\ast f_2)(x) = \int_{\R^d}\! \alpha_{-x}(\theta(y,x-y)) \alpha_{y-x}(f_1(y))f_2(x-y) 
   \,\mathrm{d}y,   &&f^*(x) = \alpha_{-x}(f(-x)^*).
\end{align*}
The unitary-valued function $\theta:\R^d\times\R^d \to \mathcal{U}(\calM(B))$ 
encodes the twist and takes values in the unitaries of the multiplier 
algebra of $B$. The twist $\theta$ is required to satisfy the cocycle identities
\begin{equation}
  \theta(x,y)\theta(x+y,z) =  \alpha_x(\theta(y,z))\theta(x,y+z),  
  \qquad \theta(x,0)=\theta(0,x)=1\quad \mbox{for all }x,\,y,\,z\in\R^d,
\label{eq:cocycle}
\end{equation}
and the following relationship with the action:
\begin{equation}
   \alpha_x \circ \alpha_y(b) = \theta(x,y)\, \alpha_{x+y}(b)\, \theta(x,y)^*, 
     \quad x,y\in\R^d,\,\, b\in B.
\label{eq:twist}
\end{equation}
We denote the crossed product completion $B \rtimes_{\theta}\R^d$ by $A$.

We do not consider crossed products with arbitrary twists $\theta$, 
but
restrict to the case that $\theta(x,-x)=1$ for all $x\in\R^d$. 
This simplifies many of 
our arguments and still encompasses the examples of interest (e.g. a 
disordered quantum system with continuously changing magnetic field). 

If for every $x,y \in \R^d$, $\theta(x,y)$ is constant in $B$ (e.g. $\theta$ comes 
from a magnetic field with constant strength), then the twist reduces to a 
map $\theta: \R^d\times \R^d \to \mathcal{U}(\mathbb{K})$ for $\mathbb{K}= \R$ or $\C$. 
Thus for complex algebras $\theta$ is a group cocycle $\R^d\times \R^d \to \T$ 
and therefore is related to the Moore cohomology group $H^2(\R^d,\T)$, which 
is constructed from Borel multipliers of $\R^d$. In the real case, we are interested 
in $H^2(\R^d, \{\pm 1\})$.
For complex group $2$-cocycles we have the following.
\begin{prop}[\cite{SimonMultipliers}, Lemma 8.3]
If $\theta:\R^d\times \R^d \to \T$  is a Borel multiplier and 
its class $[\theta]\in H^2(\R^d,\T)$  is non-torsion, 
then $\theta$ is cohomologous to $\wt{\theta}$ with $\wt{\theta}(x,-x) = 1$ for all $x\in\R^d$.
\end{prop}
\begin{proof}
By the cocycle property of $\theta$, we first note that
$$
  \theta(x,-x)\theta(0,x) = \theta(x,0)\theta(-x,x)
$$
so $\theta(x,-x) = \theta(-x,x)$. Next, provided $[\theta]$ is non-torsion, we can define
$\lambda(x) = [\theta(x,-x)]^{1/2}$ where we take the square root with argument in $[0,\pi)$.
Then we have that
$$
  \partial \lambda(x,-x) = \lambda(0) \lambda(x)^{-1} \lambda(-x)^{-1} = \theta(x,-x)^{-1}.
$$
Lastly, we define $\wt{\theta} = \theta \partial\lambda$ which by construction is such that 
$\wt{\theta}(x,-x) = 1$.
\end{proof}

For the case that $B=C(\Omega)$ for some  compact and second 
countable space $\Omega$ with twisted action, the assumption $\theta(x,-x)=1$ 
means that there is an explicit isomorphism
$$
   C(\Omega) \rtimes_\theta \R^d \cong \big( C(\Omega)\rtimes_\theta \R^{d-1} \big) \rtimes \R
$$
where the crossed product by $\R$ is untwisted,
see~\cite{KR06}. This decomposition allows us to relate the twisted crossed 
product $C(\Omega) \rtimes_\theta \R^d$ to the Wiener--Hopf extension
$$
  0 \to (C(\Omega)\rtimes_\theta \R^{d-1})\otimes\calK[L^2(\R)] \to C_0(\R\cup\{+\infty\}, C(\Omega)\rtimes_\theta \R^{d-1}) \rtimes \R
    \to C(\Omega) \rtimes_\theta \R^d \to 0.
$$ 
Such an extension plays a crucial role in the bulk-edge correspondence 
for disordered topological phases with a boundary  in
Section \ref{sec:bulkedge}.

For more general twisted actions, we first use~\cite[Theorem 4.1]{PR89} to 
decompose
$$
   B\rtimes_\theta \R^d \cong \big(B\rtimes_{\theta_{e}} \R^{d-1} \big) \rtimes_\sigma \R
$$
with $\theta_{e}$ the restriction of $\theta$ to $\R^{d-1}\times\{0\}$.
Then, letting $C = B\rtimes_{\theta_e} \R^{d-1}$ and using the Packer--Raeburn stabilisation trick~\cite[Section 3]{PR89}, 
$$
  K_\ast ( C\rtimes_\sigma \R) \cong K_\ast \big( (C\otimes \calK)\rtimes \R \big) \cong K_{\ast - 1}( C\otimes \calK) 
    \cong K_{\ast - 1}(C) \cong K_\ast( C\rtimes \R ).
$$
Therefore the Packer--Raeburn stabilisation isomorphism gives us an invertible element in $KK(C\rtimes_\sigma \R, C\rtimes \R)$ 
which allows us to relate the twisted crossed product $C\rtimes_\sigma \R$ to the Wiener--Hopf extension
$$
  0 \to C \otimes \calK \to (C_0(\R\cup\{+\infty\})\otimes C)\rtimes \R \to C\rtimes \R \to 0
$$
and corresponding class in $KK^1(C\rtimes \R, C)$.
Hence, from the perspective of Kasparov theory, we can assume that our twisted action $A=B\rtimes_\theta\R^d$
is such that $A \cong (B \rtimes_\theta \R^{d-1})\rtimes \R$ without losing any index-theoretic information. 
This unwinding of the crossed product will be important for 
boundary maps under the Wiener--Hopf extension and the
bulk-edge correspondence in Section \ref{sec:bulkedge}.

\begin{example}[Magnetic twists, \cite{BLM13, LMR10, MPR07}] 
\label{ex:mag_twist}
Let $B=C(\Omega)$ with $\Omega$ the compact space of 
disorder configurations with a
(twisted) action by $\R^d$ of magnetic translations. 
Consider a magnetic field in $\R^d$ 
with components $\{B^\omega_{jk}\}_{j,k=1}^d$ 
that continuously depend on $\omega\in\Omega$.
We then regard the cocycle $\theta$ as a function of $\omega$, where
$$
\theta(x,y)(\omega) = \exp\!\left( -i \Gamma^{B^\omega} \langle 0, x, x+y \rangle \right)
$$
with $\Gamma^{B^\omega} \langle 0, x, x+y \rangle$ the flux of the magnetic field through 
the triangle defined by the points $0$, $x$ and $x+y$. 
We see that in this case $\theta(x,-x)=1$ for all $\omega\in\Omega$ as required. 
The algebra $C(\Omega)\rtimes_{\theta}\R^d$ 
models continuous and disordered quantum systems with a (not necessarily constant) 
magnetic field. 

Let us extend this example slightly 
by considering the case when $B=C(\Omega)\rtimes_\phi \R^k$ for $1\leq k < d$.
Following~\cite[Theorem 4.1]{PR89}, 
there is a decomposition
$$
   C(\Omega)\rtimes_\theta \R^{d} \cong \left( C(\Omega)\rtimes_\phi \R^k\right) \rtimes_\sigma \R^{d-k},
$$
where, because the subgroup and quotient of $\R^{d}$ we consider is easy, the action and 
twist of $\R^k$ and $\R^{d-k}$ is simply the restriction of the action and twist of 
$\R^{d}$ to $\R^k\times\{0\}^d$ and $\{0\}^k\times \R^{d-k}$ respectively. 
Hence we retain that both $\theta(x,-x)(\omega)=1$ and $\sigma(z,-z)(\omega)=1$ 
for $x\in\R^{d}$ and $z\in\R^{d-k}$. Such a decomposition of twisted crossed products 
has applications to so-called weak topological insulators, where we may use this decomposition 
to extract $(d-k)$-dimensional invariants from $d$-dimensional systems. We will not 
emphasise this application here, though the interested reader can consult~\cite[Section 7, 8]{PSBKK} 
for results in the discrete setting.

We also remark that magnetic twists for real algebras and real crossed products are 
less interesting as we require $\theta(x,y)$ to be an orthogonal operator in 
$\calM(C(\Omega,\R))$. This puts large constraints on the type of magnetic 
field we can consider and will often mean that the magnetic field vanishes.
We will return to crossed products twisted by a magnetic 
field in Section \ref{sec:top_phases_application}.
\end{example}

\emph{We will now restrict to twisted dynamical 
systems $(B,\R^d,\alpha,\theta)$ such that $\theta(x,-x)=1$ for all $x\in\R^d$.}

\subsection{An unbounded Kasparov module}

We consider the Hilbert $C^*$-module $L^2(\R^d,B) \cong L^2(\R^d)\otimes B$ with right 
action by right-multiplication and inner-product
$$
   \left( f_1 \mid f_2 \right)_B = \int_{\R^d}\! f_1(x)^* f_2(x)\,\mathrm{d}x.
$$

\begin{lemma} 
\label{lemma:Crossed_prod_module_is_std_module}
If the twist $\theta$ is such that $\theta(x,-x)=1$ for all $x\in\R^d$, 
then the Hilbert module $L^2(\R^d,B)$ is isometrically isomorphic to 
the $C^*$-module $E_B$ given by the completion of $C_c(\R^d,B)$ with 
respect to the inner product $(f_1\mid f_2)_B = (f_1^{*}\ast f_2)(0)$.
\end{lemma}
\begin{proof}
The inner-product on $E_B$ takes the form
\begin{align*}
 (f_1\mid f_2)_B &= \int_{\R^d}\! \theta(y,-y)\alpha_{y}(f_1^*(y))f_2(-y)\,\mathrm{d}y 
   = \int_{\R^d} \theta(-y,y)f_1(y)^* f_2(y)\,\mathrm{d}y.
\end{align*}
If $\theta(y,-y)=1$ then the inner products coincide and the right-action of $B$ by 
right-multiplication is compatible with the inner product on $E_B$. 
Hence the two spaces are isomorphic as $C^*$-modules.
\end{proof}

\begin{prop} 
\label{prop:adjointable_repn_of_crossed_prod}
Let $C_c(\R^d,B)$ act on $E_B$ by left convolution multiplication. Then this action 
extends to an adjointable representation of 
$A=B\rtimes_\theta\R^d$.
\end{prop}
\begin{proof}
The action is adjointable on a dense subspace as 
$$
( f_1\ast f_2 \mid f_3)_B = (f_2^*\ast f_1^*\ast f_3)(0) = (f_2\mid f_1^*\ast f_3)_B, 
\quad f_1,f_2,f_3 \in C_c(\R^d,B)
$$
Furthermore, the action is  bounded since
$$
( f_1\ast f_2 \mid f_1\ast f_2)_B = (f_2^*\ast f_1^*\ast f_1\ast f_2)(0) \leq \Vert f_1^*\ast f_1\Vert (f_2\mid  f_2)_B, 
\quad f_1,f_2 \in C_c(\R^d,B),
$$
and so it extends to an adjointable 
action on the whole space by the completion $B\rtimes_\theta\R^d$.
\end{proof}

Using the identification of $E_B$ with $L^2(\R^d,B)$, we can define 
an adjointable action of $C_c(\R^d,B)$ (which extends to an action of
$B\rtimes_\theta\R^d)$ on $L^2(\R^d,B)$ by

\begin{align} 
\label{eq:Left_action_on_module}
  (\pi(f)\psi)(x) &= \int_{\R^d}\!\alpha_{-x}(\theta(y,x-y)) \alpha_{y-x}(f(y))\psi(x-y) \,\mathrm{d}y \nonumber \\
    &= \int_{\R^d}\! \alpha_{-x}(\theta(x-u,u)) \alpha_{-u}(f(x-u))\psi(u)\,\mathrm{d}u  \nonumber \\
    &= \int_{\R^d}\!\theta(-x,x-u) \alpha_{-u}(f(x-u))\psi(u) \,\mathrm{d}u,
\end{align}
where we have made the substitution $u=x-y$ and used the identity 
from Equation \eqref{eq:cocycle},
$$
 \alpha_{-x}(\theta(x-u,u))\theta(-x,x) = \theta(-x,x-u)\theta(-u,u),
$$
which together with the assumption $\theta(x,-x)=1$ implies that 
$\alpha_{-x}(\theta(x-u,u))=\theta(-x,x-u)$.

\begin{remark}
The two presentations of the right-$B$ $C^*$-module are useful in different 
contexts. The module $E_B$ allows us to easily define a left-action 
of $A$, 
while 
$\End^0_B(L^2(\R^d,B)) \cong \calK[L^2(\R^d)]\otimes B$, and so the 
presentation $L^2(\R^d,B)$ is useful for more analytic arguments.
\end{remark}

The algebra $C_c(\R^d,B)$ comes with the derivations $(\partial_j f)(x) = x_j f(x)$ 
(where $x_j$ is the $j$-th component of $x\in \R^d$) and we observe that 
$\partial_j(C_c(\R^d,B))\subset C_c(\R^d,B)$. A brief computation
relates the derivations $\{\partial_j\}_{j=1}^d$ to the unbounded position operators 
$\{X_j\}_{j=1}^d$ on $L^2(\R^d,B)$, where for $f\in C_c(\R^d,B)$,
\begin{equation}
\label{lemma:derivation_lines_up_with_commutator_with_dirac}
\pi(\partial_j f) = [X_j, \pi(f)].
\end{equation}

To construct the unbounded operator for our Kasparov module, we 
use the $\Z_2$-graded exterior algebra $\bigwedge^*\R^d$ and 
Clifford representations on this space. We first establish our 
notation and conventions for the Clifford algebras $C\ell_{p,q}$, 
namely 
$$   
 C\ell_{p,q} = \text{span}_\R \left\{\left.\gamma^1,\ldots,\gamma^p,\, \rho^1,\ldots,\rho^q\,\right\vert\,(\gamma^i)^2 
= 1,\,(\gamma^i)^* = \gamma^i,\,(\rho^i)^2=-1,
\,(\rho^i)^*=-\rho^i \right\},
$$
where $\mathrm{span}_\R$ means the algebraic span of the generators 
over the field $\R$,
and all the various $\gamma^j,\,\rho^k$ are odd and 
mutually anti-commute.
The exterior algebra $\bigwedge^*\R^d$ has
representations of $C\ell_{0,d}$ and $C\ell_{d,0}$ 
with generators
\begin{align*}
  &\rho^j(\omega) = e_j\wedge \omega - \iota(e_j)\omega, 
  &\gamma^j(\omega) = e_j \wedge \omega + \iota(e_j)\omega,
\end{align*}
where $\{e_j\}_{j=1}^d$ is the standard basis of $\R^d$ 
and $\iota(\nu)\omega$ is the contraction of $\omega$ along $\nu$. 
One readily checks that $\rho^j$ and $\gamma^j$ mutually anti-commute 
and generate representations of $C\ell_{0,d}$ and $C\ell_{d,0}$ 
respectively.

\begin{prop} 
\label{prop:crossed_prod_Kas_mod}
The triple
$$
 \lambda_d = \bigg( C_c(\R^d,B) \hat\otimes C\ell_{0,d}, \, L^2(\R^d,B)_B \hat\otimes 
   \bigwedge\nolimits^{\!*}\R^d, \, X= \sum_{j=1}^d X_j \hat\otimes \gamma^j \bigg)
$$
is a real or complex unbounded Kasparov module.
\end{prop}
\begin{proof}
The first thing to observe is that $X$ is self-adjoint and regular. This can be proved directly,
or by using the local-global principle  \cite{KL-LG,P-LG} 
and the fact that (up to Clifford variables) we
have a multiplication operator.
For $f\in C_c(\R^d,B)$, 
Equation \eqref{lemma:derivation_lines_up_with_commutator_with_dirac} says that
$$
  [X, \pi(f)] = \sum_{j=1}^d [X_j, \pi(f)]\hat\otimes \gamma^j 
    = \sum_{j=1}^d  \pi(\partial_j f) \hat\otimes \gamma^j,
$$
and $\pi(\partial_j f) \in \End_B(L^2(\R^d,B))$ for all $j\in\{1,\ldots,d\}$. 
For $\rho^k\in C\ell_{0,d}$ we know that $\rho^k\gamma^j=-\gamma^j\rho^k$
for $j\in\{1,\ldots,d\}$, so $C\ell_{0,d}$ graded commutes with $X$. Thus 
(graded) commutators of $X$ with elements of $C_c(\R^d,B)\hat\otimes C\ell_{0,d}$
are defined on $\Dom(X)$ and extend to adjointable operators.
Therefore all that we need to show is that $\pi(f)(1+X^2)^{-1/2}$ is compact 
in $L^2(\R^d,B)$ for $f\in C_c(\R^d,B)$. 
Using Equation \eqref{eq:Left_action_on_module}, we note that 
$\pi(f)(1+X^2)^{-1/2}$ has the $B$-valued 
integral kernel
\begin{equation} 
  k_f(x,y) = \theta(-x,x-y) \alpha_{-y}(f(x-y))(1+|y|^2)^{-1/2} \,
    \hat\otimes \,{\rm Id}_{\bigwedge^*\R^d}.
  \label{eq:kernel}
\end{equation}
The continuity of $f$ and $\theta$ shows that $k_f\in C_0(\R^d\times\R^d)\otimes B$.
This now allows us to find a sequence in $C_c(\R^d)\otimes C_c(\R^d)\otimes B^2$
such that
$$
k_f=\lim_{n\to\infty}\sum_{j=1}^{N_n}f_{n,j}\otimes g_{n,j}\otimes b_{n,j}c_{n,j}.
$$
Then computing shows that the sum of rank one operators (see the Appendix)
$$
\pi(f)(1+X^2)^{-1/2}
=\lim_{n\to\infty}\sum_{j=1}^{N_n}
\Theta_{f_{n,j}\otimes b_{n,j}, \ol{g_{n,j}}\otimes c_{n,j}^*}
$$
converges in the operator norm topology of operators on $L^2(\R^d,B)$. Hence
$\pi(f)(1+X^2)^{-1/2}$ is the norm limit of compact operators, and 
thus is compact.
\end{proof}

We have used the orientation of $\R^d$ to construct the Kasparov module using the
operator $X= \sum_{j=1}^d X_j\hat\otimes \gamma^j$ and left-Clifford 
multiplication on $\bigwedge^*\R^d$, \cite[Section 4]{LRV}. 
The exterior algebra construction has the benefit that 
the differences 
between the real and complex cases are minimal, there is no dependence on spin or spin$^c$ structure,
 and the Kasparov modules we 
construct behave well under Kasparov products (see Section \ref{sec:bulkedge}).


\section{Traces and a semifinite spectral triple} 
\label{sec:semifinite_spec_trip}
If the algebra $B$ has a faithful, semifinite and norm lower-semicontinous tracial weight, 
$\tau_B$, that is invariant under the twisted $\R^d$-action, there is a general 
method by which we can obtain a semifinite spectral 
triple, \cite{LN04,PR06,KNR,CGPRS}. Again, a summary of the relevant definitions
and results is contained in the Appendix.

The existence of such a trace on $B$ is satisfied in the physically 
interesting case of $B=C(\Omega,M_N(\C))$ (or $M_N(\R)$), 
where the disorder space of configurations $\Omega$ (typically compact)
is equipped with a probability measure $\bP$ such that $\mathrm{supp}(\bP)=\Omega$. 
In examples from aperiodic media, the measure $\bP$ is often 
invariant and ergodic under the $\R^d$-action by translations,
though many of our results only require that
$\tau_B$ is invariant under the group action.

In our examples, the semifinite spectral triple we obtain is also smoothly summable 
in the sense of Definition \ref{def:smoothly_summable}, 
which allows us to employ the local index formula, 
Theorem \ref{thm:local_index_formula_odd} and 
\ref{thm:local_index_even}~\cite[Theorem 3.33]{CGRS2}.
In turn, the local index formula gives us the 
higher Chern numbers and an approach to understanding localisation.

We let $\Dom(\tau_B)$ be the domain 
of the trace $\tau_B$ and write $\Dom(\tau_B)^{1/2}$ as the set of operators 
$b\in B$ such that $\tau_B(b^*b) < \infty$.
Given the $C^*$-module $L^2(\R^d,B)$ and trace $\tau_B$, we complete 
$C_c(\R^d)\otimes \Dom(\tau_B)$ in the norm coming from the
inner-product
\begin{equation}
  \langle \lambda_1\otimes b_1,\lambda_2 \otimes b_2 \rangle = 
    \tau_B\!\left( (\lambda_1\otimes b_1 \mid \lambda_2 \otimes b_2)_B \right) 
    = \langle \lambda_1, \lambda_2 \rangle_{L^2(\R^d)} \, \tau_B(b_1^*b_2),
\label{eq:inn-prod}
\end{equation}
which defines the Hilbert space $L^2(\R^d)\otimes L^2(B, \tau_B)$ where
$L^2(B, \tau_B)$ is the GNS space.

\begin{lemma}
\label{lem:rep}
The algebra $A=B\rtimes_\theta \R^d$ acts on $L^2(\R^d)\otimes L^2(B, \tau_B)$.
\end{lemma}
\begin{proof}
This follows from the identification $L^2(\R^d)\otimes L^2(B,\tau_B) \cong 
L^2(\R^d,B)\otimes_B L^2(B,\tau_B)$ and Proposition \ref{prop:adjointable_repn_of_crossed_prod}.
\end{proof}

\begin{prop}[\cite{LN04}, Theorem 1.1]
Given $T\in\End_{B}(L^2(\R^d,B))$ with $T\geq 0$, define
$$  
\Tr_\tau(T) = \sup_{I} \sum_{\xi\in I}\tau_B\!\left[( \xi\mid T\xi)_{B}\right], 
$$
where the supremum is taken over all finite subsets $I\subset L^2(\R^d,B)$ with
$\sum_{\xi\in I}\Theta_{\xi,\xi}\leq 1$. \\
1) Then $\Tr_\tau$ is a semifinite norm lower-semicontinuous 
trace on the compact endomorphisms $\End_{B}^0(L^2(\R^d,B))$
with the property $\Tr_{\tau}(\Theta_{\xi_1,\xi_2}) = \tau_B\big((\xi_2\mid \xi_1)_{B}\big)$.\\
2) Let $\calN$ be the von Neumann algebra $\End_{B}^{00}(L^2(\R^d,B))''\subset 
\calB(L^2(\R^d)\otimes L^2(B,\tau_B))$. 
Then the trace $\Tr_\tau$ extends to a faithful semifinite trace on the positive cone $\calN_+$.
\end{prop}

The semifinite trace $\Tr_\tau$ on $\calN$ gives a semifinite 
trace $\Tr_\tau \hat\otimes \Tr_{\bigwedge^*\R^d}$ on 
$\calN \hat\otimes \End(\bigwedge^*\R^d)$. To simplify our notation, we will 
often suppress the finite-trace and finite-dimensional von Neumann 
algebra $\End(\bigwedge^*\R^d)$.

\begin{lemma} 
\label{lemma:Hilbert_schmidt_estimate}
If $f\in C_c(\R^d,\Dom(\tau_B)^{1/2})$, then 
$\pi(f)(1+X^2)^{-s/4}$ is Hilbert-Schmidt with respect to $\Tr_\tau$ for 
$s>d$.
\end{lemma}
\begin{proof}
The operator $\pi(f)(1+X^2)^{-s/4}$ has the integral kernel
$$
  k_f(x,y) = \theta(-x,x-y) \alpha_{-y}(f(x-y))(1+|y|^2)^{-s/4} \,\hat\otimes\, {\rm Id}_{\bigwedge^*\R^d}.
$$
Ignoring the factor ${\rm Id}_{\bigwedge^*\R^d}$,
the kernel of $(\pi(f)(1+X^2)^{-s/4})^*=(1+X^2)^{-s/4}\pi(f^*)$ is then 
\begin{align*}
  \wt{k}_{f^*}(x,y) &=(1+|x|^2)^{-s/4} \theta(-x,x-y) \alpha_{-y}(f^*(x-y)) \\
    &= (1+|x|^2)^{-s/4} \theta(-x,x-y) \alpha_{-y}\circ \alpha_{y-x}(f(y-x)^*) \\
    &= (1+|x|^2)^{-s/4} \theta(-x,x-y) \theta(-y,y-x) \alpha_{-x}(f(y-x)^*) \theta(-y,y-x)^* \\
    &= (1+|x|^2)^{-s/4}  \alpha_{-x}(f(y-x)^*) \theta(-y,y-x)^*,
\end{align*}
where we used the definition of $f^*$, Equation \eqref{eq:twist} on the twisting 
of $\alpha$, and the cocycle identity
$$
  \theta(-x,x-y)\theta(-y,y-x) = \alpha_{-x}(\theta(x-y,y-x)) \theta(-x,0) \\
    = \alpha_{-x}(1) 1 = 1
$$
under the added assumption $\theta(u,-u)=1$.

Because $\tau_B$ is a faithful, semifinite and 
norm lower-semicontinuous tracial weight on $B$, the trace-class 
operators $\calL^1(\calN, \Tr_\tau)$ 
contains  $\calL^1(L^2(\R^d)) \otimes \Dom(\tau_B)$ (algebraic tensor product),
and the trace restricted to this set is $\Tr_{L^2(\R^d)}\otimes \tau_B$. 
Ignoring the trace over $\bigwedge^*\R^d$, we compute directly
\begin{align*}
  \Tr_\tau\Big( (1+X^2)^{-s/4} \pi(f^* f) (1+X^2)^{-s/4} \Big) &= 
     \int_{\R^{2d}}\! \tau_B\big( \wt{k}_{f^*}(x,y) k_f(y,x) \big) \mathrm{d}x\,\mathrm{d}y \\
   &\hspace{-4.5cm}= \int_{\R^{2d}}\! \tau_B\Big( (1+|x|^2)^{-s/4} \alpha_{-x}(f(y-x)^*) \theta(-y,y-x)^*  \\
   &\times \theta(-y,y-x)\alpha_{-x}(f(y-x)) (1+|x|^2)^{-s/4} \Big) \,\mathrm{d}x\,\mathrm{d}y \\
   &\hspace{-4.5cm}= \int_{\R^{2d}}\! \tau_B\Big(\alpha_{-x}(f(y-x)^*) \alpha_{-x}(f(y-x))  \Big)
       (1+|x|^2)^{-s/2} \,\mathrm{d}x\,\mathrm{d}y \\
   &\hspace{-4.5cm}= \int_{\R^{2d}}\! \tau_B\Big(f(y-x)^*f(y-x)\Big) 
         (1+|x|^2)^{-s/2}\,\mathrm{d}x\,\mathrm{d}y ,
\end{align*}
where we have used the invariance of $\tau_B$ under 
the $\R^d$-action. Next we make the substitution $ u=y-x$, $v=x$ and use the 
compact support of $f$ on $u$ to estimate, for $s>d$,
\begin{align*}  
  \left\| \pi(f)(1+X^2)^{-s/4}\right\|_2^2
  &=  \int_{\R^{2d}}\!\! \tau_B\big( f(u)^*f(u) \big)  
          (1+|v|^2)^{-s/2} \,\mathrm{d}u\,\mathrm{d}v 
   \hspace{-0cm}= C_s \int_{\R^{d}}\! \tau_B\big( f(u)^*f(u) \big) 
    \,\mathrm{d}u<\infty.
\end{align*}
The trace over $\bigwedge^*\R^d$ does not change the argument, only adding a factor of $2^d$, and so we are done.
\end{proof}

In the language of semifinite spectral triples (summarised in the Appendix),
the Lemma says that $C_c(\R^d,\Dom(\tau_B)^{1/2})$ is contained in $\calB_2(X,d)$,
the `square integrable' operators. In fact $C_c(\R^d,\Dom(\tau_B)^{1/2})$ is contained in 
$\calB_2^\infty(X,d)$, the `smooth square integrable' operators.

\begin{lemma}
\label{lem:square-smo}
For $f\in C_c(\R^d,B)$,
let $\delta(\pi(f))=[|X|,\pi(f)]$, defined initially on $\Dom(X)$. Then for all $m=1,2,3,\dots$, 
and all $f\in C_c(\R^d,\Dom(\tau_B)^{1/2})$, the operator
$\delta^m(\pi(f))(1+X^2)^{-s/4}$ is Hilbert-Schmidt with respect to $\Tr_\tau$.
\end{lemma}
\begin{proof}
The proof is much like that of the previous lemma. 
We just note that the operator $\delta^m(\pi(f))(1+X^2)^{-s/4}$
has $B$-valued integral kernel
$$
  k_{f,m}(x,y) = \theta(-x,x-y) (|x|-|y|)^m\alpha_{-y}(f(x-y))(1+|y|^2)^{-s/4} 
  \,\hat\otimes\, {\rm Id}_{\bigwedge^*\R^d}.
$$
Then just as in Lemma \ref{lemma:Hilbert_schmidt_estimate}, 
the kernel of 
$(\delta^m(\pi(f))(1+X^2)^{-s/4})^*=(1+X^2)^{-s/4}\delta^m(\pi(f^*))(-1)^m$ is then 
$$
  \wt{k}_{f^*,m}(x,y) 
  =(1+|x|^2)^{-s/4} (|x|-|y|)^m(-1)^m \alpha_{-x}(f(y-x)^*)\theta(-y,y-x)^*
  \hat\otimes\, {\rm Id}_{\bigwedge^*\R^d}.
$$
Then we compute as before,
\begin{align*}
 & (-1)^m\Tr_\tau\Big( (1+X^2)^{-s/4} \delta^m(\pi(f^*))\delta^m(\pi(f)) (1+X^2)^{-s/4} \Big)\\
   &\hspace{2cm}= \int_{\R^{2d}}\! \tau_B\Big(f(y-x)^*f(y-x)\Big) (|x|-|y|)^{2m}(-1)^m
         (1+|x|^2)^{-s/2}\,\mathrm{d}x\,\mathrm{dy}.
\end{align*}
Now taking absolute values, using $(|x|-|y|)^{2m}\leq |x-y|^{2m}$, and changing variables as in the last lemma we find that for $s>d$,
\begin{align*}  
  \left\| \delta^m(\pi(f))(1+X^2)^{-s/4}\right\|_2^2
  &\leq  \int_{\R^{2d}}\!\! \tau_B\big( f(u)^*f(u) \big) |u|^{2m} 
          (1+|v|^2)^{-s/2} \,\mathrm{d}u\,\mathrm{d}v \\
&   \hspace{-0cm}= C_s \int_{\R^{d}}\! \tau_B\big( f(u)^*f(u) \big) |u|^{2m}
    \,\mathrm{d}u<\infty.\qedhere
\end{align*}
\end{proof}

The next theorem is the main result of this section. The analogous 
result in the tight-binding approximation can be proved much more simply.
While the proof here is quite short, it relies on 
quite substantial machinery, which we summarise in the Appendix.
The result justifies the use of this extra machinery, because 
once we have shown that
our spectral triple satisfies the additional
requirement of smooth summability, we 
can employ the local index formula, at least in the case of
complex $C^*$-algebras. Ultimately
the local index formula yields the higher Chern numbers
and the extension of the index pairing 
to the localised regime.

\begin{thm} 
\label{prop:semifinite_spec_trip}
Let 
$\calA=C_c(\R^d,\Dom(\tau_B))$. Then
$$
\bigg( \calA \hat\otimes C\ell_{0,d}, \, L^2(\R^d)\otimes L^2(B,\tau_B) \hat\otimes 
   \bigwedge\nolimits^{\!*}\R^d, \, X=\sum_{j=1}^d X_j\otimes 1 \hat\otimes \gamma^j, \, 
   (\calN, \Tr_\tau) \bigg)
$$
is a ($\Z_2$-graded) 
smoothly summable semifinite spectral triple with spectral dimension 
$d$.
\end{thm}
\begin{proof}
The boundedness of commutators $[X,\pi(f)]$ 
is the same as in the Kasparov module
case and the self-adjointness of $X$ is clear.
By Lemma \ref{lemma:Hilbert_schmidt_estimate}, 
$\pi(f)(1+X^2)^{-s/4}$ is $\Tr_\tau$-Hilbert-Schmidt for $s>d$ 
and therefore compact in $(\calN,\Tr_\tau)$~\cite{FK}. 
As $s\to 2$, $\pi(f)(1+X^2)^{-s/4} \to \pi(f)(1+X^2)^{-1/2}$ 
in operator norm, whence $\pi(f)(1+X^2)^{-1/2}$ is a norm-limit of 
compact operators and so is compact. In all these statements, and below, we
write $\pi$ instead of $\pi\,\hat\otimes\, 1_{\bigwedge^*\R^d}$.

Using the notation from Section \ref{subsec:cgrs2_preliminaries}, our 
spectral triple will be smoothly summable if we can 
show that $\pi(\calA)\cup[X,\pi(\calA)]\subset \calB_1^\infty(X,d)$. 
Since $\delta$ is a derivation, for any $m$ 
and any $f,\,g\in C_c(\R^d,\Dom(\tau_B)^{1/2})$
$$
\delta^m(\pi(f)\pi(g))=\sum_{k=0}^m \binom{m}{k}\delta^k(\pi(f))\delta^{m-k}(\pi(g))
$$
and this is an element of $\calB_1(X,d)$ by Lemma \ref{lem:square-smo}.
Hence 
$$
C_c(\R^d,\Dom(\tau_B)^{1/2})^2=C_c(\R^d,\Dom(\tau_B))\subset \calB_1^\infty(X,d).
$$
Proposition \ref{prop:finite_summ_condition} then implies that the spectral triple is 
finitely summable with spectral dimension $d$.

Next we consider $\delta^m([X,\,\pi(fg)])$ and note that
$$  
[X,\,\pi(fg)] = \sum_{j=1}^d [X_j, \,\pi(fg)]\,\hat\otimes\, \gamma^j = \sum_{j=1}^d \partial_j(fg) \,\hat\otimes\, \gamma^j 
$$
by Equation \eqref{lemma:derivation_lines_up_with_commutator_with_dirac}. 
Because $\partial_j (fg)\in \calA$ and $|X|$ commutes with the $\gamma^j$,
$[X,\,\pi(\calA)] \subset \calB_1^\infty(X,d)$ by the same 
argument as $\pi(\calA)$ and we are done.
\end{proof}

Because we have a smoothly-summable spectral triple, we may complete $\calA$ in the $\delta$-$\varphi$ 
topology (see Equation \eqref{eq:delta_phi_topology_seminorms} in the appendix) to obtain an algebra $\calA_{\delta,\varphi}$ 
that is Fr\'{e}chet and stable under the holomorphic 
functional calculus~\cite[Proposition 2.20]{CGRS2}, 
so that 
$K_*(\calA_{\delta,\varphi})\cong K_*(A)$. Thus any $K$-theory class for
$B \rtimes_\theta \R^d$ has a representative in a matrix algebra over $\calA_{\delta,\varphi}$. In addition, the spectral triple
$$
  \left( \calA_{\delta,\varphi} \hat\otimes C\ell_{0,d}, \, L^2(\R^d)\otimes L^2(B,\tau_B) \hat\otimes \bigwedge\nolimits^{\! *}\R^d, \, 
  X, \, (\calN, \Tr_\tau ) \right)
$$
is smoothly summable with spectral dimension $d$~\cite[Proposition 2.20]{CGRS2},
and so our analytic formulae extend to pairings with projections or unitaries over
$\calA_{\delta,\varphi}$.

\section{Continuous Chern numbers for complex systems} \label{sec:Chern_nos}
Now that we have a semifinite spectral triple satisfying the regularity 
properties required for the local index formula, we restrict to 
complex algebras and Hilbert spaces to consider the semifinite index pairing 
with $K$-theory classes in $K_\ast(B\rtimes_\theta\R^d)$. The 
limitations of this approach
for real algebras will be discussed below.

Our main aim is to obtain the higher Chern numbers of continuum systems. Various tight-binding 
versions of these results were obtained in \cite{PLB13,PSB14,PSBbook,PSBKK}.

To better align our notation with the other literature 
on the topic, we consider the unbounded trace $\calT$ 
on $B\rtimes_\theta \R^d$ by the formula $\calT(f) = \tau_B(f(0))$ 
for $f\in C_c(\R^d,\Dom(\tau_B))$. We note that $\calT(f) = \Tr_\tau(f)$ for 
$f\in C_c(\R^d,\Dom(\tau_B))$ by an argument analogous to the proof of
Lemma \ref{lemma:Crossed_prod_module_is_std_module}.

The first observation we make is that
the semifinite local index formula is currently only valid
for ungraded and complex algebras (acting on possibly graded spaces),\footnote{The proofs of the 
local index formula given in
\cite{CGRS2,CPRS1} can naturally be recast for graded algebras, but the validity
of the result needs to be checked. For real (graded) algebras this will be necessary.}
while our semifinite spectral triple is defined over 
a graded algebra $\calA\hat\otimes C\ell_{0,d}$.

For complex algebras we can 
work with the semifinite spectral triple coming from the 
spin${}^c$ structure on $\R^d$. This is also what is used in \cite{PLB13,PSB14,PSBbook,PSBKK}.
Namely, we let $\nu = 2^{\lceil (d-1)/2 \rceil}$. Then the triple
\begin{equation}  \label{eq:spin_semfinite_spec_trip}
\bigg( \calA=C_c(\R^d,\Dom(\tau_B)), \, L^2(\R^d)\otimes L^2(B,\tau_B) \hat\otimes \C^\nu, 
    \, X=\sum_{j=1}^d X_j\otimes 1 \hat\otimes \gamma^j  \bigg)
\end{equation}
is a complex and smoothly summable semifinite spectral triple of 
spectral dimension $d$ and
relative to $(\calN \hat\otimes\End(\C^\nu), \Tr_\tau\hat\otimes \Tr_{\C^\nu})$. 
The spectral triple is odd (ungraded) if $d$ is odd and is even for $d$ even 
with grading operator $\gamma = (-i)^{d/2}\gamma^1\cdots\gamma^d$.

For $d$ even, the semifinite spectral triple from 
Equation \eqref{eq:spin_semfinite_spec_trip} is easily 
related to our original semifinite spectral triple from Theorem 
\ref{prop:semifinite_spec_trip} by the external product with the Morita equivalence 
bimodule $\big( \C\ell_d, \C^{2^{d/2}}_\C, 0 \big)$, which gives an invertible class 
in $KK(\C\ell_d, \C)$. For $d$ odd, we first turn our ungraded triple into a graded 
triple over $\calA \hat\otimes \C\ell_1$, then the Morita equivalence between 
$\C\ell_{d-1}$ and $\C$ recovers the original spectral triple. In both even and odd cases we
do not lose any information.

\begin{prop} 
\label{prop:residue_trace_formula}
Let $f\in C_c(\R^d,\Dom(\tau_B))$. If $\tau_B$ is invariant under the action 
of $\R^d$, then the complex function
$$
  s\mapsto \zeta_{f}(s) = \Tr_\tau(\pi(f)(1+|X|^2)^{-s/2})
$$
is holomorphic for $\Re(s)>d$ with at worst a simple pole at $s=d$ with residue
$$
 \res_{s=d} \Tr_\tau(\pi(f)(1+|X|^2)^{-s/2}) = \mathrm{Vol}_{d-1}(S^{d-1})\, \calT(f).
$$
\end{prop}
\begin{proof}
Because $\pi(f)(1+|X|^2)^{-s/2}$ is trace-class for $\Re(s)>d$, we can compute directly using 
that $\Tr_\tau = \Tr_{L^2(\R^d)}\otimes \tau_B$ (on the algebraic tensor product of 
$\mathcal{L}^1(L^2(\R^d))$ and $\Dom(\tau)\subset B$).
Using the formula in Equation \eqref{eq:kernel} for the integral kernel, we
find that for $\Re(s)>d$,
\begin{align*}
  \Tr_\tau\big(\pi(f)(1+|X|^2)^{-s/2}\big) &=  \int_{\R^d}\! \tau_B\big( k_{f}(x,x)\big) \mathrm{d}x \\
   &= \int_{\R^d}\! \tau_B\Big(\theta(-x,0) \alpha_{-x}(f(0)) \Big) 
        (1+|x|^2)^{-s/2}\,\mathrm{d}x  \\
 &  =  \int_{\R^d}\! \tau_B\Big( \alpha_{-x}(f(0)) \Big) (1+|x|^2)^{-s/2}\,\mathrm{d}x \\
   &= \tau_B(f(0)) \int_{\R^d}\!(1+|x|^2)^{-s/2}\,\mathrm{d}x,
\end{align*}
where we have used the invariance of the $\R^d$-action on the fourth line. 
Using polar coordinates we can compute explicitly for $\Re(s)>d$, 
\begin{align} \label{eq:residuetrace_comp1}
   \Tr_\tau\big(\pi(f)(1+|X|^2)^{-s/2}\big) 
   &=  \tau_B(f(0))\, \mathrm{Vol}_{d-1}(S^{d-1}) \int_0^\infty (1+r^2)^{-s/2} r^{d-1}\,\mathrm{d}r \nonumber \\
  &= \calT(f)\,\mathrm{Vol}_{d-1}(S^{d-1})  \frac{\Gamma\!\left(\frac{d}{2}\right) \Gamma\!\left(\frac{s-d}{2}\right)}{2\Gamma\!\left(\frac{s}{2}\right)}.
\end{align}
The right hand 
side of Equation \eqref{eq:residuetrace_comp1} has an analytic 
continuation to the complex plane that is holomorphic for $\Re(s)>d$ and 
with a simple pole at $\Re(s)=d$. 
Taking the residue yields
$$
 \res_{s=d} \Tr_\tau(\pi(f)(1+|X|^2)^{-s/2}) = \calT(f)\,\mathrm{Vol}_{d-1}(S^{d-1})
$$
as required.
\end{proof}

In the case of complex algebras and Kasparov modules, 
the semifinite spectral triple from Equation \eqref{eq:spin_semfinite_spec_trip} 
and tracial weight $\tau_B$ give a well-defined map
$$
  K_\ast(B\rtimes_\theta \R^d) \times KK^\ast(B\rtimes_\theta\R^d, B) \to KK(\C,B) \xrightarrow{(\tau_B)_*} \R.
$$
The semifinite local index 
formula~\cite[Theorem 3.33]{CGRS2} gives us computable expressions for this $K$-theoretic 
composition, which we now present.

\subsection{Odd formula}

Because the spectral triple of Equation \eqref{eq:spin_semfinite_spec_trip} 
 is smoothly summable with spectral dimension $d$, 
 the odd local index formula gives that 
$$ 
\langle [u],[(\calA,\calH,X_\mathrm{odd})]\rangle = \frac{-1}{\sqrt{2\pi i}} 
\res_{r=(1-d)/2} \sum_{m=1,\text{odd}}^{d} \! \phi_m^r(\mathrm{Ch}^m(u)), 
$$
where $u$ is a unitary in $M_q(\calA^\sim)$ and 
$$  
\mathrm{Ch}^{2n+1}(u) = (-1)^n n!\,u^*\otimes u \otimes u^*\otimes \cdots \otimes u \hspace{0.5cm}(2n+2\text{ entries}). 
$$
The functional $\phi_m^r$ is the resolvent cocycle from Definition \ref{def:resolvent_cocycle_defn}. 
Using notation from Section \ref{subsec:Index_pairing_defn}, we can write the 
pairing $\langle [u],[(\calA,\calH,X_\mathrm{odd})]\rangle$ as a semifinite 
Fredholm index, 
$$
\langle [u],[(\calA,\calH,X_\mathrm{odd})]\rangle = \Index_{{\tau}\otimes \Tr_{\C^{2q}}}\!\left( (P \otimes 1_{2q})\hat{u}(P \otimes 1_{2q}) \right), 
\qquad \hat{u} = \begin{pmatrix} u & 0 \\ 0 & 1_u \end{pmatrix}, 
$$
with $1_u = \pi^n(u)$ for $\pi^n:M_q(\calA)\to M_q(\C)$ the quotient map from the unitisation and 
$P = \frac{1}{2}(1+F_X)$ given as in Proposition \ref{prop:semifinite_index_is_kas_prod}. We also 
write $\Index_\tau$ to refer to the semifinite Fredholm index with respect to $\Tr_\tau$.

\begin{thm}[Odd index formula] 
\label{thm:higher_dim_chern_number_odd}
Let $d$ be odd and $u$ a complex unitary in $M_q(\calA^\sim)$
where $\calA^\sim$ is the minimal unitisation of $\calA$. 
If the trace $\tau_B$ on $B$ is invariant under the action of $\R^d$, then
the semifinite index pairing with the semifinite spectral triple 
from Equation \eqref{eq:spin_semfinite_spec_trip} with $d$ odd is given by the formula
$$  
\Index_{\tau\otimes \Tr_{\C^{2q}}}\!\left( (P_X \otimes 1_{2q})\hat{u}(P_X \otimes 1_{2q}) \right) = C_d \sum_{\sigma\in S_d}(-1)^\sigma\, 
 (\Tr_{\C^q}\otimes\calT) \bigg(\prod_{i=1}^d u^* \partial_{\sigma(i)}u \bigg), 
 $$
where $C_{2n+1} = \frac{ 2(2\pi i)^n n!}{(2n+1)!}$, 
$\Tr_{\C^q}$ is the matrix trace on $\C^q$ and $S_d$ is the permutation group on $d$ letters.
\end{thm}

We give the proof in the case $q=1$, since we  can extend to matrices by the 
standard extension of spectral triples over $\calA$ to $M_q(\calA)$. 
Except in cases where we need specific results about the spinor trace 
$\Tr_{\C^\nu}$, we will write the trace $\Tr_\tau \hat\otimes \Tr_{\C^\nu}$ 
as just $\Tr_\tau$.

 To compute the index pairing we make the following important observation.

\begin{lemma}[\cite{BCPRSW}, {\S}11.1] 
\label{prop:only_need_top_term_in_local_index_odd_case}
The only term in the sum $\sum\limits_{m=1,\mathrm{odd}}^{d}\! \phi_m^r(\mathrm{Ch}^m(u))$ 
that contributes to the index pairing is the term with $m=d$.
\end{lemma}
\begin{proof}
We first note that the spinor trace of the product of $d=2n+1$ Clifford generators is given by
\begin{equation} \label{eq:trace_of_grading}
 \Tr_{\C^\nu} (\gamma^1\cdots \gamma^d) = (-1)^{n+1}i^{-n} 2^{ n }
\end{equation} 
and will vanish on any product of $k$ Clifford generators with $0<k<d$. The resolvent cocycle involves the spinor trace of terms
$$  
a_0R_s(\lambda)[X,a_1]R_s(\lambda)\cdots [X,a_m]R_s(\lambda), \qquad R_s(\lambda) = (\lambda-(1+s^2+X^2))^{-1}, 
$$
for $a_0,\ldots,a_m\in\pi(\calA)$. We note that $[X,a_l] = \sum_{j=1}^d \partial_j(a_l)\hat\otimes\, \gamma^j$ 
and $R_s(\lambda)$ is diagonal 
in the spinor representation. Hence the product 
$a_0R_s(\lambda)[X,a_1]R_s(\lambda)\cdots [X,a_m]R_s(\lambda)$ 
will be in the span of $m$ Clifford generators for $0<m<d$ acting on 
$L^2(\R^d)\otimes L^2(B,\tau_B)\hat\otimes\C^\nu$. 
Furthermore, our trace estimates ensure that each spinor component
$$  
\int_\ell \lambda^{-d/2-r} a_0(\lambda-(1+s^2+X^2))^{-1}\partial_{j_1}(a_1) \cdots \partial_{j_m}(a_m)(\lambda-(1+s^2+X^2))^{-1}\,\mathrm{d}\lambda 
$$
is trace-class for $a_0,\ldots,a_m\in \calA$ and $\Re(r)$ sufficiently large. 
Hence for $0<m<d$, the spinor trace will vanish for $\Re(r)>0$ and so 
$\phi_m^r(\mathrm{Ch}^m(u))$ does 
not contribute to the index pairing for $0<m<d$.
\end{proof}

\begin{proof}[Proof of Theorem \ref{thm:higher_dim_chern_number_odd}]
Lemma \ref{prop:only_need_top_term_in_local_index_odd_case} 
simplifies the index computation 
substantially, where
the index is now given by the expression
\begin{align*}
  \langle [u],[(\calA,\calH,X_\mathrm{odd})]\rangle 
  &=  \frac{-1}{\sqrt{2\pi i}} \res_{r=(1-d)/2} \phi_d^r(\mathrm{Ch}^d(u)).
\end{align*}
Here the resolvent cocycle is formed using the trace 
$\Tr_{\C^q}\otimes \Tr_\tau\otimes\Tr_{\C^\nu}$
where $\Tr_{\C^\nu}$ is the trace on the spinor representation. 
We simplify the formulae below by taking $q=1$ and suppressing the spinor trace.
Therefore we need to compute the residue at $r=(1-d)/2$ of
$$  
\frac{(-1)^{n+1}n!\eta_d}{(2\pi i)^{3/2}} \int_0^\infty\! s^d\Tr_\tau\!\left( \int_\ell \lambda^{-d/2-r} u^* R_s(\lambda) [X,u] R_s(\lambda)[X,u^*]\cdots [X,u]R_s(\lambda) \,\mathrm{d}\lambda\right)\!\mathrm{d}s,  
$$
where $d=2n+1$ and 
$$
\eta_d= -\sqrt{2i}\, 2^{d+1} \frac{\Gamma(d/2+1)}{\Gamma(d+1)}.
$$ 
To compute this residue we move all terms $R_s(\lambda)$ to the right, 
which can be done up to a function holomorphic at $r=(1-d)/2$ by an argument similar 
to the proof of Proposition \ref{prop:residue_trace_formula}. 
This allows us to take the Cauchy integral. We then observe that 
$\underbrace{[X,u][X,u^*]\cdots[X,u]}_{d\text{ terms}} \in \pi(\calA) \,\hat\otimes\, 1_{\C^\nu}$, 
so Proposition \ref{prop:residue_trace_formula} implies that 
the zeta function
$$ 
\Tr_\tau\!\left(u^*{[X,u][X,u^*]\cdots[X,u]}(1+X^2)^{-z/2}\right) 
$$
has at worst a simple pole at $\Re(z)=d$. Therefore we can explicitly compute 
for $d=2n+1$,
\begin{align*}
  \frac{-1}{\sqrt{2\pi i}} \res_{r=(1-d)/2} \phi_d^r(\mathrm{Ch}^d(u)) 
   =  (-1)^{n+1}  \frac{n!}{d!}\frac{\Gamma(d/2)}{\sqrt{\pi}}\,\res_{z=d}\Tr_\tau\!\left(u^*[X,u][X,u^*]\cdots[X,u](1+X^2)^{-z/2}\right)
\end{align*}
and so our index pairing can be written as
\begin{align*}
   \Index_\tau(P\hat{u} P) &=  (-1)^{n+1} \frac{n! \Gamma(d/2)}{d!\sqrt{\pi}} \res_{z=d} \Tr_\tau\!\left(u^*[X,u][X,u^*]\cdots[X,u](1+X^2)^{-z/2}\right).
\end{align*}
We make use of the identity $[X,u^*] = -u^*[X,u]u^*$, which allows us to rewrite
\begin{align*}
   u^*\underbrace{[X,u][X,u^*]\cdots[X,u]}_{d=2n+1\text{ terms}} &= (-1)^n u^*[X,u]u^*[X,u]u^*\cdots u^*[X,u]  
     = (-1)^n \left( u^*[X,u]\right)^d.
\end{align*}
Recall that $[X,u] = \sum_{j=1}^d [X_j,u]\hat\otimes\gamma^j = \sum_{j=1}^d \partial_j(u)\hat\otimes\gamma^j$ 
so we have the relation 
$u^*[X,u] = \sum_{j=1}^d u^*\partial_j(u)\hat\otimes\gamma^j$. Taking the $d$-th 
power
$$  
\left(u^*[X,u]\right)^d =  \sum_{J=(j_1,\ldots,j_d)} u^*\partial_{j_1}(u)\cdots  u^*\partial_{j_d}(u) 
 \,\hat\otimes\, \gamma^{j_1}\cdots \gamma^{j_d} 
$$
where the sum is extended over all multi-indices $J$. Note that every term in 
the sum is a multiple of the identity on $\C^\nu$ and so has a non-zero spinor trace. 
Writing this product in terms of permutations,
$$  
  (-1)^n\left(u^*[X,u]\right)^d =  
  (-1)^n \sum_{\sigma\in S_d}(-1)^\sigma \prod_{j=1}^d (u^*\partial_{\sigma(j)}(u)\hat\otimes\, \gamma^{j}), 
$$
where $S_d$ is the permutation group of $d$ letters.
Combining these results yields
\begin{align*}
  \Index_\tau(P\hat{u} P) &= (-1)^{n+1} \frac{n! \Gamma(d/2)}{d!\sqrt{\pi}}\res_{z=d} 
       \Tr_\tau\!\left(u^*[X,u][X,u^*]\cdots[X,u](1+X^2)^{-z/2}\right) \\
    &\hspace{-0cm}= - \frac{n! \Gamma(d/2)}{d!\sqrt{\pi}}\res_{z=d} 
      \Tr_\tau \!\bigg( \bigg(\sum_{\sigma\in S_d}(-1)^\sigma \prod_{j=1}^d 
         (u^*\partial_{\sigma(j)}(u)\hat\otimes\, \gamma^{j})\bigg) \!(1+X^2)^{-z/2}\bigg) .
\end{align*}
Recalling the spinor degrees of freedom,
we can apply Equation \eqref{eq:trace_of_grading} and
Proposition \ref{prop:residue_trace_formula} to reduce the formula to
\begin{align*}
  \Index_\tau(P\hat{u} P) &= (-1)^n \frac{n! \Gamma(d/2)\mathrm{Vol}_{d-1}(S^{d-1}) 2^{n}}{i^{n} d!\sqrt{\pi}} 
   \sum_{\sigma\in S_d}(-1)^\sigma \, \calT\bigg( \prod_{j=1}^d u^* \partial_{\sigma(j)}(u)\bigg).
\end{align*}
Finally we use the equation $\mathrm{Vol}_{d-1}(S^{d-1}) = \frac{d\pi^{d/2}}{\Gamma(d/2+1)}$ to simplify our formula to
$$  
\Index_\tau(P\hat{u} P) = C_d \sum_{\sigma\in S_d}(-1)^\sigma\,  \calT \bigg(\prod_{j=1}^d u^* \partial_{\sigma(j)}(u) \bigg) 
$$
with $C_{2n+1} = \frac{ 2(-2\pi)^n n!}{i^{n}(2n+1)!}=\frac{ 2(2\pi i)^n n!}{(2n+1)!}$. 
\end{proof}

We see our result as analogous to the higher 
dimensional Chern numbers of discrete crossed products considered 
in~\cite{PSB14, PSBbook, PSBKK} for $C(\Omega)\rtimes_\theta\Z^d$ for $d$ odd. 
For $d=1$ and an untwisted crossed product, $B\rtimes \R$,
we recover the results studied in~\cite{Lesch91,PR93,CGPRS}.

\subsection{Even formula}
We now consider the case of even dimensions and recall 
the even local index formula,
\begin{align*}
  \langle [p]-[1_{p}],[(\calA,\calH,X_\mathrm{even})]\rangle 
  &= \res_{r=(1-d)/2} \sum_{m=0,\text{even}}^d 
  \! \phi_m^r(\mathrm{Ch}^m(p) - \mathrm{Ch}^m(1_{p})), \\
    \mathrm{Ch}^{2n}(p) = (-1)^n \frac{ (2n)!}{2(n!)}&(2p-1)\otimes p^{\otimes 2n},  
    \qquad \mathrm{Ch}^0(p) = p,
\end{align*}
where $\phi_m^r$ is the resolvent cocycle of Definition \ref{def:resolvent_cocycle_defn} 
and $1_{p} = \pi^q(p)$ for 
$\pi^q:M_q(\calA^\sim) \to M_q(\C)$ the quotient map. We will again use 
Proposition \ref{prop:semifinite_index_is_kas_prod} to write the pairing as a 
semifinite Fredholm index.

\begin{thm}[Even index formula] 
\label{thm:higher_dim_chern_number_even}
Let $p$ be a  projection in $M_q(\calA^\sim)$ with $d$ even.
If the trace $\tau_B$ on $B$ is invariant under the action of $\R^d$, then
the semifinite index pairing can be expressed by the formula
\begin{equation*} 
 \Index_{\tau\otimes \Tr_{\C^{2q}}}\!(  \hat{p}(F_X \otimes 1_{2q})_+ \hat{p}) =  \frac{(-2\pi i)^{d/2}}{(d/2)!} \sum_{\sigma\in S_d} 
  (-1)^\sigma\, (\Tr_{\C^q}\otimes\calT) \bigg(p \prod_{j=1}^d \partial_{\sigma(j)}p \bigg), 
\end{equation*}
where $S_d$ is the permutation group of $d$ letters.
\end{thm}

Like the setting with $d$ odd, our computation can be 
simplified with some preliminary results. We again focus on the case 
$q=1$.
\begin{lemma} 
\label{lemma:only_top_term_of_local_index_survives_even_case}
The index pairing reduces to the computation 
$\res\limits_{r=(1-d)/2}\phi_d^r(\mathrm{Ch}^d(p))$.
\end{lemma}
\begin{proof}
We first note that for $m>0$, $\phi_m^r(\mathrm{Ch}(1_{p})) = 0$ as 
these terms involve the commutators $[X,1_{p}] = 0$. The proof 
used in Lemma \ref{prop:only_need_top_term_in_local_index_odd_case} 
also holds here to show that $\phi_m^r(\mathrm{Ch}^m({p}))$ does 
not contribute to the index pairing for $0<m<d$. The $m=0$ term 
is of the form
$$  
\phi_0^r(p-1_{p}) = 2 \int_0^\infty\! \Tr_\tau\!\left( \gamma (p-1_{p})(1+s^2+X^2)^{-d/2-r}\right)\mathrm{d}s, 
$$
Because there is a symmetry of the operator $(p-1_{p})(1+s^2+X^2)^{-d/2-r}$ 
between the $\pm 1$ eigenspaces of the grading operator
$\gamma = (-i)^{d/2}\gamma^1\gamma^2\cdots\gamma^d$, the graded trace 
will vanish provided $\Re(r)$ is sufficiently large. Therefore 
$\phi_0^r(p-1_{p})$ analytically continues as a function 
holomorphic in a neighbourhood of $r=(1-d)/2$, hence the residue will vanish.
\end{proof}

\begin{proof}[Proof of Theorem \ref{thm:higher_dim_chern_number_even}]
Lemma \ref{lemma:only_top_term_of_local_index_survives_even_case} 
implies our index computation is reduced to
\begin{align*}
  \left\langle [p]-[1_{p}],\left[(\calA,\calH, X_\mathrm{even}) \right]\right\rangle &= \res_{r=(1-d)/2} \phi_d^r(\mathrm{Ch}^d(p)),
\end{align*} 
which is a residue at $r=(1-d)/2$ of the term
$$
 \frac{(-1)^{d/2}d!\eta_d}{(d/2)!2\pi i} \!\int_0^\infty\!\! s^d\, \Tr_\tau\!\left(\gamma\! \int_\ell 
 \lambda^{-d/2-r} (2p-1) R_s(\lambda) [X,p] R_s(\lambda)\cdots [X,p]R_s(\lambda) \,\mathrm{d}\lambda\right)\!\mathrm{d}s
$$
with $\eta_d =2^{d+1} \frac{\Gamma(d/2+1)}{\Gamma(d+1)}$.
Like the case of $d$ odd, we can move the resolvent terms to 
the right up to a holomorphic error in order to take the Cauchy integral. 
Proposition \ref{prop:residue_trace_formula}
also implies that the semifinite trace 
$\Tr_\tau\!\left(\gamma(2{p}-1)([X,{p}])^d(1+X^2)^{-s/2}\right)$ 
has at worst a simple pole at $s=d$. Computing the residue explicitly
using the formula of Definition \ref{def:resolvent_cocycle_defn}, we find
$$  
\res_{r=(1-d)/2} \phi_d^r(\mathrm{Ch}^d(p)) 
= \frac{(-1)^{d/2}}{2((d/2)!)} ((d/2)-1)!\,\res_{z=d} 
\Tr_\tau\!\left(\gamma (2{p}-1)([X,{p}])^d(1+X^2)^{-z/2}\right), 
$$
or
$$ 
\Index_\tau(\hat{p} (F_X)_+ \hat{p}) =  (-1)^{d/2}  \frac{1}{d} \res_{z=d} \Tr_\tau\!\left(\gamma (2{p}-1)([X,{p}])^d(1+X^2)^{-z/2}\right). 
$$

Next we compute 
\begin{align*}
   [X,{p}]^d &= \sum_{J=(j_1,\ldots,j_d)}[X_{j_1},p]\cdots [X_{j_d},{p}] \,\hat\otimes\, \gamma^{j_1}\cdots \gamma^{j_d} 
     =  i^{d/2} \sum_{\sigma\in S_d}(-1)^\sigma [X_{\sigma(1)},{p}]\cdots [X_{\sigma(d)},{p}] 
     \,\hat\otimes\, \gamma
\end{align*}
as $\gamma = (-i)^{d/2}\gamma^1\cdots\gamma^d$.
Since $[X,p]\in \calB_1^\infty(X,d)$ we can 
cycle the final term $[X_{\sigma(d)},{p}]$
in this product to the front when we apply the trace, to find that 
\begin{align*}
&\Tr_\tau\!\bigg(\gamma\sum_{\sigma\in S_d}(-1)^\sigma  \Big([X_{\sigma(1)},{p}]\cdots [X_{\sigma(d)},{p}] 
     \otimes \gamma \Big) (1+X^2)^{-z/2}\bigg)\\
 &    =\Tr_\tau\bigg( \sum_{\sigma\in S_d}(-1)^\sigma \Big([X_{\sigma(1)},{p}]\cdots [X_{\sigma(d)},{p}] 
     (1+|X|^2)^{-z/2}\Big) \hat\otimes\, {\rm Id}_{\C^\nu} \bigg)\\
&     =\Tr_\tau\bigg( \sum_{\sigma\in S_d}(-1)^\sigma \Big( [X_{\sigma(d)},{p}][X_{\sigma(1)},{p}]\cdots [X_{\sigma(d-1)},{p}] 
     (1+|X|^2)^{-z/2} \Big) \hat\otimes\, {\rm Id}_{\C^\nu} \bigg).
\end{align*}
Since the cyclic permutation exchanging the first and last term is odd, we see that this sum
runs over the same set of permutations twice, once with a plus sign and once with a minus sign.
Hence for the real part of $z$ greater than $d$ we have
$$
\Tr_\tau\!\left(\gamma ([X,{p}])^d(1+X^2)^{-z/2}\right)=0,
$$
and so we need only compute the remaining term with `integrand' $2p([X,p])^d$.
As above
$$ 
{p}([X,{p}])^d =  {p} \sum_{\sigma\in S_d}(-1)^\sigma 
\prod_{j=1}^d \partial_{\sigma(j)}({p})\,\hat\otimes\, \gamma^j. 
$$
Therefore, using the relation 
$\Tr_{\C^\nu}(\gamma \gamma^1\cdots\gamma^d) = i^{d/2} 2^{d/2-1}$ 
and Proposition \ref{prop:residue_trace_formula},
\begin{align*}
 \Index_{\tau} (  \hat{p}(F_X)_+ \hat{p}) &=   (-1)^{d/2}  \frac{1}{d} \res_{z=d} \Tr_\tau\!\left(\gamma \,2{p}([X,{p}])^d(1+X^2)^{-z/2} \right) \\
    &\hspace{0cm}= \frac{(-2i)^{d/2} \mathrm{Vol}_{d-1}(S^{d-1}) }{d}\, 
    \calT \bigg(p\sum_{\sigma\in S_d}(-1)^\sigma 
      \prod_{j=1}^d \partial_{\sigma(j)}({p}) \bigg)
\end{align*}
We use the 
equation $\mathrm{Vol}_{d-1}(S^{d-1}) = \frac{d\pi^{d/2}}{(d/2)!}$ for 
$d$ even to simplify
\begin{equation} \label{eq:even_dims_chern_number}
   \Index_\tau(\hat{p} (F_X)_+ \hat{p}) =   \frac{(-2\pi i)^{d/2}}{(d/2)!} 
   \sum_{\sigma\in S_d}(-1)^\sigma \, \calT \bigg({p} \prod_{j=1}^d \partial_{\sigma(j)}({p})\bigg).  
   \qedhere
\end{equation}
\end{proof}

We remark that Equation \eqref{eq:even_dims_chern_number} appears in the case $B=C(\Omega)$ and $d=2$ 
in~\cite{Xia88, NB90}.
To relate Equation \eqref{eq:even_dims_chern_number} to the results in~\cite{PLB13,PSBbook,PSBKK}, 
we note that we have used the derivation $\partial_j(a) = [X_j,a]$, whereas 
Prodan et al. use $\tilde{\partial}(a) = \pm i[X_j,a]$. Applying our argument 
with $\tilde{\partial}$ as our algebraic derivation will bring in an extra factor of 
$i^d = (-1)^{d/2}$ and, hence, we have that
$$
   \Index_\tau(\hat{p} (F_X)_+ \hat{p}) =   \frac{(2\pi i)^{d/2}}{(d/2)!} \sum_{\sigma\in S_d}
      (-1)^\sigma \, \calT \bigg({p} \prod_{j=1}^d \tilde{\partial}_{\sigma(j)}({p})\bigg).
$$
We compare this expression to~\cite[Equation (4)]{PLB13} and see that, in 
the case of $B=C(\Omega)$ with invariant probability measure, we have reproduced 
the expression for the higher-dimensional even Chern numbers 
in the continuous (non-unital) setting. Of course, Theorem \ref{thm:higher_dim_chern_number_odd} 
and \ref{thm:higher_dim_chern_number_even} are valid for a
wider range of examples by taking $B$ to be a more general $C^*$-algebra.


\section{Extending the index pairing} 
\label{sec:localisation_general}

In this section we exploit the `flatness' of the (possibly noncommutative) Euclidean spaces
which comprise our observable algebra. Of course there is also the disorder
space $\Omega$, or `base algebra' $B$ more generally, 
but our operator $X$ does not see this data. As a consequence of the
flatness,
all but one term of the local index formula is identically zero, and this allows us
to extend the index pairing to a larger algebra. This larger algebra will be determined by the continuity of the 
Chern--Kubo functional computing the index.

Let $\calM = \pi_{\text{GNS}}(B)''$ denote the weak closure of $B$ 
under the GNS representation $B\to \calB[L^2(B,\tau_B)]$. 
The twisted action $\alpha$ of $\R^d$ on 
$B$ extends to $\calM$ and we can consider the von Neumann 
crossed product $ \calM \rtimes_\theta \R^d$. We note the following 
equivalent presentations,
$$
   \calM \rtimes_\theta \R^d \cong (B\rtimes_\theta\R^d)^{''} \cong \End_{B}^{00}(L^2(\R^d,B))''
$$
and so $\calM\rtimes_\theta\R^d$ is the same as the semifinite von Neumann algebra $\calN$ 
considered in the previous section. 
While we have a presentation of $\calN$ as a von Neumann crossed product, we will generally 
interpret $\calN$ as the weak closure $\calM \cong (B\rtimes_\theta\R^d)^{''}$ in 
$\calB[L^2(\R^d)\otimes L^2(B,\tau_B)]$. We denote the  
representation of $\calN$ on $L^2(\R^d)\otimes L^2(B,\tau_B)$ by $\wt{\pi}$.

The operator $X = \sum_{j=1}^d X_j\hat\otimes\gamma^j$
is affiliated to $\calN \hat\otimes \End(\bigwedge^*\R^d)$ and is measurable with 
respect to the trace ${\Tr}_\tau \hat\otimes \Tr_{\bigwedge^*\R^d}$. 
To identify the larger algebra to which the Chern-Kubo formula extends we will
first consider the  Fr\'{e}chet $\ast$-algebras $\calB_2(X,d)$ and $\calB_1(X,d)$ 
introduced in Appendix 
\ref{subsec:cgrs2_preliminaries}. These algebras give us an arena to study summability,
but the topology on the algebras $\calB_2(X,d)$ and $\calB_1(X,d)$ 
will prove unsuitable and we will need to consider
subalgebras endowed with different topologies.
\begin{prop} 
\label{prop:qc_condition}
Suppose that $g_1,\,g_2 \in  \calN$ 
are such that $g_j(x)\in \Dom(\tau_B)^{1/2}$ 
for almost all $x$ and satisfy the bound
\begin{equation}  
\label{eq:quasicontinuous_condition}
  \int_{\R^d}\! (1+|x|^2)^{n} \tau_B( |g_j(x)|^2) \,\mathrm{d}x < \infty, \qquad 
  j=1,2, \,\,\, n\in\N_+.
\end{equation}
Then $g_1g_2 \in \calB_1^n(X,d)$. In particular, any projection in 
$\calN$ that satisfies Equation \eqref{eq:quasicontinuous_condition} 
is in $\calB_1^n(X,d)$.
\end{prop}
The proof of Lemma \ref{lemma:Hilbert_schmidt_estimate} 
also shows that the ambiguity 
of the notation $|g(x)|$ (as convolution or 
pointwise product absolute value) disappears. 
\begin{proof}
Recall from the appendix, the norms $\varphi_s$ 
on $\calB_2^{n}(X,d)^2$. A short calculation shows that in our case
$$
\varphi_{d+1/m}( | \delta^k(g) |^2 )
=C_{d+1/m}\int_{\R^d} 2|x|^{2k}  \tau_B(|g(x)|^2) \,\mathrm{d}x  + \| \delta^k(g) \|^2,
$$
where we have used the 
cyclicity of $\tau_B$, $\tau_B(b^*b)=\tau_B(bb^*$). 
We  use this equality to 
estimate in $\calB_2^{n}(X,d)^2 \subset \calB_1^{n}(X,d)$,
\begin{align*}
\calP_{n,l}(g_1g_2)
&\leq \sum_{k=0}^l \calQ_n(\delta^k(g_1)) \calQ_n(\delta^{l-k}(g_2)) \\
&\hspace{-1cm}\leq \sum_{k=0}^l \Big( \big(\| \delta^k(g_1) \|^2 +2\int_{\R^d} |x|^{2k} \tau(|g_1(x)|^2)\,\mathrm{d}x \big) 
     \big( \| \delta^{l-k}(g_2) \|^2 +2\int_{\R^d} |z|^{2(l-k)} \tau(|g_2(z)|^2\,\mathrm{d}z \big) \Big)  \\
&\hspace{-1cm}\leq  \sum_{k=0}^l C_k \int_{\R^d} (1+3|x|^{2k}) \tau(|g_1(x)|^2)\,\mathrm{d}x  \times \int_{\R^d} (1+3|z|^{2(l-k)}) \tau(|g_2(z)|^2)\,\mathrm{d}z \\
&\hspace{-1cm}\leq \textrm{max}_{j}\, C \int_{\R^d} (1+|x|^2)^{l} \tau(|g_j(x)|^2)\,\mathrm{d}x
\end{align*}
The third inequality 
uses the fact that the $L^1$-norm dominates 
the crossed product norm. 
Hence the seminorms $P_{n,l}$ are finite for $l\leq n$.
\end{proof}

\begin{lemma}
\label{lem:tight-sum}
If $a \in \calB_1^{\lfloor d/2\rfloor + 1}(X,d)$, then
$a(1+X^2)^{-d/2-r}\in\calL^1(\calN, \Tr_\tau)$ for all $r>0$. 
\end{lemma}
\begin{proof}
We start by writing
$$
a(1+X^2)^{-d/2-r}
=(1+X^2)^{d/4+r/2}\Big((1+X^2)^{-d/4-r/2}a(1+X^2)^{-d/4-r/2}\Big)(1+X^2)^{-d/4-r/2}.
$$
Now by \cite[Lemma 1.13, Proposition 1.14]{CGRS2} we can write 
$a=\sum_{j=1}^4b_jc_j$ with $b_j,\,c_j\in \calB_2(X,d)$. Since $a\in \calB_1^{{\lfloor d/2\rfloor + 1}}(X,d)$,
we can then show that each $b_j,\,c_j\in \calB_2^{{\lfloor d/2\rfloor + 1}}(X,d)$. 
For notational simplicity we write $a=bc$ with $b,\,c\in \calB_2^{{\lfloor d/2\rfloor + 1}}(X,d)$.
Then we know that for all $r>0$,
$$
(1+X^2)^{-d/4-r/2}a(1+X^2)^{-d/4-r/2}\in \calL^1(\calM\rtimes_\theta \R^d, \Tr_\tau).
$$
Thus we have reduced the problem to considering the behaviour of the one-parameter
group $T\mapsto \sigma^z(T)=(1+X^2)^{z/2}T(1+X^2)^{-z/2}$ 
as in \cite[Section 1.4]{CGRS2}. In particular, for sufficiently smooth elements 
$b,\,c\in \calB_2(X,d)$,
we wish to show that 
$$
(1+X^2)^{z/2-d/2-r}bc(1+X^2)^{-z/2-d/2-r}
$$
is trace class.
Since
$(1+X^2)^{1/2}(1+|X|)^{-1}$ is a bounded invertible element in $L^\infty(|X|)$, 
we can simplify the computations by removing the 
square roots.
It also suffices to consider $0<r<1/2$,
and so we let $m$ be the greatest integer less than or equal to $d/2$. Iterating
the identity $(1+|X|)T(1+|X|)^{-1}=T+\delta(T)(1+|X|)^{-1}$ 
for $T\in \calB^{\lfloor d/2\rfloor}_1(X,d)$ we have
$$
(1+|X|)^{d/2+r}T(1+|X|)^{-d/2-r}=(1+|X|)^{d/2+r-m}\sum_{j=0}^m\delta^j(T)(1+|X|)^{-d/2-r+m-j},
$$
and so we will be done if we can show that for $T_1,\,T_2\in \calB_2^1(X,d)$ and
$0\leq\alpha<1$
$$
(1+|X|)^{\alpha}T_1T_2(1+|X|)^{-\alpha}\,-\, T_1T_2\,\in\,\calB_1(X,d).
$$
For this we use the integral formula for fractional powers, \cite[p701]{CP1}, and write
\begin{align*}
(1+|X|)^{\alpha}T_1T_2(1+|X|)^{-\alpha}
&=(1+|X|)^{\alpha}\,T_1T_2\,\frac{\sin(\pi\alpha)}{\pi}\,
\int_0^\infty \lambda^{-\alpha}(1+\lambda+|X|)^{-1}\,d\lambda.
\end{align*}
Taking commutators yields
\begin{align*}
&(1+|X|)^{\alpha}T_1T_2(1+|X|)^{-\alpha}
=T_1T_2
+(1+|X|)^{\alpha}\,\frac{\sin(\pi\alpha)}{\pi}\,
\int_0^\infty \lambda^{-\alpha}[T_1T_2,(1+\lambda+|X|)^{-1}]\,d\lambda\\
&=T_1T_2
-(1+|X|)^{\alpha}\,\frac{\sin(\pi\alpha)}{\pi}\,
\int_0^\infty \lambda^{-\alpha}(1+\lambda+|X|)^{-1}
(\delta(T_1)T_2+T_1\delta(T_2))(1+\lambda+|X|)^{-1}\,d\lambda.
\end{align*}
Using $(1+|X|)^{\alpha}(1+\lambda+|X|)^{-1}\leq 1$ and 
$\Vert (1+\lambda+|X|)^{-1}\Vert\leq \frac{1}{1+\lambda}$ we find that
we can estimate the $n$-th seminorm $\calP_n$ on $\calB_1(X,d)$ by
the $n$-th seminorm $\calQ_n$ on $\calB_2(X,d)$ via
\begin{align*}
\calP_n\Big((1+|X|)^{\alpha}&\,T_1T_2\,(1+|X|)^{-\alpha}\,-\,T_1T_2\Big)\\
&\leq
\frac{\sin(\pi\alpha)}{\pi}\,
\int_0^\infty \lambda^{-\alpha} \frac{1}{1+\lambda}
\Big(\calQ_n(\delta(T_1))\calQ_n(T_2)+\calQ_n(T_1)\calQ_n(\delta(T_2))\Big)\,d\lambda
\end{align*}
and this is finite for every $\alpha>0$. In particular for $T_1,\,T_2\in \calB_2^1(X,d)$,
for all $r>0$ and $\alpha>0$, the operator
$$
(1+|X|)^{\alpha-d-r}\,T_1T_2\,(1+|X|)^{-\alpha-d-r}
$$
is trace class.

Lastly, we note that in the above proof, at no point do we need to apply $\delta$ 
to either of the factors $b,\,c$ more than $\lfloor d/2 \rfloor +1$ times.
\end{proof}

To extend our index pairing to a larger algebra, we use the Sobolev spaces and 
Sobolev algebra considered in~\cite{PLB13,PSBbook} for the discrete setting.

\begin{defn} \label{def:Sob_space}
The Sobolev spaces $\calW_{r,p}$ are defined as the Banach spaces
obtained as the completion of $C_c(\R^d,B)$ in the norms 
$$
  \| f\|_{r,p} = \sum_{|\alpha| \leq r} \Tr_\tau \Big( |\partial^\alpha f|^p \Big)^{1/p}, \qquad r\in \N, \,\, p\in [1,\infty),
$$
where we use multi-index notation, $\alpha\in \N^d$, 
$\partial^\alpha = \partial_1^{\alpha_1}\partial_2^{\alpha_2}\cdots \partial_d^{\alpha_d}$ 
and $|\alpha| = \alpha_1 +\cdots +\alpha_d$.
\end{defn}

The Sobolev spaces are not closed under multiplication, but if we employ the 
H\"{o}lder inequality of noncommutative $L^p$-spaces (cf.~\cite[Theorem 4.2]{FK}),
$$
  \| a_1 \cdots a_k \|_{r,p} \leq \|a_1 \|_{r,p_1} \cdots \|a_k\|_{r,p_k}, \qquad 
  \frac{1}{p_1} + \cdots + \frac{1}{p_k} = \frac{1}{p},
$$
then we see that the intersection $\cap \calW_{r,p}$ of all Sobolev spaces is a $*$-algebra.

\begin{defn} 
\label{def:Sob_alg}
The Sobolev algebra $\calA_\mathrm{Sob}$ is defined as the algebraic 
span of products $\calA_\mathrm{Sob}=\mathrm{span}\{ab\,:\,a,b\in \wt{\calA}_\mathrm{Sob}\}$ 
with $\wt{\calA}_\mathrm{Sob}$ the intersection
$$
  \wt{\calA}_\mathrm{Sob} =  \Big( \bigcap_{ r\in \N,\, p\in \N_+}\!\! \calW_{r,p}  \Big) \cap \calN.
$$
\end{defn}

\begin{remark}[The topology of $\calA_\mathrm{Sob}$]
Let us emphasise that while we restrict our Sobolev algebra to be contained within the von Neumann 
algebra $\calN$, the von Neumann norm does not enter the topology of $\calA_\mathrm{Sob}$, 
which is entirely determined by the Sobolev norms $\|\cdot\|_{r,p}$. This choice of topology 
means that $\calA_\mathrm{Sob}$ is locally convex $\ast$-algebra (but not a Banach nor Fr\'{e}chet algebra). Hence,
 the topology of $\calA_\mathrm{Sob}$ is quite different from the topology determined by the operator norm.
This difference is necessary in order to meaningfully extend our index pairing. While we can make 
sense of index pairings in $\calA_\mathrm{Sob}$, 
it is not easy in general to relate $\A_\mathrm{Sob}$ to the $C^*$-algebra we first considered.

Our Sobolev algebra $\A_\mathrm{Sob}$ is defined using the span of products rather than the algebra 
$\wt{\calA}_\mathrm{Sob}$ for largely technical reasons that appear in 
the non-unital setting. For applications to topological phases, this extra detail 
is not an issue as the $K$-theoretic phase of interest is constructed out of 
the Fermi projection $P_\mu= P_\mu^2$ or $1-2P_\mu$.
\end{remark}

While the $\ast$-algebra $\calA_\mathrm{Sob}$ and its topology is quite different from the 
Fr\'{e}chet algebras $\calB_2^n(X,d)$ and $\calB_1^n(X,d)$, there is still a relationship 
between the two constructions.
	
\begin{lemma} 
\label{lemma:Sob_in_qc}
If $a\in \calA_\mathrm{Sob}$, then $a\in \calB_1^n(X,d)$ for any $n\in\N_+$.
\end{lemma}
\begin{proof}
Because $\calB_2^n(X,d)^2 \subset \calB_1^n(X,d)$, 
the result follows if we can show that $b \in \calB_2^{n}(X,d)$ for $b\in\wt{\calA}_\mathrm{Sob}$. 
Recalling the weight $\varphi_s$ from Appendix \ref{subsec:cgrs2_preliminaries}, we note that for $s>d$,
$$
  \varphi_s(b^*b) = \Tr_\tau \!\big( (1+X^2)^{-s/4} b^*b (1+X^2)^{-s/4} \big)
$$
can easily be bound by $\|b^*b\|_{r,p}$ for some $r,p$. Similarly, 
$\varphi_s(\delta^k(b^*b))$ is bound by $\|b^*b\|_{r+k,p}$ for $\delta(T) = [|X|,T]$. 
The seminorms $\calQ_n$ on $\calB_2(X,d)$ also contain the (operator) norm of the von Neumann algebra $\calN$. While 
this norm does \emph{not} enter the topology of $\calA_\mathrm{Sob}$, because $\calA_\mathrm{Sob}$ 
is defined as an intersection with $\calN$, the norm is still finite.
Hence $b\in \calB_2^k(X,d)$, a result that also follows 
using the condition from Equation \eqref{eq:quasicontinuous_condition}.
\end{proof}	

By adapting arguments developed for $\calB_1^n(X,d)$, Lemma \ref{lemma:Sob_in_qc} can then be used to obtain the following.

\begin{prop}  \label{prop:quasicts_spectrip_fin_summ}
\label{prop:quasicts_semifinite_spec_trip}
The tuple
$\left(\calA_\mathrm{Sob}\hat\otimes C\ell_{0,d}, L^2(\R^d)\otimes L^2(B,\tau_B) 
 \hat\otimes\bigwedge\nolimits^{\!*}\R^d, \sum_j X_j \hat\otimes\gamma^j \right)$ 
is a finitely summable semifinite spectral triple  
with spectral dimension $d$.
\end{prop}
\begin{proof}
The operators $[X,\tilde{\pi}(g)]$ are bounded by the regularity 
of elements in the Sobolev spaces and algebra. For finite summability 
we apply Lemmas \ref{lemma:Sob_in_qc} and \ref{lem:tight-sum}.
\end{proof}

Proposition \ref{prop:quasicts_spectrip_fin_summ} is valid for both real and complex 
Sobolev algebras. We now restrict to complex pairings and the extension of 
the Chern number formulas derived in Section \ref{sec:Chern_nos}.

\begin{lemma} \label{lem:sobolev_cocycle}
The multi-linear functional
$$
  \phi(a_0,a_1,\ldots,a_d) = {\rm res}_{s=d}\Tr_\tau\!\left( a_0 \partial_1(a_1)\cdots \partial_d(a_d)(1+X^2)^{-s/2} \right), \qquad 
    a_0,\ldots,a_d \in \calA_\mathrm{Sob}
$$
is well-defined and continuous with respect to the topology on $\calA_\mathrm{Sob}$.
Furthermore, 
$$
   \phi(a_0,a_1,\ldots,a_d) = K_d \sum_{\sigma\in S_d} (-1)^\sigma 
   \calT \big(a_0 \partial_{\sigma(1)}a_1 \cdots \partial_{\sigma(d)} a_d \big).
$$
\end{lemma}
\begin{proof}
The functional 
is well-defined and continuous  by the H\"{o}lder inequality 
of the Sobolev spaces (or Lemma \ref{lem:tight-sum}).
Finally, the last equality follows by analogous algebraic arguments 
as was done in Section \ref{sec:Chern_nos} and the observation that 
Proposition \ref{prop:residue_trace_formula} can also be applied to 
elements in $\calA_\mathrm{Sob}$.
\end{proof}

\begin{thm} 
\label{thm:quasicts_index_extension}
The index formulas given in 
Theorems \ref{thm:higher_dim_chern_number_odd} and 
\ref{thm:higher_dim_chern_number_even} extend to any 
projection or unitary in $M_q(\calA_\mathrm{Sob}^\sim)$ (complex algebras).
\end{thm}
\begin{proof}
By Lemma \ref{lem:sobolev_cocycle}, we know that the tracial formula for the index
is well-defined
and so we just need to identify
the formula with the index pairing.

By \cite[Proposition 2.14]{CGRS2} our Sobolev spectral 
triple determines a semifinite Fredholm module 
with operator
$X(1+X^2)^{-1/2}$ and is
 $(d+1)$-summable over $\calA_\mathrm{Sob}$. 
 Therefore the operators 
\begin{align*}
    &(P_X \otimes 1_{2q})\hat{u}(P_X \otimes 1_{2q}), 
    &&\hat{p}(F_X \otimes 1_{2q})_+ \hat{p}
\end{align*}
 from the statement of Theorems \ref{thm:higher_dim_chern_number_odd} and 
\ref{thm:higher_dim_chern_number_even}
 are $\Tr_\tau$-Fredholm for $p,u \in M_q(\calA_\mathrm{Sob}^\sim)$. 

Because the left and right hand side of the index formulas from Theorems \ref{thm:higher_dim_chern_number_odd} and 
\ref{thm:higher_dim_chern_number_even} continue to be well-defined for $\calA_\mathrm{Sob}$, 
which is defined via a completion of $C_c(\R^d,B)$, the index formulas continuously extend. 
 \end{proof}

A difficulty that we encounter with extending the 
index pairing is that, as defined, there is no guarantee that the algebra
$\calA_{\mathrm{Sob}}$ is separable, 
and typically it will not be. 
For index pairings the lack of separability is not a problem:
given a projection or unitary over $\calA_{\mathrm{Sob}}$, we can restrict to the separable
algebra generated by this projection or unitary as in \cite{BCPRSW}, and so the 
formulae for the pairing are valid.
What is in question is homotopy invariance of the pairing 
for homotopies continuous in the topology of $\calA_\mathrm{Sob}$. 
 
Consider a separable 
subalgebra $\mathcal{C}$ of 
$\calA_\mathrm{Sob}$ and suppose it 
has a $C^*$-closure $C$.\footnote{We remark there may be no clear connection between 
$C$ and $B\rtimes_\theta \R^d$ in general.} We define a new Kasparov 
module. The semifinite spectral triple above 
has Hilbert space $\mathcal{H}=L^2(E,\tau_B)$. Here $E_B\cong L^2(\R^d,B)$ is the $B$-module of
the Kasparov $A$-$B$-module
$\lambda_d$ from Proposition \ref{prop:crossed_prod_Kas_mod}, and $\tau_B:\,B\to \C$ the trace. The inner product
on $E_B$ remains well-defined for elements of $\mathcal{C}\cdot E$, and we complete 
$(\mathcal{C}\cdot E\mid \mathcal{C}\cdot E)\subset \pi_{GNS}(B)''$ in norm
to obtain an algebra $D$. In turn we complete
$\mathcal{C}\cdot E$ in the resulting Hilbert module norm, and we 
obtain a 
Kasparov $C$-$D$-module.
So we obtain well-defined pairings
$$
K_\ast(C)\times KK^\ast(C,D)\to K_0(D).
$$
We can then use Proposition 
\ref{prop:semifinite_index_is_kas_prod} in the appendix to conclude that the semifinite index 
represents the composition 
$$
 K_\ast(C)\times KK^\ast(C,D)\to K_0(D) \xrightarrow{\tau_B} \R.
$$
Because $D$ can be taken to be separable, we therefore have that the range of the 
semifinite index is countably generated (though not necessarily discrete).

\subsection{The case $B=C(\Omega_0)$}

We consider the case of $B=C(\Omega)$ with $\Omega$ a compact Hausdorff space 
with a faithful measure probabililty measure $\bP$ that is invariant and under a twisted $\R^d$-action.

\subsubsection{Direct integral decomposition} \label{sec:chop-chop}

Before computing index pairings, we record some further details about the 
representation from Lemma \ref{lem:rep} in the special case when $B=C(\Omega)$ 
with the 
the trace $\tau_\bP$ on $C(\Omega)$ coming from the probability measure $\bP$.

For each $\omega\in \Omega$, the evaluation homomorphism ${\rm ev}_\omega:\,\Omega\to\C$ defines a Kasparov 
module $(C(\Omega),{}_{{\rm ev}_\omega}\C_\C,0)$ with class in 
$KK(C(\Omega),\C)$ (we similarly obtain a class in $KKO(C(\Omega),\R)$ for 
real-valued functions). The product of our Kasparov module from Proposition \ref{prop:crossed_prod_Kas_mod}
with the class of ${\rm ev}_\omega$ is a spectral triple
$$
\Big( \A\hat\otimes C\ell_{0,d},\, L^2(\R^d,C(\Omega))\otimes_{\mathrm{ev}_\omega}\C \, \hat\otimes \bigwedge\nolimits^{\! *}\R^d, 
    \, X\hat\otimes 1  \Big)
$$
since $L^2(\R^d,C(\Omega))\otimes_{\mathrm{ev}_\omega} \C\cong L^2(\R^d)$.

Thus for each $\omega\in\Omega$ we obtain a representation 
$\pi_\omega:\,\A\to \calB(L^2(\R^d))$ by setting
$$
\pi_\omega(T)=T\otimes1:L^2(\R^d,C(\Omega))\otimes_{\mathrm{ev}_\omega}\C\to L^2(\R^d,C(\Omega))\otimes_{\mathrm{ev}_\omega} \C.
$$
The definition of the Hilbert space completion of $L^2(\R^d,C(\Omega))$ uses the inner product
defined  
in Equation \eqref{eq:inn-prod}. For $f_1,\,f_2\in L^2(\R^d,C(\Omega))$ we have 
$$
\langle f_1,f_2\rangle=\tau_\bP \!\big( (f_1 \mid f_2)_{C(\Omega)} \big) =\int_{\Omega}\langle f_1(\cdot,\omega),f_2(\cdot,\omega)\rangle\,\mathrm{d}\bP(\omega),
$$
and so we have the direct integral decomposition (coming from the abelian subalgebra $C(\Omega)''$ of the commutant of $\A$, \cite[Section III.1.6]{Blackagain})
$$
L^2(\R^d)\otimes L^2(\Omega,\bP)\cong \int_\Omega^{\oplus}L^2(\R^d)_\omega\,\mathrm{d}\bP(\omega).
$$
The integral decomposition is compatible with the representations $\pi_\omega$
in that the action of $\A$ on $L^2(\R^d)\otimes L^2(\Omega,\bP)$ is the direct integral of the representations $\pi_\omega$.

As well as the representation, the operator $X$ (densely) defined on $L^2(\R^d)\otimes L^2(\Omega,\bP)$ is the direct integral of the 
operators $X$ (densely) defined on  $L^2(\R^d)$. Hence the semifinite spectral triple
$$
\bigg( \calA \hat\otimes C\ell_{0,d}, \, L^2(\R^d)\otimes L^2(\Omega,\bP) \hat\otimes 
   \bigwedge\nolimits^{\!*}\R^d, \, X=\sum_{j=1}^d X_j\otimes 1 \hat\otimes \gamma^j, \, 
   (\calN, \Tr_\tau) \bigg)
$$
can be regarded as the direct integral of the spectral triples
$$
\bigg( \calA \hat\otimes C\ell_{0,d}, \, {}_{\pi_\omega}L^2(\R^d) \hat\otimes 
   \bigwedge\nolimits^{\!*}\R^d, \, X=\sum_{j=1}^d X_j \hat\otimes \gamma^j, \, 
    \bigg).
$$
This decomposition remains valid (a.e.) for the algebra $\calA_\mathrm{Sob}$ since it is contained in $\calN$.
The trace $\Tr_\tau$ on $\calN$ is 
given by $\Tr_\tau(T)=\int_\Omega \Tr(T_\omega)\,\mathrm{d}\bP(\omega)$ for any measurable family 
$(T_\omega)\in \Dom(\Tr_\tau)\subset \calN$. 
Therefore, for any projection $p\in M_q(\calA_\mathrm{Sob}^\sim)$, the semifinite index of the Fredholm operator
$\hat{p}(F_{X}\otimes1_{2q})_+ \hat{p}$ is 
$$
\int_\Omega{\rm Index}(\wt{\pi}_\omega(\hat{p})(F_{X}\otimes1_{2q})_+ \wt{\pi}_\omega(\hat{p}))\,\mathrm{d}\bP(\omega).
$$
We can also consider pointwise defined torsion classes, e.g. 
$\dim\Ker((\wt{\pi}_\omega(\hat{p}) F_X \wt{\pi}_\omega(\hat{p}))_+)\,\mathrm{mod}\, 2$ arising from skew-adjoint Fredholm indices, 
but the integral of these quantities needs more care.

\subsubsection{Pairings with ergodic measures} \label{sec:ergodic_pairing}

We have used a direct integral decomposition of the semifinite spectral triple to obtain a  
concrete formula for the semifinite index pairing. We now consider the case where the probability 
measure $\bP$ is invariant and ergodic under the twisted $\R^d$-action (that is, 
the only functions $L^2(\Omega,\bP)$ invariant under the $\R^d$-action are constant functions), where we can further 
reduce our complex semifinite pairings to $\Z$-valued quantities.

\begin{thm} 
\label{thm:ergodic_loc_pairing}
If the trace $\tau_\bP$ on $C(\Omega)$ comes from a faithful measure $\bP$ that is invariant and ergodic 
under the twisted $\R^d$-action, then the index formulas given in 
Theorems \ref{thm:higher_dim_chern_number_odd} and 
\ref{thm:higher_dim_chern_number_even} extend to 
$\calA_\mathrm{Sob}$ and are  integer-valued. 
Furthermore, the index is invariant under continuous deformations in the 
Sobolev topology.
\end{thm}
\begin{proof}
Theorem \ref{thm:quasicts_index_extension} gives the semifinite index
pairing of the semifinite spectral 
triple from Proposition \ref{prop:quasicts_semifinite_spec_trip}. 
The direct integral decomposition now shows that the semifinite index pairing is the 
integral of the family of index pairings with the spectral triples
$$
\bigg( \calA_\mathrm{Sob}, \, {}_{\wt{\pi}_\omega}L^2(\R^d) \hat\otimes 
   \C^\nu , \, X=\sum_{j=1}^d X_j \hat\otimes \gamma^j 
    \bigg),
$$
where we have changed to the spin$^c$ Clifford representation as these are the semifinite 
spectral triples used in Theorems \ref{thm:higher_dim_chern_number_odd} and 
\ref{thm:higher_dim_chern_number_even}.

Because the measure on $\Omega$ is ergodic, it suffices to check that the pairing is 
$\bP$-almost surely constant on any orbit (and so constant a.e.).
To show this constancy, we remark that if $\omega' = T_{-a} \omega$, then using the corresponding 
covariance relation, $F_X$ is unitarily equivalent to $F_{X+a}$, the bounded transform of 
$\sum_j (X_j + a_j) \hat\otimes \gamma^j$, via the unitary $U_a$ implementing $T_{-a}$.
Since $U_a[X,U_a^*]$ is bounded, we have a bounded perturbation of the unbounded operator $X$. This implies that 
$$
\bigg( \calA_\mathrm{Sob}, \, {}_{\wt{\pi}_\omega}L^2(\R^d) \hat\otimes \C^\nu, \, X=\sum_{j=1}^d X_j \hat\otimes \gamma^j 
    \bigg)
$$
is unitarily equivalent to 
$$
\bigg( \calA_\mathrm{Sob}, \, {}_{\wt{\pi}_{\omega'}}L^2(\R^d) \hat\otimes \C^\nu,
 \, X=\sum_{j=1}^d X_j \hat\otimes \gamma^j+K  \bigg)
    $$
    where $K$ is bounded. Hence the bounded operator 
    $\wt{\pi}_\omega(\hat{p})(F_{X+a})_+\wt{\pi}_\omega(\hat{p})$ will be a compact 
    perturbation of $\wt{\pi}_\omega(\hat{p})(F_{X})_+\wt{\pi}_\omega(\hat{p})$ and 
    so the index will not change. The same argument also applies for pairing with 
    unitaries in $\calA_\mathrm{Sob}^\sim$.

Next, we 
consider a continuous deformation in $\calA_\mathrm{Sob}$. Because 
the Hochschild cocycle is continuous in the Sobolev topology, 
the cyclic expression for the index will change continuously as we make this deformation. 
However, the equality of the cyclic formula with the Fredholm index for \emph{any} 
projection or unitary in $\calA_\mathrm{Sob}^\sim$ ensures 
that the cyclic pairing is always $\Z$-valued. Thus the index pairing along this path 
gives a continuous $\Z$-valued function, and so is constant.
\end{proof}

\begin{remark}
While taking the intersection over all Sobolev spaces appears to be 
  quite restrictive, we will see in Section \ref{subsec:qc_and_localisation} that 
dynamically localised observables often have, on average over the 
configuration space $\Omega$, \emph{exponentially 
decaying} integral kernels. As such, our index theory over $\calA_\mathrm{Sob}$ 
can be applied in this situation.
We will also show that 
(under extra restrictions), deformations within a region of 
dynamical localisation are continuous in the Sobolev
topology. 
\end{remark}



\section{The bulk-edge correspondence} \label{sec:bulkedge}

The bulk-edge correspondence is a key property of topological 
states of matter, where
non-trivial topological properties in the bulk (interior) of a physical 
system give rise to edge behaviour, e.g. the existence of stable 
edge states and edge conductivity. Driving the bulk-edge correspondence 
for the $C^*$-algebraic approach to condensed matter physics is 
a short exact sequence linking bulk and edge 
observable algebras~\cite{SBKR02,KSB04b,KR06, HMTBulkedge}.

The short exact sequence encodes the boundary map in $K$-theory or $K$-homology 
(or their extension $KK$-theory). Gapped topological phases are encoded in 
$K$-theory classes and one then studies the effect of the boundary map on 
this class. However, because the topological phases of interest 
arise as index \emph{pairings}, we need to understand how the invariants change under the boundary 
map in both $K$-theory and its dual theory. In earlier work on the quantum Hall effect, 
this was achieved using cyclic cohomology~\cite{KSB04b}, but this will not apply to 
torsion phases, which include some of the most interesting examples of topological insulators. 
Therefore we instead work with 
$K$-homology (actually, $KK$-theory) and study the boundary map in full 
generality. By expressing topological phases as index 
pairings, 
 our $K$-theoretic result on boundary maps 
immediately implies the bulk-edge correspondence for 
the numerical phase labels.

In this section we work with the $C^*$-algebra $B\rtimes_\theta\R^d$ and the 
unbounded Kasparov module from Proposition \ref{prop:crossed_prod_Kas_mod}. 
Hence the section is mostly independent from the extension of the index 
formulas in Section \ref{sec:localisation_general} (for a connection, 
see Section \ref{sec:deloc_edge}, 
 where the bulk-boundary correspondence and the Sobolev 
algebra can be used to prove delocalisation of complex edge states).

\subsection{The Wiener--Hopf extension}
By considering crossed product algebras by $\R$, there is a 
natural short exact sequence, namely the Wiener--Hopf extension 
(see for example~\cite{Rieffel82}). 
Recalling the discussion in Section \ref{Sec:twisted_systems_prelim}, 
we can decompose the twisted crossed product $B\rtimes_\theta\R^d$ 
as an iterated crossed product of a twisted 
crossed product by $\R^{d-1}$ and an untwisted $\R$-crossed 
product, $(B\rtimes_\theta\R^{d-1})\rtimes \R$. 
In the case $B=C(\Omega)$ with $\theta(x,-x) =1$ this can be 
done via an explicit isomorphism~\cite{KR06}. For general $B$, 
the decomposition is equivalent to our original twisted crossed product 
at the level of $KK$-theory.
We let $A_e = B\rtimes_\theta\R^{d-1}$, the 
observables on the edge of a system with boundary, and 
$A_b = A_e\rtimes\R$ the algebra on a boundaryless system.

Following~\cite{KR06,KSB04b} our bulk-edge short exact sequence is
\begin{equation} \label{eq:cts_bulkedge_SES}
  0 \to \calK \otimes A_e \to \left(C_0(\R\cup \{+\infty\})\otimes A_e \right)\rtimes \R 
     \to A_e \rtimes \R \to 0
\end{equation}
where the $\R$-action on $C_0(\R\cup \{+\infty\})\otimes A_e$ 
is by translation on 
$C_0(\R\cup \{+\infty\})$ (with fixed point at $+\infty$) 
and by the automorphism on $A_e$ such that $A_b = A_e\rtimes\R$. 
In order to compute boundary maps in $KK$-theory,
we first represent Equation \eqref{eq:cts_bulkedge_SES} as an unbounded 
Kasparov module by the isomorphism 
$KKO(A\hat\otimes C\ell_{0,1},C) \cong \mathrm{Ext}^{-1}(A,C)$ for 
separable $C^*$-algebras $A$ and $C$~\cite[{\S}7]{Kasparov80}.

\begin{prop} 
\label{prop:ext_is_crossed_prod_module}
The unbounded crossed-product Kasparov module
\begin{equation} 
\label{eq:ext_unbdd_kasmod}
 \left( C_c(\R,A_e) \hat\otimes C\ell_{0,1}, \, 
 L^2(\R,A_e)_{A_e} \hat\otimes \bigwedge\nolimits^{\!*}\R, \, 
    X_\mathrm{ext} \hat\otimes \gamma_\mathrm{ext} \right),
\end{equation}
represents the class of the extension of 
Equation \eqref{eq:cts_bulkedge_SES} in $KKO(A_b\hat\otimes C\ell_{0,1},A_e)$.
Here $\gamma_{\mathrm{ext}}$ is the generator of $C\ell_{1,0}$ and 
$X_{\mathrm{ext}}$ is the multiplication operator by the independent variable
in $\R$.
\end{prop}
\begin{proof}
Our Kasparov module is precisely the unbounded Kasparov module $\lambda_d$ we have already 
considered in
Proposition \ref{prop:crossed_prod_Kas_mod} for $d=1$. 
Our task, therefore, is to show that this unbounded module represents the 
Wiener--Hopf extension in Equation \eqref{eq:cts_bulkedge_SES}. 

Associated to the graded Kasparov module from the Equation \eqref{eq:ext_unbdd_kasmod} is the ungraded 
(odd) module $\left(C_c(\R,A_e), L^2(\R,A_e), X_\mathrm{ext}\right)$, from which we can 
construct an extension.
First we use Connes' trick~\cite{Connes85} to double our unbounded Kasparov module 
to the tuple
$$
 \bigg( \begin{pmatrix} a & 0 \\ 0 & 0 \end{pmatrix},\, L^2(\R,A_e)\oplus L^2(\R,A_e), \,
   X_m = \begin{pmatrix} X_\mathrm{ext} & m \\ m & -X_\mathrm{ext} \end{pmatrix} \bigg),
   \quad m>0,
$$
which does not change the class in $KKO^1(A_b,A_e)$ and has the advantage that 
$X_m$ has a spectral gap around $0$ (see also \cite[Section 2.7]{CGRS2} for
another method). Next we let $P = \chi_{[0,\infty)}(X_m)$, 
which up to a locally compact pertubation is exactly the projection $\Pi\oplus \Pi$, with
$\Pi: L^2(\R,A_e)\to L^2(\R_+,A_e)$ the projection onto the half-space Hilbert module. 
Therefore given the module $\left(C_c(\R,A_e), L^2(\R,A_e), X_\mathrm{ext}\right)$ we can 
associate the extension
$$
 0 \to \calK[L^2(\R_+,A_e)] \to C^*(PA_b P, \calK[L^2(\R_+,A_e)]) 
   \to A_b \to 0
$$
with positive semisplitting by $P$. Hence the Kasparov module gives rise to the 
Busby invariant
$$
   \phi: A_b \to \calQ(A_e),  \qquad \phi(a) = p(PaP),
$$ 
with $p:\calM(A_e\otimes \calK) \to \calQ(A_e \otimes \calK)$ the corona projection.
Next we consider the Wiener--Hopf extension
$$
  0 \to \calK \otimes A_e \to (C_0(\R\cup\{+\infty\})\otimes A_e)\rtimes\R \to A_b \to 0.
$$
We take a function $g\in C_0(\R\cup\{+\infty\})$ 
that is $0$ for $x\leq 0$, smoothly goes to $1$ for $0\leq x \leq m/2$ and is $1$ for 
all $x>m/2$. Then the map $f\mapsto gf$ for $f\in C_c(\R,A_e)$ gives rise to a 
map $A_e\rtimes\R \to (C_0(\R\cup\{+\infty\})\otimes A_e)\rtimes\R$ and the 
Busby invariant $\tilde{\phi}(f) = p((gf)(+\infty))$, where
$$
  \tilde{\phi}(f) = p((gf)(+\infty)) \in \calQ(C_0(\R)\rtimes \R \otimes A_e) \cong 
    \calQ(\calK\otimes A_e).
$$ 
The maps $a\mapsto PaP$ and $f\mapsto gf$ differ by a compact operator on $L^2(\R_+,A_e)$ 
and, so we have that $\phi = \tilde{\phi}$ and the extensions are equivalent.
\end{proof}

\begin{remark}[The Thom class]
We note that the unbounded Kasparov module coming from 
 an (untwisted) $\R$-action and representing the Wiener--Hopf extension 
is the inverse of the class in $KK$-theory implementing the Connes--Thom isomorphism. 
An explicit representative of the inverse to the class from Proposition \ref{prop:ext_is_crossed_prod_module} 
is constructed in~\cite{Andersson14, AnderssonThesis}, where it is shown that the 
class implements the Connes--Thom isomorphism. 
See also the work of Rieffel~\cite{Rieffel82}, who showed 
that the boundary map from the Wiener--Hopf extension implements the inverse of the 
Connes--Thom isomorphism. 
\end{remark}

\subsection{The edge Kasparov module and the product}
Given the edge algebra $A_e = B\rtimes_\theta\R^{d-1}$ with $d\geq 2$, we 
can construct an unbounded Kasparov module 
$$
 \lambda_{d-1} = \bigg( C_c(\R^{d-1},B)\hat\otimes C\ell_{0,d-1}, \, 
  L^2(\R^{d-1},B)\hat\otimes \bigwedge\nolimits^{\!*} \R^{d-1}, \, 
   \sum_{j=1}^{d-1} X_j\hat\otimes \gamma^j \bigg)
$$ 
by Proposition \ref{prop:crossed_prod_Kas_mod}. 
The internal Kasparov product of the extension 
class from Proposition \ref{prop:ext_is_crossed_prod_module} 
with $\lambda_{d-1}$ defines 
a map $ KKO^1(A_b,A_e)\times KKO^{d-1}(A_e,B) \to KKO^d(A_b,B)$. Our 
central result of this section is that the product at 
the unbounded level produces, up to 
a permutation of Clifford generators, the `bulk' Kasparov module 
$\lambda_d$. The result is a continuous analogue of~\cite{BCR14,BCR15,BKR1}, which 
studied crossed products by $\Z^d$.

\begin{thm} \label{thm:gen_bulkedge}
The Kasparov product $[\mathrm{ext}]\hat\otimes_{A_e}[\lambda_{d-1}]$ is 
represented by the unbounded Kasparov module,
$$
  \bigg( C_c(\R^d,B) \hat\otimes C\ell_{0,d},\, L^2(\R^d, B)
 \hat\otimes \bigwedge\nolimits^{\!*}\R^d,\, 
 X_d\hat\otimes \gamma^1 + \sum_{j=1}^{d-1} X_{j} \hat\otimes \gamma^{j+1} \bigg).
$$
Furthermore $[\mathrm{ext}]\hat\otimes_{A_e}[\lambda_{d-1}] = (-1)^{d-1}[\lambda_d]$, 
where $-[x]$ represents the inverse class in the $KK$-group.
\end{thm}
\begin{proof}
We will focus on the real setting as the case of complex algebras and 
spaces follows the same argument. We denote 
by $A_e = B\rtimes_{\theta}\R^{d-1}$ and 
$A_b = B\rtimes_{\theta}\R^d \cong A_e \rtimes \R$.
We are taking the internal product of 
an $A_b\hat\otimes C\ell_{0,1}$-$A_e$ module 
with an $A_e \hat\otimes C\ell_{0,d-1}$-$B$ module. To 
take this product, we first take the 
external product of the $A_b \hat\otimes C\ell_{0,1}$-$A_e$ module 
with the identity class in $KKO(C\ell_{0,d-1},C\ell_{0,d-1})$.
This class can be represented by the Kasparov module
$$ 
 \left( C\ell_{0,d-1},\, \left(C\ell_{0,d-1}\right)_{C\ell_{0,d-1}}, \,0\right) 
$$
with right and left actions given by right and left multiplication.
The external product gives the
$A_b\hat\otimes C\ell_{0,d}$-$A_e\hat\otimes C\ell_{0,d-1}$ module 
$$ 
\left(C_c(\R,A_e) \hat\otimes C\ell_{0,1} \hat\otimes C\ell_{0,d-1}, \, \left(L^2(\R,A_e)\hat\otimes 
    \bigwedge\nolimits^{\!*}\R \hat\otimes  C\ell_{0,d-1}\right)_{A_e \hat\otimes C\ell_{0,d-1}} , \,
    X_\mathrm{ext} \hat\otimes \gamma_{\mathrm{ext}}\hat\otimes 1  \right).
$$
We now take the internal product of this module with the edge 
module $\lambda_{d-1}$.
We start with the $C^*$-modules, where
\begin{align*}
 &\left( L^2(\R,A_e)\hat\otimes_\R \bigwedge\nolimits^{\!*}\R \,\hat\otimes_\R \, C\ell_{0,d-1}\right) 
    \hat\otimes_{A_e\hat\otimes C\ell_{0,d-1}} 
      \left(L^2(\R^{d-1},B) \hat\otimes_\R \bigwedge\nolimits^{\!*}\R^{d-1} \right) \\
 &\hspace{4cm}  \cong \left( L^2(\R, A_e)\otimes_{A_e} L^2(\R^{d-1},B)\right) 
   \hat\otimes_\R \bigwedge\nolimits^{\!*}\R \,\hat\otimes_\R 
      \left(C\ell_{0,d-1}\cdot \bigwedge\nolimits^{\!*}\R^{d-1} \right) \\
 &\hspace{4cm} \cong \left( L^2(\R, A_e)\otimes_{A_e} L^2(\R^{d-1},B)\right)
 \hat\otimes_\R \bigwedge\nolimits^{\!*}\R\, \hat\otimes_\R \bigwedge\nolimits^{\!*}\R^{d-1}
\end{align*}
as the action of $C\ell_{0,d-1}$ on $\bigwedge^*\R^{d-1}$ by left-multiplication is nondegenerate.

Next we define $1\otimes_\nabla X_j$ on the dense 
submodule $C_c(\R^{d-1},B)\otimes_{\calA_e} L^2(\R^{d-1},B)$ 
for all $j\in\{1,\ldots,d-1\}$ and $\calA_e = C_c(\R^{d-1},B)$. We consider the 
connection $\nabla_j:\calA_e\to \calA_e\otimes_{\calA_e^\sim}\Omega^1(\calA_e^\sim)$ 
defined from the derivation $\partial_j a_e = [X_j,a_e]$. 
From this connection we construct the unbounded operator
\begin{equation} \label{eq:twisted_dirac}
  (1\otimes_\nabla X_j)(\psi_1\otimes \psi_2) = \psi_1\otimes X_j \psi_2 + \nabla_j(\psi_1)\psi_2
\end{equation}
for $\psi_1\in C_c(\R^{d-1},B)$ and $\psi_2\in \Dom(X_j)\subset L^2(\R^{d-1},B)$.
We refer the reader to~\cite{MeslandMonster,KL13,MR15} for more details on 
connections and the construction of operators like $1\hat\otimes_\nabla X_j$. 
Then
\begin{align} \label{eq:prod_module_candidate}
  &\bigg( C_c(\R,\calA_e) \hat\otimes C\ell_{0,1}\hat\otimes C\ell_{0,d-1},\, 
     \left( L^2(\R, A_e)\otimes_{A_e} L^2(\R^{d-1},B)\right) \hat\otimes_\R
       \bigwedge\nolimits^{\!*}\R\, \hat\otimes_\R \bigwedge\nolimits^{\!*}\R^{d-1}, \\ 
    &\hspace{7.5cm}  X_\mathrm{ext} \otimes 1 \hat\otimes \gamma_{\mathrm{ext}}\hat\otimes 1 
     + \sum_{j=1}^{d-1}(1\otimes_\nabla X_j)\hat\otimes 1\hat\otimes \gamma^j \bigg) \nonumber
\end{align}
is a candidate for the unbounded product module, where the Clifford actions take the form
\begin{align*}
   \rho_{\mathrm{ext}}\hat\otimes 1(\omega_1\hat\otimes\omega_2) &= (e_1\wedge \omega_1 - \iota(e_1)\omega_1)\hat\otimes \omega_2 \\
   1\hat\otimes \rho^j(\omega_1\hat\otimes\omega_2) &= (-1)^{|\omega_1|} \omega_1\hat\otimes(e_j\wedge \omega_2 - \iota(e_j)\omega_2),
\end{align*}
for $j\in\{1,\ldots,d-1\}$ and $|\omega_1|$ is the degree of the form $\omega_1$. 
Analogous formulas exist for the representation of $\gamma_\mathrm{ext}\hat\otimes 1$ 
and $1\hat\otimes\gamma^j$. Arguments very 
similar to the proof of Proposition \ref{prop:crossed_prod_Kas_mod} show that Equation 
\eqref{eq:prod_module_candidate} is a real or complex Kasparov module depending on what 
setting we are in.
A simple check of
Kucerovsky's criterion \cite[Theorem 13]{Kucerovsky97}, as in \cite{BCR15, BKR1},
shows that the unbounded Kasparov module of Equation 
\eqref{eq:prod_module_candidate}
is an unbounded representative of the class $[\text{ext}]\hat\otimes_{A_e}[\lambda_{d-1}]$.

Our next task is to relate the module \eqref{eq:prod_module_candidate} to 
$\lambda_d$. 
We first identify  $\bigwedge^*\R\, \hat\otimes_\R \bigwedge^* \R^{d-1} \cong \bigwedge^*\R^d$ 
and use the graded isomorphism $C\ell_{p,q}\hat\otimes C\ell_{r,s} \cong C\ell_{p+r,q+s}$ 
from~\cite[{\S}2.16]{Kasparov80}
on the left and right Clifford generators by the mapping
\begin{align*}
   &\rho_{\mathrm{ext}}\hat\otimes 1 \mapsto \rho^1,  &&1\hat\otimes \rho^j \mapsto \rho^{j+1}, \\
   &\gamma_{\mathrm{ext}} \hat\otimes 1 \mapsto \gamma^1,  &&1\hat\otimes \gamma^j \mapsto \gamma^{j+1}.
\end{align*} 
Applying this isomorphism gives the unbounded Kasparov module 
representing the product,
$$ 
 \bigg( C_c(\R,\calA_e) \hat\otimes C\ell_{0,d}, \,(L^2(\R, A_e)\otimes_{A_e} L^2(\R^{d-1},B)) 
 \hat\otimes \bigwedge\nolimits^{\!*}\R^d,\, X_\mathrm{ext}\otimes 1 \hat\otimes \gamma^1 + 
  \sum_{j=1}^{d-1} (1\otimes_\nabla X_{j}) \hat\otimes \gamma^{j+1} \bigg),
$$
with $C\ell_{0,d}$-action generated by $\rho^j(\omega) = e_j\wedge\omega -\iota(e_j)\omega$ and 
$C\ell_{d,0}$-action generated by $\gamma^j(\omega) = e_j\wedge\omega +\iota(e_j)\omega$ for 
$\omega\in\bigwedge^*\R^d$ and $\{e_j\}_{j=1}^d$ the standard basis of $\R^d$.


Next we define a unitary map $L^2(\R, A_e)\otimes_{A_e} L^2(\R^{d-1},B) \to L^2(\R^d,B)$. 
To write this map, it is advantageous to use the isomorphisms 
from Lemma \ref{lemma:Crossed_prod_module_is_std_module},
\begin{align*}
  &L^2(\R,A_e) \cong \ol{C_c(\R,A_e)},  &&L^2(\R^{d-1},B) \cong \ol{C_c(\R^{d-1},B)} 
    &&L^2(\R^{d},B) \cong \ol{C_c(\R^{d},B)},
\end{align*}
with inner-products $(f_1\mid f_2) = (f_1^* \ast f_2)(0)$ and the left-action is the 
extension of left multiplication. We work with the dense submodule
$$
  C_c(\R,\calA_e) \cong C_c(\R,C_c(\R^{d-1},B)) \cong C_c(\R)\otimes C_c(\R^{d-1},B),
$$
which allows us to write down the map
\begin{align*}
   &\varrho: C_c(\R)\otimes C_c(\R^{d-1},B) \otimes_{C_c(\R^{d-1},B)} C_c(\R^{d-1},B) \to 
     C_c(\R^d,B), \\
    &\varrho( f_1\otimes f_2 \otimes_{C_c(\R^{d-1},B)} f_3) = f_1 \otimes f_2\ast f_3 
      \in C_c(\R) \otimes C_c(\R^{d-1},B) \cong C_c(\R^d,B).
\end{align*}
As the action of left-multiplication is non-degenerate and uniformly bounded, this map extends 
to a unitary map when we take the closure in the module norm. It is easy to 
check that 
$$
\varrho\big((X_\mathrm{ext} \otimes 1\otimes 1)(f_1\otimes f_2\otimes_{C_c(\R^{d-1},B)} f_3)\big) 
  = (X_d f_1)\otimes f_2\ast f_3
   = X_d\, \varrho(f_1\otimes f_2\otimes_{C_c(\R^{d-1},B)} f_3).
$$ 
For the left-action, we see that for $g_1\otimes g_2 \in C_c(\R)\otimes C_c(\R^{d-1},B)$,
\begin{align*}
  \varrho\!\left( \pi(g_1\otimes g_2)(f_1\otimes f_2\otimes_{C_c(\R^{d-1},B)}f_3 )\right)
     &= \varrho\!\left( (g_1\ast f_1)\otimes (g_2\ast f_2)\otimes_{C_c(\R^{d-1},B)}f_3 \right) \\
     &= (g_1\ast f_1)\otimes (g_2\ast f_2)\ast f_3 \\
     &= (g_1\ast f_1)\otimes g_2\ast (f_2\ast f_3) \\
     &= (g_1\otimes g_2)(f_1 \otimes f_2\ast f_3) \\
     &= \pi(g_1\otimes g_2) \varrho(f_1\otimes f_2\otimes_{C_c(\R^{d-1},B)}f_3 ).
\end{align*}
Again, as the left-action is uniformly bounded, the result extends to show that the 
left action is compatible with the unitary map on the whole module. 
Finally we note that on the dense submodule, the operator $1\otimes_\nabla X_j$ 
from Equation \eqref{eq:twisted_dirac} has 
the form
$$
  (1\otimes_\nabla X_j)(f_1\otimes f_2 \otimes_{C_c(\R^{d-1},B)}f_3) 
    = f_1\otimes X_jf_2 \otimes_{C_c(\R^{d-1},B)} f_3 + 
       f_1\otimes f_2 \otimes_{C_c(\R^{d-1},B)} X_jf_3.
$$
Therefore we compute
\begin{align*}
   \varrho\!\left( (1\otimes_\nabla X_j)(f_1\otimes f_2 \otimes_{C_c(\R^{d-1},B)}f_3) \right) 
    &= f_1 \otimes \left((X_j f_2)\ast f_3 + f_2\ast(X_j f_3) \right) \\
    &= f_1 \otimes X_j\left( f_2\ast f_3\right) \\
    &= X_j \,\varrho(f_1\otimes f_2 \otimes_{C_c(\R^{d-1},B)}f_3)
\end{align*}
as the operator $X_j$ is a derivation on $C_c(\R^{d-1},B)$. Taking closures, we have that 
$1\otimes_\nabla X_j \mapsto X_j$ under $\varrho$.
To summarise, the unbounded Kasparov module representing the product 
is unitarily equivalent to 
\begin{equation} \label{eq:prod_module_simplified}
  \bigg( C_c(\R^d,B)\hat\otimes C\ell_{0,d},\, L^2(\R^d, B)
 \hat\otimes \bigwedge\nolimits^{\!*}\R^d,\, X_d\hat\otimes \gamma^1 
 + \sum_{j=1}^{d-1} X_{j}  \hat\otimes \gamma^{j+1}  \bigg)
\end{equation}
with left and right Clifford actions as before. The only difference 
between our product module and the bulk module  $\lambda_d$ 
from Proposition \ref{prop:crossed_prod_Kas_mod}
is the labelling of the Clifford basis. 
The map $\eta(\gamma^j)= \gamma^{\sigma(j)}$ and 
$\eta(\rho^j) = \rho^{\sigma(j)}$ for $\sigma(j) = (j-1)\,\mathrm{mod}\,d$ 
is an isomorphism of Clifford algebras that may reverse the 
orientation of the algebra.
Taking the canonical orientation 
 $\omega_{C\ell_{0,d}}=\rho^1\cdots\rho^d$ of $C\ell_{0,d}$, 
$$
  \eta(\omega_{C\ell_{0,d}}) = \rho^d \rho^1\cdots\rho^{d-1} = (-1)^{d-1} \rho^1\cdots \rho^d 
   = (-1)^{d-1} \omega_{C\ell_{0,d}},
$$
similarly $\gamma^j$ and $C\ell_{d,0}$. Using~\cite[{\S}5: Theorem 3]{Kasparov80}, such 
a map on Clifford algebras will send the $KK$-class of the Kasparov module of Equation 
\eqref{eq:prod_module_simplified} to its inverse if $\eta(\omega)=-\omega$ or 
will leave the class invariant if $\eta(\omega)=\omega$.
Hence, at the level 
of $KK$-classes, $[\mathrm{ext}]\hat\otimes_{A_e} [\lambda_{d-1}] = (-1)^{d-1}[\lambda_d]$ 
as required.
\end{proof}

\subsection{Pairings and the bulk-edge correspondence}
Given a real or complex $K$-theory element $[x]\in KO_j(B\rtimes_\theta\R^d)$ 
(we will consider the real setting as the complex case is simpler),
the unbounded Kasparov module $\lambda_d$ 
from Propostion \ref{prop:crossed_prod_Kas_mod} 
gives a map
$$
  KO_j(B\rtimes_\theta\R^d) \times KKO^d(B\rtimes_\theta\R^d,B) \to 
    KO_{j-d}(B)
$$
via the internal product. Theorem \ref{thm:gen_bulkedge} implies that 
we may decompose this bulk pairing as the product
$(-1)^{d-1}[x] \hat\otimes_{A_b}(\left [\mathrm{ext}] \hat\otimes_{A_e} [\lambda_{d-1}] \right)$.
The associativity of the Kasparov product ensures that this product can 
be expressed as 
$(-1)^{d-1}\left([x] \hat\otimes_{A_b} [\mathrm{ext}]\right) \hat\otimes_{A_e} [\lambda_{d-1}]$,
which is now a pairing over the edge algebra
$$
  KO_{j-1}(B\rtimes_\theta \R^{d-1}) \times KKO^{d-1}(B\rtimes_\theta\R^{d-1},B) 
    \to KO_{j-d}(B).
$$
The equality of these bulk and boundary pairings is the bulk-edge correspondence. 

If $d=1$, then $A_b = B\rtimes \R$, $A_e = B$ and 
$[x]\hat\otimes_{A_b} [\lambda_1] = \partial[x] \in KO_{j-1}(B)$
as $\lambda_1$ represents the extension class of the Wiener--Hopf extension. 
As the 
boundary map of the Wiener--Hopf extension is an isomorphism in $K$-theory, 
the bulk pairing will be non-trivial only if the boundary $K$-theory class is 
non-trivial.

The bulk-edge correspondence is usually associated to topological states 
of matter though we note that there may be other applications of such 
a result. We may iterate the bulk-edge correspondence to say that 
pairings of continuous crossed products $B\rtimes_\theta\R^d$ can 
be expressed as (up to a sign) a pairing of a crossed product 
of any order $B\rtimes\R^k$ (though this will place
restrictions on the allowed twists). At the level of $KK$-classes, this result follows
from the Connes--Thom isomorphism, but our explicit formulae allow us 
to derive concrete formulas for index pairings in terms of the physical operators.

Theorem \ref{thm:gen_bulkedge} immediately implies that there 
is also a bulk-edge correspondence of the semifinite index pairing.
This remark is particularly useful for complex algebras as the bulk 
and edge invariants may be computed by the local formulas from 
Theorem \ref{thm:higher_dim_chern_number_odd} and 
\ref{thm:higher_dim_chern_number_even}. Hence we obtain an equality 
of cyclic pairings for twisted dynamical systems $(B,\R^d,\alpha,\theta)$ 
with a faithful and invariant tracial weight $\tau_B$ on $B$. Furthermore, 
we know that the range of these cyclic pairings is countably generated 
and often discrete. Hence the 
bulk-edge result extends to a wide range of potential examples 
and index pairings.


\section{Applications to disordered quantum systems and topological phases} \label{sec:top_phases_application}
Here we link our mathematical framework back to continuous models of 
free-fermionic quantum mechanical systems and their topological 
invariants. Our aim is to show how the framework for modelling such 
systems developed in~\cite{NB90, BelGapLabel,Bellissard94} naturally 
fits into our constructions and results from the previous sections. 
We also 
make some comments on localisation and the stability of the 
index pairings in the strong disorder regime.

\subsection{Review: Disordered Hamiltonians and twisted crossed products}
We model a particle in $\R^d$ subject to a uniform magnetic 
field perpendicular to the sample. There is a choice of magnetic 
potential $A$, where $B = \mathrm{d}A + A\wedge A$ is the magnetic 
field. We take $A=(A_1,\ldots,A_d)$ such that 
$A_j\in L^2_\text{loc.}(\R^d)$ and differentiable with
$$  
\frac{\partial}{\partial x_j} A_k - \frac{\partial}{\partial x_k} A_j = B_{j,k} = \text{const.} 
$$
for all $j,k\in\{1,\ldots,d\}$. The case of a magnetic field 
continuously depending on $x$ is also possible (cf. Example \ref{ex:mag_twist}), 
though we will consider 
constant field strength for simplicity. We direct the reader to 
\cite{BLM13, LMR10, MPR07} for a more detailed study on magnetic fields 
and twisted crossed products. The Schr\"{o}dinger operator 
is given by
$$  
H_0 = \frac{1}{2m^*} \sum_{j=1}^d \!\left(-i\hbar\frac{\partial}{\partial x_j} - eA_j \right)^2,   
$$
where $m^*$ is the effective mass of the particle. We choose 
units such that $m^* = \frac{\hbar}{2}$ and introduce the operators
$K_j = -i\frac{\partial}{\partial x_j} - \frac{e}{h}A_j$ for $j\in\{1,\ldots,d\}$. 
We choose the 
symmetric gauge and define 
$A_j = -\frac{1}{2}\sum\limits_{k=1}^d B_{j,k}x_k$ for $j=1,\ldots,d$, 
where $B_{j,k}$ is antisymmetric and real. We introduce the parameter 
$\theta_{j,k}$ so that we can rewrite 
$$  
K_j = -i\frac{\partial}{\partial x_j} - \sum_{k=1}^d \theta_{j,k}X_k \quad \text{and} \quad \sum_{j=1}^d K_j^2 = \frac{2m^*}{\hbar^2}H_0 = H_0. 
$$

\begin{example}[Quantum Hall Hamiltonian]
In the case where $d=2$, our Hamitonian is given in the symmetric gauge as
$$ 
 H_0 = \left(-i\frac{\partial}{\partial x_1} +\theta X_2\right)^2 + \left(-i\frac{\partial}{\partial x_2} -\theta X_1\right)^2, 
$$
where $\theta\in\R$ represents the magnetic flux through a unit cell. 
We recognise this Hamiltonian as the $2$-dimensional Landau 
Hamiltonian used to model the quantum Hall effect.
\end{example}

The presence of the magnetic field means that $H_0$ does not 
commute with ordinary translation operators $S_a$, where 
$(S_a\psi)(x) = \psi(x-a)$ for $\psi\in L^2(\R^d)$ and 
$x,a\in\R^d$. However, we may define the so-called 
\emph{magnetic translations} $U_a$ such that in the symmetric 
gauge $(U_a\psi)(x) = e^{-i\theta (x\wedge a)}(x-a)$ for 
$\psi\in L^2(\R^d)$, where $\theta(x\wedge a)=\sum_{j,k=1}^d\theta_{j,k}x_j a_k$. 
We note that $\theta(x\wedge x) = 0$ and $\theta(x\wedge y) = -\theta(y\wedge x)$.
One checks that $[U_a,K_j]=0$ on $\Dom(K_j)$ for any $a\in\R^d$ and $j\in\{1,\ldots,d\}$ 
(see~\cite{Zak64} for more details on magnetic translations for general gauge choices).

We wish to consider a system with disorder or impurities. Following 
Bellissard and co-authors~\cite{BelGapLabel,NB90}, such effects can be encoded by 
the hull. We let $H = H_0+V$ with $V$ a potential coming from 
an essentially bounded, real-valued and measurable 
function on $\R^d$. Consider the set 
$$
  \Omega = \ol{\{ U_a V U_{-a}\,:a\in\R^d\}},
$$
with compact closure in the weak operator topology. Clearly $\Omega$ is endowed 
with a twisted action of $\R^d$ by magnetic translations, denoted 
by $\{T_{a}\}_{a\in\R^d}$. By considering the possible translates 
of the potential $V$, we obtain a family of Hamiltonians 
$\{H_\omega\}_{\omega\in\Omega}$ representing a disordered or 
aperiodic system.

\begin{prop}[\cite{BelGapLabel}, {\S}2.4. Also see~\cite{BourneThesis}, Corollary 3.2.11] \label{cor:Hamiltonian_hull_is_potential_hull}
Denote by $V_\omega$ the bounded function representing the 
point $\omega\in\Omega$. Then there is a Borel function 
$v$ on $\Omega$ such that $V_\omega(x) = v(T_{-x}\omega)$ 
for almost all $x\in\R^d$ and all $\omega\in\Omega$. If in 
addition $V$ is uniformly continuous and bounded, $v$ is continuous.
\end{prop}
The proposition shows that if our disordered potential $V$ 
is uniformly continuous and bounded, we 
can associate a continuous function $v\in C(\Omega)$, a unital $C^*$-algebra. 
Furthermore, $C(\Omega)$ comes with a twisted action of $\R^d$ via the action 
on $\Omega$. Hence we may consider the twisted crossed 
product $C(\Omega)\rtimes_{\theta}\R^d$. We start with the 
dense subalgebra $\calA = C_c(\R^d,C(\Omega))$ and denote functions 
$f\in C_c(\R^d,C(\Omega))$ by $f(x;\omega)$.
We use 
the symmetric gauge which determines the twist $\theta$ and gives rise 
to the $\ast$-algebra structure,
\begin{align*}
  &(f_1 \ast f_2)(x;\omega) = \int_{\R^d}\! e^{i\theta(x\wedge y)} f_1(y;\omega) 
  f_2(x-y;T_{-y}\omega)\,\mathrm{d}y, 
   &&f^{*}(x;\omega) = \overline{ f(-x;T_{-x}\omega) }.
\end{align*}
For a fixed $\omega\in\Omega$, we can represent $\calA$
on $L^2(\R^d)$ by the map $\pi_\omega$, where
$$  
[\pi_\omega(f)\psi](x) = \int_{\R^d}\!e^{-i\theta(x\wedge y)} f(y-x;T_{-x}\omega)\psi(y)\,\mathrm{d}y 
$$
for all $\psi\in L^2(\R^d)$. A computation shows that 
$\pi_\omega$ is a $\ast$-algebra homomorphism. 
Furthermore, the representation satisfies 
the covariance condition 
\begin{equation}  
U_a \pi_\omega(f) U_{-a} = \pi_{T_a\omega}(f) 
\label{eq:covariance}
\end{equation}
for all $\omega\in\Omega$ and $f\in C_c(\R^d,C(\Omega))$. 
Hence we obtain a family of representations of the resulting 
crossed product completion $A=C(\Omega)\rtimes_\theta\R^d$.
A key result of~\cite{BelGapLabel} is the following.

\begin{thm}[\cite{BelGapLabel}, Theorem 6] 
\label{thm:H_w_represented_by_algebra}
Take $H= \sum_j K_j^2 +V$ acting on $L^2(\R^d)$ with hull $\Omega$. 
For each $z$ in the resolvent set of $H$ 
and $x\in\R^d$ there is an element $R(z;x)\in A$ 
such that for all $\omega\in\Omega$, 
$\pi_\omega[R(z;x)] = (z-H_{T_{-x}\omega})^{-1}$.
\end{thm}
Therefore we can take the algebra of observables of a disordered 
magnetic Hamiltonian to be $C(\Omega)\rtimes_\theta \R^d$. Hence we
may consider topological properties of the physical system by 
studying the topology of the crossed product algebra.

\subsubsection{Measures and traces}
We now suppose that the disorder space of configurations $\Omega$ 
has a probability measure $\bP$ that is invariant 
under the $\R^d$-action and $\mathrm{supp}(\bP)=\Omega$. 
The measure $\bP$ induces a trace $\tau_\bP$ on $C(\Omega)$ by 
integration that is 
semifinite, norm lower-semicontinuous (by Fatou's lemma applied to $\bP$) and 
faithful by the support assumption on $\bP$.

We now extend $\tau_\bP$ to an unbounded trace on the crossed 
product $C(\Omega)\rtimes_\theta\R^d$. For $f\in C_c(\R^d,C(\Omega))$ and 
$f\geq 0$, we define 
$$  \calT(f) = \tau_\bP[f(0)] = \int_{\Omega} f(0;\omega)\,\mathrm{d}\bP(\omega). $$

Using \cite{LN04}, and Propositions \ref{prop:adjointable_repn_of_crossed_prod}, \ref{prop:crossed_prod_Kas_mod}, we deduce that $\calT$ satisfies our
technical hypotheses. 
\begin{lemma} \label{lemma:trace_on_algebra_is_faithful}  \label{lemma:trace_is_semifinite_and_lsc}
The functional $\calT$ is a faithful semifinite norm lower-semicontinous 
trace with $\calA\subset \Dom(\calT)$. Furthermore, $\calT$ extends to a semifinite 
trace on $(C(\Omega)\rtimes_\theta \R^d)'' \subset \calB[L^2(\R^d)\otimes L^2(\Omega, \bP)]$.
\end{lemma}

\begin{prop}[\cite{Bellissard94}, Proposition 1]  \label{prop:traces_equal}
Let $f\in\calA_+$. If $\bP$ is an ergodic measure, then 
for almost all $\omega\in\Omega$,
$$   
\calT(f) = \Tr_{\mathrm{Vol}}[\pi_\omega(f)], 
$$
with $\Tr_\mathrm{Vol}$ the trace per unit volume on $L^2(\R^d)$.
\end{prop}
\begin{proof}
Given $g\in\calA$, we know that 
$$ 
[\pi_\omega(g)\psi](x) = \int_{\R^d}\! e^{-i\theta(x\wedge y)}g(y-x;T_{-x}\omega)\psi(y)\,\mathrm{d}y 
$$
so $\pi_\omega(g)$ is an integral operator with kernel 
$k_\omega(x,y) = e^{-i\theta(x\wedge y)}g(y-x;T_{-x}\omega)$.
We let $\Lambda_j$ be a sequence of increasing sets that converge 
in the appropriate sense to $\R^d$, e.g. $\Lambda_j = [-j,j]^d$. 
Because $\Lambda_j$ is bounded and $k_\omega(x,y)$ is continuous, 
$\pi_\omega(g)$ is Hilbert-Schmidt on $L^2(\Lambda_j)$ 
by~\cite[Theorem VI.23]{ReedSimon1} for any 
$g\in C_c(\R^d,C(\Omega))$. Therefore we can say that the product 
$\pi_\omega(g^*g)$ is trace-class 
by~\cite[Theorem VI.22, part (h)]{ReedSimon1} 
for $g\in C_c(\R^d,C(\Omega))$. We can take the trace $\Tr_{\Lambda_j}$ 
by integrating along the diagonal~\cite[Theorem 3.9]{SimonTraceIdeals}. 
Computing the trace of $f= g^*g$,
\begin{align*}
  \Tr_{\Lambda_j}[\pi_\omega(f)] = \int_{\Lambda_j}k_\omega(x,x)\,\mathrm{d}x 
   &= \int_{\Lambda_j}e^{-i\theta(x\wedge x)}f(x-x;T_{-x}\omega)\,\mathrm{d}x 
   = \int_{\Lambda_j}f(0;T_{-x}\omega)\,\mathrm{d}x.
\end{align*}
As the action of $\R^d$ by $T$ on $\Omega$ is $\bP$-measure 
preserving, a continuous version of Birkhoff's Ergodic Theorem 
in higher dimensions~\cite[Section 4]{XZ79} gives that
\begin{align*}
  \Tr_{\mathrm{Vol}}[\pi_\omega(f)] 
  &= \lim_{j\to\infty} \frac{1}{|\Lambda_j|} \Tr_{\Lambda_j}[\pi_\omega(f)] 
  = \lim_{j\to\infty} \frac{1}{|\Lambda_j|}\int_{\Lambda_j}\!f(0;T_{-x}\omega)\,\mathrm{d}x 
  = \int_{\Omega}f(0;\omega)\,\mathrm{d}\bP(\omega) = \calT(f)   
\end{align*} 
for almost all $\omega$.
\end{proof}

\subsection{Invariants of topological systems} \label{subsec:disordered_pairings}
Because the hull $\Omega$ is compact and Hausdorff, $C(\Omega)$ is a separable 
$C^*$-algebra and our general theory applies to the example. We note that a disordered 
family of Hamiltonians $\{H_\omega\}_{\omega\in\Omega}$ affiliated to $C(\Omega)\rtimes_\theta \R^d$ 
can be considered as a single element $H$ affiliated to $C(\Omega)\rtimes_\theta \R^d$. Often 
we will work with $H$ considered as a family of Hamiltonians, though occasionally it will be 
useful to consider a particular configuration $H_\omega$ acting on 
$L^2(\R^d)$ for $\omega\in\Omega$.

We may also be interested in the case when the Hamiltonian 
$H$ satisfies a $CT$-type symmetry. That is, $H$ is time reversal and/or 
particle-hole and/or chiral symmetric. If we are interested in 
anti-linear symmetries on $H$, then this introduces a Real structure 
on the crossed product algebra $C(\Omega)\rtimes_\theta \R^d$ via 
complex conjugation. As such we are interested in the topological invariants 
of the real subalgebra of $C(\Omega)\rtimes_\theta \R^d$ that is 
compatible with the symmetries under consideration. Such an algebra is still a real crossed 
product $C(\Omega)\rtimes_\theta \R^d$, but the anti-linear symmetries 
present may put limitations on the type of magnetic 
fields and twists $\theta$ that are possible.

\begin{prop}[\cite{FM13, Thiang14, Kellendonk15, Kubota15b, BCR15}]
\label{prop:symmetries_give_k_theory}
If the disordered family of Hamiltonians $\{H_\omega\}_{\omega\in\Omega}$ satisfies a $CT$-type symmetry and retains 
a spectral gap for all $\omega\in\Omega$, then we can associate a class in $KO_n(C(\Omega)\rtimes_\theta\R^d)$ or 
$K_n(C(\Omega)\rtimes_\theta\R^d)$, where $n$ is determined by the symmetry.
\end{prop}

We can pair the $K$-theory class from Proposition 
\ref{prop:symmetries_give_k_theory} with our complex or real 
unbounded Kasparov module
\begin{equation} \label{eq:disorder_Kas_mod}
\bigg( \calA\hat\otimes C\ell_{0,d}, \, L^2(\R^d,C(\Omega))_{C(\Omega)} \hat\otimes 
  \bigwedge\nolimits^{\!*}\R^d, \, \sum_{j=1}^d X_j \hat\otimes \gamma^j \bigg).
\end{equation}
Furthermore, the invariant trace $\tau_\bP$ on $C(\Omega)$ allows us to 
construct the 
semifinite spectral triple
\begin{equation*} \label{eq:disorder_semifinite_spec_trip}
 \bigg( \calA \hat\otimes C\ell_{0,d}, \, L^2(\R^d)\otimes L^2(\Omega,\bP) 
  \hat\otimes \bigwedge\nolimits^{\!*}\R^d, \, \sum_{j=1}^d X_j \hat\otimes \gamma^j \bigg)
\end{equation*}
relative to the trace $\calT$ and von Neumann algebra $(C(\Omega)\rtimes_\theta \R^d)''$. This spectral triple
 is smoothly summable with spectral dimension $d$.
 
\begin{remark}
Let us briefly comment on the possible $K$-theoretic phases encoded by the group 
$KO_n(C(\Omega)\rtimes_\theta\R^d)$ or $K_n(C(\Omega)\rtimes_\theta\R^d)$ 
from Proposition \ref{prop:symmetries_give_k_theory}. We note that the (twisted) $d$-fold 
Connes--Thom isomorphism can be implemented by the Kasparov product with 
the Kasparov module from Equation \eqref{eq:disorder_Kas_mod}. Hence we have the 
explicit isomorphism $KO_n(C(\Omega)\rtimes_\theta \R^d) \cong KO_{n-d}(C(\Omega))$ 
(similarly complex).

Let us consider the computation of $KO_{n-d}(C(\Omega))$ in a few simple cases.
If $\Omega$ is contractible, then $KO_{n-d}(C(\Omega))\cong KO_{n-d}(\R)$. 
If $\Omega$ a totally disconnected space, then by the continuity of the $K$-functor 
$KO_{n-d}(C(\Omega)) \cong C(\Omega, KO_{n-d}(\R))$.
Outside of these examples, the computation of the $K$-theory of the configuration 
space is much more involved. For continuous hulls $\Omega$ that come 
from certain tilings, Savinien and Bellissard construct a spectral sequence that
converges to $K_{n-d}(C(\Omega))$ and whose page-2 is isomorphic to the 
integer \v{C}ech cohomology of $\Omega$~\cite{BelSav}. Hence the computation of 
$KO_n(C(\Omega)\rtimes_\theta \R^d)$ (or complex) in general is a highly non-trivial problem 
which we will not consider here.
\end{remark}

\subsubsection{Complex invariants}
For a disordered Hamiltonian $H=\{H_\omega\}_{\omega\in\Omega}$ without additional symmetries, 
the $K$-theory 
class of interest is the Fermi projection $P_\mu = \chi_{(-\infty,\mu]}(H)$. 
If $P_\mu \in \calA_\mathrm{Sob}$, the Sobolev algebra 
from Section \ref{sec:localisation_general}, then the index pairing with $P_\mu$ is 
well-defined. We will show in Section \ref{subsec:qc_and_localisation} that 
$P_\mu \in \calA_\mathrm{Sob}$ when the Fermi energy $\mu$ is in a region of 
dynamical localisation (using results from~\cite{AENSS}). If $\mu$ is in a gap 
in the spectrum, then $P_\mu \in\calA$, the smooth subalgebra of $C(\Omega)\rtimes_\theta\R^d$ 
and $[P_\mu] \in K_0(\calA)$.

If the Hamiltonian has a chiral symmetry, then $H$ is invertible and there is a self-adjoint 
complex unitary $R_c$ such that $R_c H R_c^* = -H$. Diagonalising 
$R_c$ if necessary, this implies that $H$ can be written as 
an off-diagonal matrix. We can use the chiral symmetry to define 
the so-called Fermi unitary $U_\mu$, where
$$
   1-2P_\mu = \mathrm{sgn}(H)= \begin{pmatrix} 0 & U_\mu^* \\ U_\mu & 0 \end{pmatrix}.
$$
If $H$ is invertible, then 
$U_\mu = \frac{1}{2}(1-R_c)(1-2P_\mu)\frac{1}{2}(1+R_c) \in \calA^\sim$ 
and we obtain a $K$-theory class $[U_\mu] \in K_1(\calA)$
 provided $R_c \in \calA^\sim$. We can also consider 
a more general setting where $P_\mu\in \calA_\mathrm{Sob}$ and 
$R_c \in \calA_\mathrm{Sob}^\sim$; then 
$U_\mu = \frac{1}{2}(1-R_c)(1-2P_\mu)\frac{1}{2}(1+R_c) \in \calA_\mathrm{Sob}^\sim$. 
Other assumptions are also possible to ensure that $U_\mu$ is a 
well-defined unitary in $\calA_\mathrm{Sob}^\sim$.

If $P_\mu \in M_q(\calA_\mathrm{Sob})$ and $U_\mu \in M_q(\calA_{\mathrm{Sob}}^\sim)$, then our semifinite 
pairing gives the cyclic expressions
\begin{align*}  
\langle [U_\mu], [X] \rangle &= C_d \sum_{\sigma\in S_d}(-1)^\sigma\, 
 (\Tr_{\C^q}\otimes\calT) \bigg(\prod_{i=1}^d U_\mu^* \,\partial_{\sigma(i)}U_\mu \bigg), 
    \quad d \text{ odd}, \\
  \langle [P_\mu], [X]\rangle &=  \frac{(-2\pi i)^{d/2}}{(d/2)!} \sum_{\sigma\in S_d} 
  (-1)^\sigma\, (\Tr_{\C^q}\otimes\calT) \bigg(P_\mu \prod_{i=1}^d 
  \partial_{\sigma(i)}P_\mu \bigg),
    \quad d \text{ even}
\end{align*}
where, we recall, $C_{2n+1} = \frac{ 2(2\pi i)^n n!}{(2n+1)!}$ and 
$(\partial_jf)(x;\omega) = x_j f(x;\omega)$ with the property that
$\pi_\omega(\partial_j f)=[X_j,\pi_\omega(f)]$ for all $\omega\in\Omega$ 
and $f\in C_c(\R^d,C(\Omega))$ (and then extended to $\calA_{\mathrm{Sob}}$).

We require that the probability measure $\bP$ on $\Omega$ is invariant under the 
action of magnetic translations. If we also assume that the measure $\bP$ is ergodic under 
the twisted $\R^d$-action, then by Proposition \ref{prop:traces_equal} we can write
\begin{align*}  
\langle [U_\mu], [X] \rangle &= C_d \sum_{\sigma\in S_d}(-1)^\sigma\, 
 (\Tr_{\C^q}\otimes \Tr_{\mathrm{Vol}}) \bigg(\prod_{j=1}^d 
       \pi_\omega(U_\mu)^* [X_{\sigma(j)},\pi_\omega(U_\mu)] \bigg), \quad d \text{ odd}, \\
  \langle [P_\mu], [X]\rangle &=  \frac{(-2\pi i)^{d/2}}{(d/2)!} \sum_{\sigma\in S_d} 
  (-1)^\sigma\, (\Tr_{\C^q}\otimes \Tr_{\mathrm{Vol}}) \bigg(\pi_\omega(P_\mu) 
     \prod_{j=1}^d [X_{\sigma(j)},\pi_\omega(P_\mu)] \bigg),
    \quad d \text{ even}
\end{align*}
for almost all $\omega\in\Omega$. By Theorem \ref{thm:ergodic_loc_pairing}, 
the pairings for ergodic measures are almost surely integer valued and constant in $\omega$. 
Furthermore, the formulas for the pairing now involve a specific 
configuration $H_\omega$ acting on the concrete Hilbert space $L^2(\R^d)$ with 
physical trace $\Tr_\mathrm{Vol}$. Therefore our cyclic formulas can be 
linked to physical phenomema more easily. 
It is shown in~\cite[Chapter 5, 7]{PSBbook} 
that, for $A=C(\Omega)\rtimes_\theta\Z^d$, 
the Chern number formulas can be linked to transport 
coefficients of the linear conductivity tensor of solid state systems. 
The link between lower order coefficients of the conductivity tensor 
and our cyclic pairings in the continuous case has been studied in~\cite{DNLeinBook}.

\subsubsection{Real invariants}  \label{sec:Real_smooth_pairings}
Semifinite index pairings play an important role in characterising complex 
topological phases, but their application to topological phases with anti-linear 
symmetries is somewhat limited. Instead, we will show that the 
Kasparov product and Clifford module valued indices are a more natural tool 
for characterising topological phases. 
We assume for the time being that the Hamiltonian 
has a spectral gap so we can work in the smooth algebra 
$\calA \subset C(\Omega)\rtimes_\theta \R^d$ (the case of the Sobolev 
algebra and the link to dynamical localisation will be considered in Section \ref{Sec:real_loc_pairings}).

We can pair the semifinite spectral triple from 
Equation \eqref{eq:disorder_semifinite_spec_trip} with 
a class in $KO_d(A)$ to obtain
numerical phase labels that take value in $(\tau_\bP)_\ast [KO_{0}(C(\Omega))]$. 
Because we take an 
expectation over the $K$-theory class, the semifinite pairing 
gives `disorder-averaged' numerical invariants.

For non-torsion elements in $KO_{d+4}(A)$, we take the  
product with the Kasparov module from Equation \eqref{eq:disorder_Kas_mod}, 
which gives a class in $KO_4(C(\Omega))$. By the identification 
$KO_4(C(\Omega)) \cong KO_0(\mathbb{H}\otimes C(\Omega))$, 
we can again take an average over the configuration space $\Omega$ provided the 
measure on $\Omega$ is compatible with the quaternionic structure 
that comes from classes in $KO_4$.

For other $K$-theoretic phases that are non-torsion 
and whose product with the bulk Kasparov module lands in 
$KO_k(C(\Omega))$ with $k\neq 0,4$, we can also 
obtain disorder-averaged invariants by first noting that
$$
  KO_{k}(C(\Omega)) \cong  KKO(C\ell_{k,0},C(\Omega)) \cong KKO(\R, C(\Omega)\hat\otimes C\ell_{0,k}) 
$$
and then applying the graded trace on $C(\Omega)\hat\otimes C\ell_{0,k}$. If 
we wish to avoid graded traces, then we can use Bott periodicity in 
$KKO$-theory 
to relate $KKO(\R, C(\Omega)\hat\otimes C\ell_{0,k}) \cong KKO(\R, C(\Omega) \otimes C_0(\R^k))$,
which comes via the product with the Bott class
$$
  \bigg( C\ell_{0,k}, \, C_0(\R^k)_{C_0(\R^k)} \hat\otimes \bigwedge\nolimits^{\! *} \R^k, \, 
    \sum_{j=1}^k x_j \hat\otimes \gamma^j \bigg),
$$
see~\cite[{\S}5]{Kasparov80}. The trace on $C(\Omega) \otimes C_0(\R^k)$ is 
easier to 
understand, but using this trace requires that we have to take another product. 
A concrete expression for this pairing would depend on the specific symmetries
of the Hamiltonian that feed into the construction of the $K$-theory class 
in Proposition \ref{prop:symmetries_give_k_theory}.

Let us now consider phases that come from torsion classes in 
$KO_n(C(\Omega)\rtimes_\theta \R^d)$. Because the semifinite index pairing 
involves taking a trace, it is not well-suited to detecting torsion indices. 
In such circumstances, we instead take the Kasparov product 
of the $K$-theory class of the Hamiltonian with the Kasparov module from Equation 
\eqref{eq:disorder_Kas_mod},
$$
  KO_n(A) \times KKO^d(A,C(\Omega)) \to KKO(C\ell_{n,d},C(\Omega)) \cong
    KO_{n-d}(C(\Omega)).
$$
The class in $KO_{n-d}(C(\Omega))$ is encoded via a Clifford index, analogous to the
approach of \cite{ABS64} and extended in~\cite[Section 2.2]{Schroder}. Using the unbounded representative 
of the Kasparov product, $(C\ell_{n,d},E_{C(\Omega)}, \tilde{X})$, this 
Clifford index is given by the equivalence class of the $C^*$-module $\Ker(\tilde{X})_{C(\Omega)}$ as 
a graded $C\ell_{n,d}$-module (when this makes sense). Analytic formulas for this index 
can be written down in concrete examples of Hamiltonians, 
see~\cite[Section 4.1]{BCR15} for example.

Because the $K$-theory of $C(\Omega)$ is often hard to compute, we can simplify our 
Clifford module valued index by composing with the evaluation map 
$\mathrm{ev}_\omega: C(\Omega)\to \R$, which gives pairings of the form
$$
   KO_n(A) \times KKO^d(A,C(\Omega)) \to KKO(C\ell_{n,d},C(\Omega)) 
   \xrightarrow{ \mathrm{ev}_\omega} KO_{n-d}(\R).
$$
This pairing can more simply be described as the product of 
$[H] \in KO_n(\calA)$ (from Proposition \ref{prop:symmetries_give_k_theory}) 
with the class of the spectral triple $[\lambda_d(\omega)]\in KO^d(\calA)$ that 
comes from the evaluation map (or the direct integral decomposition of the 
semifinite spectral triple from Section \ref{sec:chop-chop}).

Suppose the spectral triple 
$( C\ell_{n,d}, \, \calH, \, \wt{X}_\omega )$ represents the pairing 
$[H] \hat\otimes_A [\lambda_d(\omega)]$ and that $\wt{X}_\omega$ 
graded-commutes with the generators of $C\ell_{n,d}$ (which can always be 
guaranteed without changing the $K$-homology class). The equivalence class of the Clifford 
module $\Ker(\wt{X}_\omega)$  can be written as an analytic 
formula using the index map for  
skew-adjoint Fredholm operators considered in~\cite{AS69}. Suppose that 
$T$ is an odd Fredholm operator on a graded Hilbert space 
$\calH\cong \calH_+ \oplus \calH_-$ and $T$ anti-commutes (graded-commutes) with a 
representation of $C\ell_{n,d}$ on $\calH$. Then $\Ker(T)$ has the 
structure of a graded left $C\ell_{n,d}$-module to which we can be associate the 
analytic index
\begin{equation}  \label{eq:skew_Fred_index}
  \Index_{n,d}(T) = 
  \begin{cases}  \mathrm{dim}_\R \Ker(T_+) - \mathrm{dim}_\R \Ker(T_+^*), & n-d = 0\,\mathrm{mod}\, 8, \\
     \mathrm{dim}_\R \Ker(T_+)\,\,\mathrm{mod}\,2, & n-d = 1\,\mathrm{mod}\, 8, \\
     \mathrm{dim}_\C \Ker(T_+)\, \,\mathrm{mod}\, 2, & n-d= 2\,\mathrm{mod}\, 8, \\
     \mathrm{dim}_\mathbb{H} \Ker(T_+) - \mathrm{dim}_\mathbb{H} \Ker(T_+^*), & n-d = 4\,\mathrm{mod}\, 8, \\
     0, & \text{otherwise}, \end{cases}   \qquad 
     T_\pm :\calH_\pm \to \calH_\mp.
\end{equation}
Considering the example of $\wt{X}_\omega$,  the Clifford module valued index 
$[\Ker(\wt{X}_\omega) ]$ will be non-trivial (where the class is zero 
if $\Ker(\wt{X}_\omega)$ comes from the restriction of a left $C\ell_{n+1,d}$-module) 
if and only if $\Index_{n,d}\!\big( (\wt{X}_\omega)_+ \big)$ is non-zero.  
See~\cite{GSB15} for concrete examples.

Let us also remark that we can also define 
real indices with values in $KO_{n-d}(C(\Omega))$ using the 
local-global principle, 
provided that the $\R^d$ action is ergodic. 
The most useful way to view this index is as the class of a right $C(\Omega)$-module 
(so sections of a real vector bundle over $\Omega$) with a left $C\ell_{n,d}$ action.

\subsubsection{Example: The disordered Kane--Mele model}
We consider the famous Kane--Mele model for two-dimensional topological insulators with 
fermionic time-reversal symmetry. The Hamiltonian of interest is
$$
   H_{KM}^\omega = \begin{pmatrix} h_\omega & g^* \\ g & \calC h_\omega \calC \end{pmatrix},
$$
where $h_\omega$ is a self-adjoint operator acting on $L^2(\R^2,\C^n)$, $g$ is playing the 
role of the Rashba 
coupling in the continuous setting and $\calC$ is (point-wise) complex conjugation on $L^2(\R^2,\C^n)$. The time 
reversal involution we consider is given by $R_T = \begin{pmatrix} 0 & -\calC \\ \calC & 0 \end{pmatrix}$. 
We see that for $H_{KM}^\omega$ to be time-reversal symmetric, then we  require 
$g^* = -\calC g \calC$. In fact for $K$-theoretic purposes, 
it would also be sufficient for $g$ to be sufficiently 
bounded by $h_\omega$ so that there is a 
homotopy (in the resolvent topology) 
to a Hamiltonian with $g=0$.

We assume that the disorder on $h_\omega$ is such that $\{h_\omega\}_{\omega\in\Omega}$ 
is affiliated to $C(\Omega, M_{n}(\C))\rtimes \R^2$ (e.g. $\Omega$ is the hull of a non-periodic 
but absolutely continuous potential $V$) and so $H_{KM}$ is affiliated to $C(\Omega, M_{2n}(\C))\rtimes \R^2$. 
We restrict to the time-reversal invariant and real subalgebra $A_{KM} = C(\Omega, M_{2n}(\R))\rtimes \R^2$. 

Let us assume for the time being that $H_{KM}^\omega$ has a gap for all $\omega\in\Omega$. 
Then the Fermi projection $P_\mu \in C(\Omega, M_{2n}(\R))\rtimes \R^2$. If we incorporate 
the extra structure $R_T P_\mu R_T^* = P_\mu$ and $R_T^2=-1$, we obtain a projection 
compatible with a quaternionic structure and
so a  class $[P_\mu] \in KO_4(A_{KM})$. Following~\cite{BCR15} 
this class is represented by the Kasparov module 
$$
   \left( C\ell_{4,0}, \,  P_\mu (A_{KM})^{\oplus 2}_{A_{KM}}, \, 0 \right),
$$
where the $C\ell_{4,0}$-action comes from the equivalence between 
quaternionic spaces and $C\ell_{4,0}$-modules~\cite[Appendix B]{FM13}.
We pair this class with our unbounded Kasparov module
$$
  \lambda_2 =  \bigg( C_c(\R^2, C(\Omega,M_{2n}))\hat\otimes C\ell_{0,2}, \, 
    L^2(\R^2, C(\Omega, M_{2n}))_{C(\Omega, M_{2n})}  \hat\otimes \bigwedge\nolimits^{\!*}\R^2, \, 
   X= \sum_{j=1}^2 X_j  \hat\otimes \gamma^j \bigg),
$$
where, if we complexify this module, then we have that 
$$
   (R_T\otimes 1_{\bigwedge^* \C^d})X (R_T \otimes 1_{\bigwedge^* \C^d})^* = \sum_{j=1}^2  R_T(X_j\otimes 1_{\C^{2n}})R_T^* \hat\otimes \gamma^j 
    = \sum_{j=1}^2 \calC X_j \calC \otimes 1_{\C^{2n}} \hat\otimes \gamma^j = X.
$$
The product of the two Kasparov modules gives the following,
$$
  \bigg( C\ell_{4,2}, \, P_\mu (L^2(\R^2, C(\Omega, M_{2n}))_{C(\Omega,M_{2n})} )^{\oplus 2} 
   \hat\otimes \bigwedge\nolimits^{\!*}\R^2, \, 
   \sum_{j=1}^2 P_\mu(X_j \otimes 1)P_\mu \hat\otimes \gamma^j \bigg).
$$
The topological information of interest is obtained in the class of $\Ker(P_\mu X P_\mu )$ considered 
as a Clifford module over $C\ell_{4,2}$.

By composing our pairing with the evaluation map 
at  $\omega\in\Omega$, we can use the indices 
for skew-adoint Fredholm operators considered in~\cite{AS69} to obtain analytic formulas for the map
$$
  KO_4(A_{KM}) \times KKO^2(A_{KM}, C(\Omega, M_{2n})) \to KO_2(C(\Omega,M_{2n})) 
   \xrightarrow{\mathrm{ev}_\omega} KO_2(\R) \cong \Z_2.
$$
Taking this composition, we have that 
\begin{align*}
   \langle [P_\mu], [\lambda_2] \rangle (\omega) &= 
     \mathrm{dim}_\C \Ker \big(\pi_\omega(P_\mu)( X_1 \otimes 1 \hat\otimes \gamma^1 + X_2 \otimes 1 \hat\otimes \gamma_2)_+ 
        \pi_\omega(P_\mu) \big)  \,\mathrm{mod}\, 2 \\
     &=  \mathrm{dim}_\C \Ker\big(\pi_\omega(P_\mu)( X_1 \otimes 1  + iX_2 \otimes 1) \pi_\omega(P_\mu) \big)\, \mathrm{mod}\, 2,
\end{align*}
where in the last line we have chosen particular Clifford generators. We recognise this index as 
analogous to the indices considered in~\cite{SchulzBaldes13b, KK15}. 

It is shown in~\cite{KK15} that for the 
Kane--Mele model (without disorder), the defined analytic index 
agrees with the Kane--Mele invariant and is non-trivial. 
This result is proved in the discrete setting, but 
can be linked to our framework via the Bloch--Floquet transform for periodic potentials, see~\cite[Chapter XIII.16]{ReedSimon4}. 
We also remark that the operator $\pi_\omega(P_\mu)( X_1 \otimes 1  + iX_2 \otimes 1) \pi_\omega(P_\mu)$ 
continues to be Fredholm if $P_\mu\in\calA_\mathrm{Sob}$. Hence, our $\Z_2$-valued index 
is still well-defined as a pairing over the Sobolev algebra. Then applying results from 
Section \ref{subsec:qc_and_localisation} and \ref{Sec:real_loc_pairings}, our formula is 
still well-defined  in regions of dynamical localisation.

We note that if the Rashba coupling is zero, we have that $[\sigma_z, H^\omega_{KM}] = 0$ for all 
$\omega\in\Omega$. Therefore there is a decomposition of $P_\mu$ into $P_\mu^\pm$ corresponding 
to the $+1$ and $-1$ eigenspaces of the spin operator $\sigma_z$. From this point we can directly 
adapt results from~\cite{SchulzBaldes13b, Kellendonk16} to simplify the computation of the mod 2 index. Namely, 
provided that $P_\mu^\pm \in \calA_\mathrm{Sob}$, then using that the time-reversal involution $R_T$
is such that $R_T{\pi_\omega(P_\mu^\pm)}R_T^* = \pi_\omega(P_\mu^\mp)$, we compute with $F$ the phase 
of $X_1\otimes 1 + i X_2\otimes 1$ and such that $R_TF R_T^* = F^*$,
\begin{align*}
  \mathrm{dim}_\C \Ker(\pi_\omega(P_\mu) F \pi_\omega(P_\mu)) &= \mathrm{dim}_\C \Ker(\pi_\omega(P_\mu^+) F \pi_\omega(P_\mu^+)) 
            + \mathrm{dim}_\C \Ker(\pi_\omega(P_\mu^-) F \pi_\omega(P_\mu^-))  \\
    &\hspace{-0cm}= \Index(\pi_\omega(P_\mu^+) F \pi_\omega(P_\mu^+)) +  
      \mathrm{dim}_\C \Ker(\pi_\omega(P_\mu^+) F^* \pi_\omega(P_\mu^+)) \\ 
       &\hspace{5.5cm}+ \mathrm{dim}_\C \Ker(\pi_\omega(P_\mu^-) F \pi_\omega(P_\mu^-))  \\
    &\hspace{-0cm}= \Index(\pi_\omega(P_\mu^+) F \pi_\omega(P_\mu^+)) + 
       \mathrm{dim}_\C \Ker(R_T \pi_\omega(P_\mu^+) F^* \pi_\omega(P_\mu^+) R_T^*) \\
         &\hspace{6cm}+ \mathrm{dim}_\C \Ker(\pi_\omega(P_\mu^-) F \pi_\omega(P_\mu^-)) \\
    &\hspace{-0cm}= \Index(\pi_\omega(P_\mu^+) F \pi_\omega(P_\mu^+)) + 
      2 \mathrm{dim}_\C \Ker(\pi_\omega(P_\mu^-) F \pi_\omega(P_\mu^-))
\end{align*}
and so $\mathrm{dim}_\C \Ker(P_\mu F P_\mu)\,\mathrm{mod}\,2 = \Index(\pi_\omega(P_\mu^+) F \pi_\omega(P_\mu^+))\,\mathrm{mod}\,2$ 
(the same formula is also true for $\pi_\omega(P_\mu^-)$). 
Because we have assumed $P_\mu^\pm \in \calA_\mathrm{Sob}$, we can use the cyclic formula for the index pairing to conclude that,
$$
  \langle [P_\mu], [\lambda_2 ] \rangle (\omega) = 
    -2\pi i \, \calT\!\left( \pi_\omega(P_\mu^\pm) \big[ [X_1,\pi_\omega(P_\mu^\pm)], [X_2, \pi_\omega(P_\mu^\pm)] \big] \right) 
    \mathrm{mod}\, 2.
$$
Similar results with less restrictive assumptions can be found in~\cite[Section 6]{Kellendonk16}.

We now consider the relation of our pairing to (real) Poincar\'{e} duality when 
there is no disorder. If there is no disorder and the potential is periodic, then using the 
Bloch--Floquet transform, the relevant observable algebra is (up to stabilisation) 
the real $C^*$-algebra 
$C(\{\mathrm{pt}\},\R) \rtimes\Z^2 \cong C^*(\Z^2) \cong C(i\T^2)$ with 
$$
  C(i\T^2) = \left\{ f\in C(\T^2,\C)\,:\, \ol{f(k)} = f(-k) \right\},  \qquad KO_\ast (C(i\T^2)) \cong KR^{-\ast}(\T^2,\zeta)
$$
with $\zeta$ the involution $k\mapsto -k$.
The Kasparov module $\lambda_2$ is now just a real spectral triple 
for $C(i\T^2)\otimes Cl_{0,2}$ and can be extended to a spectral triple $\Lambda_2$ for 
$ C(i\T^2)\otimes C(i\T^2) \hat\otimes Cl_{0,2}$ (via the diagonal map
$\Delta:\T^2\to\T^2\times\T^2$).
Hence we have a representative of Kasparov's fundamental class~\cite{KasparovNovikov} and 
obtain a graded group isomorphism
$$
\cdot\otimes_{C(i\T^2)}[\Lambda_2]:\,KO_*(C(i\T^2))
\to KO^{*-2}(C(i\T^2)).
$$
Thus for every non-zero
element $[x]\in KO_*(C(i\T^2))$ we have $[x]\otimes_{C(i\T^2)}[\Lambda_2]\neq 0$.
Specialising to the case $[z]\in KO_4(C(i\T^2))$ and 
pairing the $K$-homology class $[z]\otimes_{C(i\T^2)}[\Lambda_2]$ 
with $[1]\in KO_0(C(i\T^2))$ gives
$$
[1]\otimes_{C(i\T^2)}([z]\otimes_{C(i\T^2)}[\Lambda_2])=[z]\otimes_{C(i\T^2)}[\lambda_2].
$$
Using the Atiyah--Bott--Shapiro framework, 
the pairing on the right hand side is also computable as follows. One finds
a representative $(Cl_{2,0},\mathcal{H},T)$ of the pairing and then regards
the kernel of  the operator $T$ in $KO_{2}(\R)$
as a graded left $C\ell_{2,0}$-module~\cite{ABS64}. Actually, given the structure of the product module, 
we actually start with $C\ell_{4,2}$-modules, and have simply removed a $Cl_{2,2}$ module using the Morita equivalence $Cl_{2,2}\sim \R$.
As noted above, this index lies in a group isomorphic to $\Z_2$, and is computable as a complex index computed $\bmod\,2$.

It remains to see that we surject onto $KO_2(\R)$. 
Recall that $KO_4(C(i\T^2)) \cong KR^{-4}(\T^2,\zeta)\cong \Z\oplus\Z_2$.
Here $\Z$ is the quaternionic rank of the Bloch 
bundle $\{P_\mu(k)\}_{k\in\T^d}$ and $\Z_2$ is the Kane--Mele 
invariant.
Since the Bloch bundle is quaternionic, it has even complex dimension, and thus the $\bmod\,2$ 
kernel dimension is zero.

For the torsion generator, suppose that
$\mathrm{dim}_\C \Ker\big(\pi_\omega(P_\mu)( X_1 \otimes 1  + iX_2 \otimes 1) \pi_\omega(P_\mu) \big)\, \mathrm{mod}\, 2=0$. 
If this index is trivial, then our $Cl_{4,2}$ module is the restriction of a $C\ell_{5,2}$-module. 
Using that $C\ell_{4,2} \cong C\ell_{4,0}\hat\otimes C\ell_{0,2}$, which 
is our decomposition of the index pairing, then the restriction  of a 
$C\ell_{5,2}$-module implies that the
original class $[P_\mu] \in KKO(C\ell_{4,0},C(i\T^2))$ is  the 
restriction of a $C\ell_{5,0}$ module (where we now refer to a Clifford module over 
a finitely generated and projective $C(i\T^2)$-module). 
We  remark that the $C\ell_{5,2}$-module cannot come from a $Cl_{1,2}$ structure 
on $\lambda_2$ as $\End_\R(\bigwedge^*\R^2) \cong C\ell_{0,2}\hat\otimes C\ell_{2,0}$ 
and the entirety of $\End_\R(\bigwedge^*\R^2)$ is used in the construction of 
$\lambda_2$.
If $[P_\mu]$ is the restriction of a $C\ell_{5,0}$-module,
 this implies that the class $[P_\mu] \in KO_4(C(i\T^2))$ is zero. 
The contrapositive of this argument then implies that a non-trivial 
$\Z_2$ component of the class 
$[P_\mu]$ will then give a non-trivial index pairing with $\lambda_2$. That is, 
the torsion part of $KO_4(A_{KM})$ is detected by
$\lambda_2$, at least in the absence of disorder.

\subsection{The bulk-edge correspondence}
Kellendonk and Richard consider disordered systems with boundary 
using the short exact sequence coming from the Wiener--Hopf extension,
$$
 0 \to (C(\Omega)\rtimes_\theta\R^{d-1}) \otimes \calK \to 
   \left(C_0(\R\cup \{+\infty\})\otimes C(\Omega)\rtimes_\theta\R^{d-1} \right)\rtimes \R 
     \to (C(\Omega)\rtimes_\theta\R^{d-1}) \rtimes \R \to 0,
$$
see~\cite{KR06}. It is proved in the case $d=2$ in~\cite{KSB04a} that 
a disordered Hamiltonian $H_\omega$ acting on $L^2(\R^{d-1}\times(-\infty,s])$ 
is affiliated to the algebra $\left(C_0(\R\cup \{+\infty\})\otimes 
C(\Omega)\rtimes_\theta\R^{d-1} \right)\rtimes \R$ for any 
$s\in\R$. Hence we can think of the Wiener--Hopf algebra as 
representing the half-infinite system with boundary.

Recall the unbounded Kasparov module $[\lambda_d] \in KKO^d(C(\Omega)\rtimes_\theta \R^d,C(\Omega))$
that is used to derive the noncommutative Chern numbers and disordered Clifford indices 
from Section \ref{subsec:disordered_pairings}. Analogous to the discrete setting in~\cite{BKR1}, 
factorisation of this Kasparov module via the Wiener--Hopf extension (Theorem \ref{thm:gen_bulkedge}) 
means that our analytic indices can be written as pairings over the bulk or edge 
algebra. Up to the sign $(-1)^{d-1}$, the bulk and edge pairings 
coincide. In particular, non-trivial topological effects 
are present on the boundary if and only if non-trivial effects 
are present in the bulk.

For complex invariants, our Chern number formulas apply for both bulk and edge 
invariants. For dimensions $1$, $2$ and $3$, work by~\cite{KSB04b,PSBbook} explicitly 
links our edge pairings, $\langle \partial[U_\mu],[X_{d-1}] \rangle$ and 
 $\langle \partial[P_\mu],[X_{d-1}]\rangle$, to the edge states, 
edge conductance or surface quantum Hall-like effect of a disordered Hamiltonian 
acting on a system with boundary (here $\partial$ is the boundary map in complex $K$-theory of the 
Wiener--Hopf extension). 

For real invariants, while we have an 
explicit equivalence of the analytic bulk and edge pairings, the link to 
the physical system is much harder to interpret. This is particularly true for torsion 
invariants which often cannot be detected by any local formula (see for example 
the discussion in~\cite[p148]{AS5}).

\subsection{Localisation of complex bulk invariants} \label{subsec:qc_and_localisation}
As a final step in our study of continuous models of disordered quantum systems, 
we connect the Sobolev index pairings considered
in Section \ref{sec:localisation_general} to dynamically localised observables. 

Recall the von Neumann algebra 
$(C(\Omega)\rtimes_\theta \R^d)'' \subset \calB[L^2(\R^d)\otimes L^2(\Omega,\bP)]$, 
which we denote $L^\infty(\Omega,\bP)\rtimes_\theta \R^d$. 
In the case of an ergodic measure $\bP$, we can also characterise this 
algebra as the $\ast$-algebra of weakly measurable families 
$\Omega \ni \omega\mapsto B_\omega\in\calB[L^2(\R^2)]$ which 
satisfy the covariance condition $U_a B_\omega U_{-a} = B_{T_a\omega}$ 
for all $a\in\R^d$ and with the norm 
$\| B\|_\infty = \bP\text{-ess}\sup_\omega \|B_\omega\|_{\calB[L^2(\R^d)]}$, 
see~\cite[Section 2]{LenzCrossedProd} or~\cite[Section 6]{ConnesMeasure}.

The study of localisation of observables in continuous models is 
considerably more complicated than its discrete counterpart. 
We will simply quote a result from~\cite{AENSS} and apply 
it to our models of interest. Similar results can also be found 
in~\cite{GK13, GT13}. We first note some notation. 
In the following theorem, we let $\chi_I(A)$ be the spectral projection of 
an operator $A$ associated to the interval $I\subset \R$. 
The kinds of random potentials treated in \cite{AENSS}
are very general, but not completely so. The 
characterisation
of the class of potentials requires a
fixed radius $r$ as described in~\cite[Section 1.7]{AENSS}). 
Given this $r$, we
denote by $\chi_{B^r_x}$ the characteristic function 
of a ball $B^r_x$ centred at $x$ of radius $r$.

\begin{thm}[\cite{AENSS}, Theorem 1.1] 
\label{thm:big_localisation_thm}
Let $\{H_\omega\}_{\omega\in\Omega}$ be a family of 
random magnetic Schr\"{o}dinger operators on $L^2(\R^d)$ satisfying the regularity
assumptions outlined in~\cite[Section 1.7]{AENSS}.
Let $\Xi$ be an open subset
of $\R^d$, and $\Lambda_n$ an increasing 
sequence of bounded open subsets of $\Xi$ with
$\cup \Lambda_n = \Xi$. Suppose that for some 
$0 < s < 1$ and an open bounded interval
$J$ there are constants $C < \infty$ and $m > 0$ such that
\begin{equation} \label{eq:localisation_hypoth_bound}
  \int_J \mathbb{E}\left(\left\| \chi_{B^r_x} (H_\omega^{(\Lambda_n)} - E)^{-1} \chi_{B^r_y} \right\|^s 
    \right) \mathrm{d}E \leq C Ae^{-m\, \mathrm{dist}_{\Lambda_n}(x,y)}
\end{equation}
for all $n \in \N$, $x, y \in \Lambda_n$. 
Then for every $v < 1/(2s)$ there exists $C_v < \infty$
such that, for all $x, y \in \Xi$,
\begin{equation}  \label{eq:localisation_bound}
  \mathbb{E}\bigg( \sup_{g:|g|\leq 1} \big\| \chi_{B^r_x} g(H_\omega^{(\Xi)}) \chi_J(H_\omega^{(\Xi)}) \chi_{B^r_y} \big\| \bigg) 
    \leq C_v e^{-v m \,\mathrm{dist}_{\Xi}(x,y)},
\end{equation}
where the supremum is taken over all Borel measurable functions $g$ which
satisfy $|g| \leq 1$ pointwise.
\end{thm}

Applying Equation \eqref{eq:localisation_bound} to the operators $g_t(H_\omega) = e^{-itH_\omega}$ yields
$$
  \mathbb{E}\left( \sup_t \big\| \chi_{B^r_x} e^{-itH_\omega^{(\Xi)}} \chi_J(H_\omega^{(\Xi)}) \chi_{B^r_y} \big\| \right) 
   \leq C_v e^{-vm\, \mathrm{dist}_{\Xi}(x,y)},
$$
a strong version of dynamical localisation.

The Hamiltonians we consider for applications to topological phases will always be 
bounded from below and affiliated to $L^\infty(\Omega,\bP)\rtimes_\theta \R^d$. 
Though we will occasionally require stricter assumptions.

\begin{cor} \label{cor:Fermi_proj_quasicts}
Let $H$ be 
affiliated to $L^\infty (\Omega)\rtimes_\theta\R^d$ and representing 
a family $\{H_\omega\}_{\omega\in\Omega}$ of disordered Hamiltonians.
 If $\{H_{\omega}\}_{\omega\in\Omega}$
 satisfies the hypothesis of 
Theorem \ref{thm:big_localisation_thm}, then 
$\chi_{(-\infty, E]}(H) \in \calA_\mathrm{Sob}$ for any $E$ in the 
localised region $J \subset\R$. In particular, if the Fermi energy $\mu$ 
is in $J$, then
the Fermi projection $P_\mu \in \calA_\mathrm{Sob}$.
\end{cor}
\begin{proof}
Applying the theorem, the operator $P_\mu$ 
has, on average over $\Omega$, an exponentially decaying integral kernel 
(in particular, see~\cite[Equation (1.10)]{AENSS}). 
We note that $P_\mu = \tilde{\pi}(p_\mu)$ with 
$p_\mu \in L^\infty(\Omega,\bP)\rtimes_\theta\R^d$. 
Therefore we check the Sobolev condition where by Theorem \ref{thm:big_localisation_thm} 
there are strictly positive constants $C_1$ and $C_2$ such that
\begin{align*}
  \| p_\mu \|_{r,1} &\leq C_r \int_{\R^d}\!\int_\Omega  (1+|x|)^{r}p_\mu(x;\omega)\,\mathrm{d}\bP(\omega) \,\mathrm{d}x 
    \leq C_1 C_r \int_{\R^d}\! (1+|x|^2)^r e^{-C_2 |x|}\,\mathrm{d}x  < \infty,
\end{align*}
and so is finite for any $r\in \N$. Because $p_\mu$ is a projection, we then 
obtain that $p_\mu \in \calW_{r,p}$ for 
any $r,p$.
\end{proof}

A key property of a dynamically-localised region of the spectrum $J\subset\sigma(H_\omega)$ (also 
called a mobility gap) is that 
the pure point spectrum of $H_\omega$ in $J$ is $\bP$-almost surely dense in $J$~\cite{AENSS}. 
Considering the element $H$ affiliated to $\calA_{\mathrm{Sob}}$ represented by the family $\{H_\omega\}_{\omega\in\Omega}$, 
then with probability $1$ the 
pure point spectrum of $H$ in $J$ is also dense in $J$. Therefore if $\mu\in J$, 
then $\bP$-almost surely $\mu$ is a limit point of 
eigenvalues and $\mu \in \sigma_\mathrm{ess}(H)$.

A key success of the noncommutative geometry approach to the quantum Hall 
effect is the proof that the Hall conductance is constant within a mobility gap, 
proved in the discrete case in~\cite[Section 5]{Bellissard94}. 
The continuous analogue of this result is quite involved. We have not been able 
to resolve this question fully, but present a result for a random family of 
Hamiltonians $H = \{H_\omega\}_{\omega\in\Omega}$ bounded from 
below and affiliated to the {$C^*$}-algebra $C(\Omega)\rtimes_\theta\R^d$ 
(\emph{not} the von Neumann 
closure $L^\infty(\Omega,\bP)\rtimes_\theta \R^d$).

\begin{prop}[\cite{ProdanBook}, Proposition 3.31] \label{prop:Borel_mobility}
Let $H  = \{H_\omega\}_{\omega\in\Omega}$ be a self-adjoint element 
that is bounded from below and affiliated to the 
\emph{$C^*$-algebra $C(\Omega)\rtimes_\theta\R^d$} with mobility gap $J$. 
Then $G(h) \in \calA_\mathrm{Sob}$ for every Borel function $G$ with 
support in $J$.
\end{prop}

The proof in~\cite[Proposition 3.31]{ProdanBook} is for a different setting with 
a slightly different localisation bound, but we observe that the key argument 
continues to hold here. The main difference is the replacement of a sum with 
an integral in~\cite[Equation (3.63), (3.69)]{ProdanBook} and 
we use the bound from Equation \eqref{eq:localisation_hypoth_bound} 
rather than~\cite[Equation (3.55)]{ProdanBook}.

A key motivation for considering $\calA_\mathrm{Sob}$ was to find a topology such 
that the cyclic cocycles used in our index formulas are continuous, but deformations in a fixed 
mobility gap are also continuous. Because $\calA_\mathrm{Sob}$ is defined using  
a tracial norm (which is weaker than the operator norm) but with strong regularity 
under the algebraic derivations, it is able to manage these two roles. 
The following result demonstrates this property.

\begin{prop} \label{prop:Mobility_homotopy_P}
Let $H=\{H_\omega\}_{\omega\in\Omega}$ be a random family 
of Hamiltonians affiliated to $C(\Omega)\rtimes_\theta \R^d$ with 
mobility gap $J$ and 
bounded from below. Then 
the map $J \ni E \mapsto \chi_{(-\infty,E]}(H) \in \calA_\mathrm{Sob}$ 
is continuous.
\end{prop}
\begin{proof}
We use the notation $P_E = \chi_{(-\infty,E]}(H)$ 
and $P_{[E,E']} = \chi_{[E,E']}(H)$.
By Proposition \ref{prop:Borel_mobility}, 
the spectral projections $P_{[a_n, b_n]} \in \calA_\mathrm{Sob}$ 
for any $a_n,b_n\in J$. In particular, by 
the Borel functional calculus (see~\cite[Theorem VIII.5]{ReedSimon1} for example)
if $a_n\to a\in J$ and 
$b_n\to b \in J$, then $P_{[a_n,b_n]}\to P_{[a,b]}$ in the  
strong operator topology with $P_{[a,b]}\in \calA_\mathrm{Sob}$.

Next we note that if $p_n$ is a sequence of trace-class 
projections with $p_n \to p$ strongly with $p$ trace-class, 
then $\Tr(p_n)\to \Tr(p)$.
Because the Sobolev algebra contains trace-class elements 
(under the dual trace), we can say that 
$\Tr_\tau(P_{[a_n,b_n]}) \to \Tr_\tau(P_{[a,b]})$ 
and in particular  $P_{[a_n,b_n]} \to P_{[a,b]}$ in $\calW_{0,1}$ (the Sobolev space). 
Assuming that $P_{[a_n,b_n]}  \preceq P_{[a,b]}$ 
(e.g. $[a_n,b_n] \subset [a,b]$ for all $n$), then we also have 
that $P_{[a_n,b_n]} \to P_{[a,b]}$ in $\calW_{0,p}$.
Because 
$P_{[a_n,b_n]},P_{[a,b]} \in \calW_{r,p}$ for any 
$r,p$ and all $n$, the localisation bound 
and the convergence of $P_{[a_n,b_n]}$ in the trace norm 
ensures that 
$P_{[a_n,b_n]} \to P_{[a,b]}$ in $\calW_{r,p}$ for all 
$r,p$. That is, deformations of the spectral projection 
of $h$ within a fixed mobility gap $J$ are continuous 
in the Sobolev topology. 
Hence, for $E,E'\in J$ with $E<E'$, then 
$\|P_{E'} - P_{E}\|_\mathrm{Sob} = \|P_{[E,E']}\|_\mathrm{Sob}$ 
can be controlled by $|E'-E|$.
\end{proof}

Let us put together our results for complex pairings.

\begin{cor}
Let $H=\{H_\omega\}_{\omega\in\Omega}$ be a random family of 
Hamiltonians satisfying the assumptions of 
Proposition \ref{prop:Mobility_homotopy_P} and fix an ergodic measure 
$\bP$ on $\Omega$. If $H$ has a chiral symmetry, we also assume the chiral involution 
$R_c\in \calA_\mathrm{Sob}^\sim$ is uniform in the mobility gap $J$. Then $\bP$-almost surely, the complex topological pairings from 
Section \ref{subsec:disordered_pairings} are well defined, 
$\Z$-valued, constant in $\Omega$ and 
 constant in a region of 
dynamical localisation. 
\end{cor}
\begin{proof}
The pairings extend to $\calA_\mathrm{Sob}$ by Theorem \ref{thm:quasicts_index_extension} 
and are $\Z$-valued for ergodic measures by Theorem \ref{thm:ergodic_loc_pairing}. 
If we make a 
deformation within a mobility gap, then for the Fermi projection, this is continuous in the 
Sobolev topology by Proposition \ref{prop:Mobility_homotopy_P}. For the 
case of the Fermi unitary, as $R_c$ is uniform in $J$ we can write the deformation 
$U_\mu(t) = \frac{1}{2}(1-R_c)(1-2P_\mu(t))\frac{1}{2}(1+R_c)$ and  the deformation is 
again continuous in $\calA_\mathrm{Sob}$ by Proposition \ref{prop:Mobility_homotopy_P}. 
Therefore the index pairing will be constant over either deformation by Theorem \ref{thm:ergodic_loc_pairing}.
\end{proof}
Hence we are able to obtain analogous results to those in~\cite{NB90,Bellissard94, PLB13,PSB14}.

\subsection{Delocalisation of complex edge states} \label{sec:deloc_edge}
Our argument follows~\cite[Section 6.6]{PSBbook}.
Let us briefly review our basic setup as well as some additional assumptions we will require. 
We have the short-exact sequence
$$
  0 \to C(\Omega)\rtimes_\theta \R^{d-1} \otimes \calK \to \calE
    \to C(\Omega) \rtimes_\theta \R^d \to 0,
$$
where $H_s$ is a disordered magnetic Schr\"{o}dinger operator acting on $L^2(\R^{d-1}\times (-\infty,s])$ with 
Dirichlet boundary conditions and 
affiliated to $\calE = \big( C_0(\R\cup\{+\infty\}) \otimes C(\Omega)\rtimes_\theta\R^{d-1}\big) \rtimes \R$ 
for every $s\in \R$. If $H_s$ is chiral symmetric, we assume that the chiral involution 
$R_c$ is sufficiently local so that $R_c$ is in the minimal unitisation of (matrices of) $C(\Omega)\rtimes_\theta\R^d$, $\calE$ 
and $C(\Omega)\rtimes_\theta\R^{d-1}$. Often $R_c = \begin{pmatrix} 1 & 0 \\ 0 & -1 \end{pmatrix}$ and so this 
criterion is trivially satisfied here.

\begin{lemma}[See~\cite{PSBbook}, Section 4.3, or \cite{HigsonRoe}, Chapter 4] \label{lem:k_theory_bdry_map}
If the Fermi level is in a gap in the spectrum, then
$$
   \partial[P_\mu] = [ \exp( 2 \pi i f_\mathrm{exp}(H_s) ) ] \in K_1( C(\Omega)\rtimes_\theta \R^{d-1}\otimes\calK)
$$
with $f_\mathrm{exp}(H_s)$ a smooth non-decreasing function that is $0$ below the spectral gap 
and $1$ above the spectral gap.

If $H$ has a chiral symmetry, let $U_\mu = \frac{1}{2}(1-R_c)(1-2P_\mu)\frac{1}{2}(1+R_c)$ be the 
chiral unitary. Then
$$
  \partial [ U_\mu ] = 
     \big[ e^{-i\frac{\pi}{2} f_\mathrm{ind}(H_s)} \frac{1}{2}(1+R_c) e^{-i\frac{\pi}{2} f_\mathrm{ind}(H_s)} \big] 
       - \big[ \frac{1}{2}(1+R_c) \big] \in K_0(C(\Omega)\rtimes_\theta \R^{d-1}\otimes\calK).
$$
with $f_\mathrm{ind}$ an odd and smooth non-decreasing function that is $-1$ below the spectral gap and 
$+1$ above the spectral gap.
\end{lemma}

We wish to consider our Chern number formulas for the edge pairings with $\partial[P_\mu]$ 
and $\partial[U_\mu]$ for a particular disorder space. Namely we take  
$\Omega\times [-L,L]$ for $L$ large but finite. This space still has an $\R^{d-1}$-action 
and given the invariant measure $\bP$ on $\Omega$, we can extend $\bP$ 
to $\tilde{\bP}$, the product of $\bP$ and (normalised) integration. The new measure $\tilde{\bP}$ is still invariant 
under the $\R^{d-1}$-action on $\Omega$ and so defines an unbounded trace $\tilde{\calT}$ on 
$C(\Omega\times[-L,L]) \rtimes_\theta \R^{d-1}$. 
As our general theorems only require an invariant tracial weight, our key Chern number 
results still apply for $\Omega\times[-L,L]$ and $\tilde{\bP}$.

\begin{thm} \label{thm:delocalised_edge}
Suppose that the Fermi level $\mu$ is in a spectral gap $J\subset \R$ of the bulk Hamiltonian $H$, 
and adopt the notation of Theorem \ref{thm:big_localisation_thm}.
\begin{enumerate}
  \item Suppose that $d$ is even and $\mathrm{Ch}_d(P_\mu)$ is non-zero.
Let $H_s$ be the Hamiltonian for the system with boundary. The 
  localisation bound with respect to $\Omega\times [-L,L]$ of $H_s$, 
  $$
    \int_J \mathbb{E}_{\tilde{\bP}}\left(\left\| \chi_{B^r_x} \big((H_\omega)_s^{(\Lambda_n)} - E \big)^{-1} \chi_{B^r_y} \right\|^s 
    \right) \mathrm{d}E \leq C Ae^{-m\, \mathrm{dist}_{\Lambda_n}(x,y)},
  $$
 cannot hold for large but finite $L$. 
  \item If $d$ is odd and $\mathrm{Ch}_d(U_\mu)$ is non-zero, then a localisation 
  bound of $H_s$ cannot hold for large but finite $L$.
\end{enumerate}
\end{thm}
The above result says that if our bulk-invariants are non-trivial, then the boundary spectrum 
of $H_s$ cannot become localised by the addition of an arbitrarily thick surface layer.
\begin{proof}
Using the measure $\tilde{\bP}$, Lemma \ref{lem:k_theory_bdry_map} 
and the bulk-edge correspondence, we have the equality for $d$ even
\begin{align} \label{eq:localised_bulkboundary}
  \mathrm{Ch}_d(P_\mu) = -C_{d-1} \sum_{\sigma\in S_{d-1}} (-1)^\sigma 
      (\tilde{\calT} \otimes \Tr_{L^2(\R)}) \Big(\prod_{j=1}^{d-1} (e^{ 2 \pi i f_\mathrm{exp}(H_s) })^* \partial_{\sigma(j)} e^{ 2 \pi i f_\mathrm{exp}(H_s) } \Big),
\end{align}
where we also take the trace over $L^2(\R)$ as the image of the boundary 
map is represented by 
elements in $C(\Omega \times [-L,L])\rtimes_\theta \R^{d-1} \otimes \calK[L^2(\R)]$.
If the localisation bound were  to hold for $H_s$, then there is a homotopy in 
$\calA_\mathrm{Sob}$ from $f_\mathrm{exp}(H_s)$ to $1-\chi_{(-\infty,\mu]}(H_s)$ which 
will not  change the right-hand side of Equation \eqref{eq:localised_bulkboundary} 
by Theorem \ref{thm:quasicts_index_extension}. 
However, the real part of $e^{2 \pi i (1-\chi_{(-\infty,\mu]}(H_s))}$ 
is precisely $\chi_{(-\infty,\mu]}(H_s)$ 
and so the Fermi projection lifts to another projection, which implies that the image of $P_\mu$ under 
the boundary map will be trivial. Hence our boundary pairing must be trivial, which contradicts
 that $\mathrm{Ch}_d(P_\mu)$ is non-zero.
 
Similarly for $d$ odd and non-trivial bulk pairing, if the localisation bound holds, then there is a 
homotopy from $f_\mathrm{ind}(H_s)$ to $\mathrm{sgn}(H_s)$ in the topology of $\calA_\mathrm{Sob}$.  
Then, because $H_s$ is chiral symmetric, $\mathrm{sgn}(H_s)$ anti-commutes with $R_c$ and so
$e^{-i\frac{\pi}{2} \mathrm{sgn}(H_s)} \frac{1}{2}(1+R_c) e^{-i\frac{\pi}{2} \mathrm{sgn}(H_s)} =  \frac{1}{2}(1+R_c)$ 
Hence $\partial[U_\mu]$ will be trivial and the edge pairing will vanish, again contradicting
the assumption on the bulk pairing.
\end{proof}

We remark that the hypotheses of Theorem \ref{thm:delocalised_edge} still require
a spectral gap of the bulk Hamiltonian $H$ affiliated to $C(\Omega)\rtimes_\theta\R^d$. This is because we 
used the boundary map in $K$-theory associated to the Wiener--Hopf extension, and we do not have 
an analogous extension for the Sobolev algebra $\calA_\mathrm{Sob}$. A weakening of this 
assumption as in~\cite{EGS05, Taarabt14,  GrafShapiro}, while desirable, is beyond the scope of this manuscript.

\subsection{Real pairings and localisation} \label{Sec:real_loc_pairings}

Our aim is to extend the analytic $\Z$ or $\Z_2$-valued skew-adjoint Fredholm indices 
from Equation \eqref{eq:skew_Fred_index} to the Sobolev algebra 
and topological phases in strong disorder. 

We assume that 
$H$ is affiliated to $L^\infty(\Omega,\bP)\rtimes_\theta \R^d$, satisfies the 
hypotheses of Theorem \ref{thm:big_localisation_thm}, the Fermi energy $\mu$ is in a mobility gap $J$ 
and $\bP$ is ergodic under the group action.
We also assume that all $CT$-symmetries 
$R_T$, $R_P$ and $R_c$ are self-adjoint unitaries in the real locally convex 
algebra $\calA_\mathrm{Sob}^\sim$ (recall that $R_T^2 = \pm 1$, $R_P^2=\pm 1$). 
Following the procedure in \cite[Section 3.3]{BCR15}, we can construct a finitely generated 
and projective right-$\calA_\mathrm{Sob}$ module $p{\calA_\mathrm{Sob}}^{\oplus N}_{\calA_\mathrm{Sob}}$ 
equipped with a left $C\ell_{n,0}$-action that is constructed from the symmetry operators. 
We remark that the module and the Clifford action changes depending on the symmetry 
type of the Hamiltonian (in particular, $p$ need not be the Fermi projection).

We now construct $\Z$ or $\Z_2$-valued pairings of the Clifford module 
$p{\calA_\mathrm{Sob}}^{\oplus N}_{\calA_\mathrm{Sob}}$ via a skew-adjoint 
Fredholm index. We first note that the decomposition of the semifinite spectral triple 
into a direct integral 
of spectral triples in Section \ref{sec:chop-chop} is valid for both complex and real semifinite 
spectral triples, where we have the fibre
$$
\bigg( \calA_\mathrm{Sob} \hat\otimes C\ell_{0,d}, \, {}_{\wt{\pi}_\omega}L^2(\R^d) \hat\otimes 
   \bigwedge\nolimits^{\!*}\R^d, \, X=\sum_{j=1}^d X_j \hat\otimes \gamma^j
    \bigg).
$$

\begin{prop} 
Suppose that the operator $H$ is affiliated to $L^\infty(\Omega,\bP)\rtimes_\theta \R^d$, the  
Fermi energy lies in a mobility gap $J$ and the measure $\bP$ 
is invariant and ergodic under the $\R^d$-action.
Let $(C\ell_{n,0}, p {\calA_\mathrm{Sob}}^{\oplus N}_{\calA_\mathrm{Sob}})$ be a finitely-generated 
and projective $\calA_\mathrm{Sob}$ module with a left $C\ell_{n,0}$-action 
 (see \cite[Section 3.3]{BCR15} for a
construction of this class). 
Let $p_\omega F_X p_\omega = \wt{\pi}_\omega(p) (X(1+X^2)^{-1/2}\otimes 1_N) \wt{\pi}_\omega(p)$.
Then the skew-adjoint Fredholm index 
$\Index_{n,d}\!\big( (p_\omega F_X p_\omega)_+ \big)$ 
from Equation \eqref{eq:skew_Fred_index} 
is $\bP$-almost surely well-defined and constant over $\Omega$. 
\end{prop}
\begin{proof}
The operator $p_\omega F_X p_\omega$ is $\bP$-almost surely Fredholm so, 
as in Theorem \ref{thm:ergodic_loc_pairing}, the ergodic 
assumption means that 
we only have to check constancy 
of the Clifford indices on an orbit in $\Omega$.
The covariance described
in Equation \eqref{eq:covariance}
gives that  $F_X=X(1+X^2)^{-1/2}$ is unitarily equivalent to 
$F_{X+a}$ (modulo compacts) via the unitary $U_a$ implementing translation on the product 
module $p \calH$ (where, to be precise 
$\calH = L^2(\R^d, \R^N)\hat\otimes \bigwedge^*\R^d$).
As in the complex case, the transformation will be a compact perturbation of the original 
Fredholm operator.
By the stability of the index pairings, the defined index is $\bP$-almost surely constant.
\end{proof}

We have defined $\Z$ or $\Z_2$-valued strong topological phase indices whenever the 
Fermi energy is in a mobility gap $J$ and the symmetries are uniform in $J$.  
Because of the ergodicity assumption and almost-sure constancy of the index pairings, 
in this special case one can take an average of $\Index_{n,d}\!\big( (p_\omega F_X p_\omega)_+ \big)$ 
over the disorder space $\Omega$ and the result will not change (similar to 
the case of complex pairings in Section \ref{sec:ergodic_pairing}).

\begin{remarks}
\begin{enumerate}
  \item We do not currently have a proof that the defined $\Z$ and $\Z_2$-valued indices 
  for dynamically localised Hamiltonians with anti-linear symmetries are constant in a mobility gap. 
  In the case of complex pairings, we used the fact that the analytic index was equal to a cyclic 
  cocycle that was continuous in the topology of $\calA_\mathrm{Sob}$. Because there is no general 
  cyclic formula for our real indices, we cannot use this argument. 
  We leave open this question for the time being (which to the best of our knowledge 
  is also unresolved in the discrete case).
  \item Our argument for delocalisation of boundary states for complex topological phases 
in Section \ref{sec:deloc_edge} relies on both an explicit computation of the boundary map in $K$-theory of the 
  symmetry class of the Hamiltonian and the stability of index pairings under deformations continuous 
  in $\calA_\mathrm{Sob}$. We have constructed an explicit representative of the 
  extension class in Proposition \ref{prop:ext_is_crossed_prod_module}, but a representative of 
  $[H]\hat\otimes [\mathrm{ext}]$ in $KO$-theory expressed in terms of the Hamiltonian with boundary $H_s$ is, 
  in general, a much harder task than the complex case. Explicit formulas for
   boundary maps in real $K$-theory are known, see~\cite{BL15}, but
  the passage from a disordered Schr\"{o}dinger operator to a unitary with symmetries that is required for the 
  Boersema--Loring picture is non-trivial. Similarly, because the argument in Section \ref{sec:deloc_edge} 
  relies on the stability of indices within the mobility gap $J$,  we leave the delocalisation of edge states for 
  non-trivial bulk pairings anti-linear symmetries as an open problem.
\end{enumerate}
\end{remarks}


\section{Concluding remarks}
We finish by making some brief comments on our results, their possible extensions and 
current limitations.

Throughout the paper, whenever the semifinite local index theorem is employed, we 
are restricted to complex algebras and spectral triples only. This is a large limitation 
as there are many materials of interest which have invariants arising from
a non-torsion  real pairing 
(see \cite{GSB15} for example). A local and cyclic expression for 
real pairings of the form
$$
  KO_d(B\rtimes_\theta\R^d) \times KKO^d(B\rtimes_\theta\R^d,B) \to KO_0(B) 
    \xrightarrow{(\tau_B)_\ast}  \R
$$ 
should be possible, though the details of the proof of the even local index formula 
need to be carefully checked to see if they extend. Pairings that take value in 
$KO_j(B)$ for $j\neq 0$ could also be studied in this way, though this would require using the
(graded) trace on the Clifford algebra or working with suspensions.

The lack of a local formula for torsion invariants, while unsurprising, 
means that an explicit link between torsion-valued pairings and 
physical phenomena is a key challenge for mathematical physicists interested 
in topological insulators. 
See \cite{Kellendonk16} for recent progress.

In order to study the stability of complex topological phases under perturbations within a fixed 
mobility gap, we had to increase the hypothesis on our Hamiltonian so that 
it is affiliated to the $C^*$-algebra $C(\Omega)\rtimes_\theta \R^d$ and not 
the von Neumann closure $L^\infty(\Omega, \bP)\rtimes_\theta \R^d$. While a 
complete resolution of this somewhat technical problem may be quite difficult, 
perhaps analogous results to ours can be obtained for specific model 
Hamiltonians of interest to condensed matter physicists and with weaker 
affiliation hypotheses. 

The question of stability of $\Z$ or $\Z_2$ strong phases with anti-linear symmetries under perturbations 
within a mobility gap remains a difficult but physically pertinent question. The lack of 
a cyclic formula means that a careful analysis of the stability of the skew-adjoint Fredholm indices under 
perturbations in the  Sobolev topology is required. 
As previously mentioned, such questions do not appear to be resolved even in the 
discrete case.

We have only considered the extension of bulk phases to strong disorder. The 
question of edge indices and the 
bulk-boundary correspondence in a mobility gap as in~\cite{EGS05, Taarabt14,  GrafShapiro} 
requires extra study.

\subsubsection*{Acknowledgements}
CB thanks Hermann Schulz-Baldes and Giuseppe De Nittis for 
posing the question of continuous Chern numbers to him, which 
eventually turned into the present manuscript, and for helpful discussions 
on the topic. 
The authors also thank Andreas Andersson, Alan Carey, Johannes Kellendonk 
and Emil Prodan 
for useful discussions. 
We would also like to thank an anonymous referee who pointed out an important error 
in an earlier version of this work.
CB is supported by a postdoctoral fellowship for overseas researchers from The Japan Society for the Promotion of Science (No. P16728)
and a KAKENHI Grant-in-Aid for JSPS fellows (No. 16F16728). CB also acknowledges 
support from the Australian Mathematical Society and the Australian Research Council during the production of this work.
This work was supported by World Premier International Research Center Initiative (WPI), MEXT, Japan.
Lastly, both authors thank the mathematical research institute MATRIX in Australia 
where part of this work was carried out.


\appendix

\section{Summary of non-unital index theory} \label{Sec:Non_unital_prelims}
In what follows we will assume that the algebras we deal with are separable. 
 A useful exposition of $KK$-theory and its 
applications can be found in~\cite{Blackadar, PSBKK} and~\cite{BMvS13} 
for the unbounded setting.

Given a real or complex $C^*$-module $E_B$ over a $\Z_2$-graded $C^*$-algebra $B$, 
we denote by $\End_B(E)$  the algebra of adjointable endomorphisms of 
$E$ subject to the $B$-valued inner-product $(\cdot\mid\cdot)_B$. 
The algebra of finite rank endomorphisms $\End_B^{00}(E)$ is the algebraic 
span of the operators 
$\Theta_{e_1,e_2}$ for $e_1,e_2\in E$ such that
$$  
  \Theta_{e_1,e_2}(e_3) = e_1\cdot (e_2\mid e_3)_B
$$
with $e\cdot b$ the (possibly graded) right-action of $B$ on $E_B$. 
The algebra of compact endomorphisms $\End_B^0(E)$ is the 
$C^*$-closure of $\End^{00}_B(E)$.

\begin{defn}
Let $A$ and $B$ be real $\Z_2$-graded $C^*$-algebras.
A real unbounded Kasparov module $(\calA, {}_\pi{E}_B, D)$ is a
$\Z_2$-graded real $C^*$-module ${E}_B$, a graded 
representation of $A$ on $E_B$, $\pi:A \to \End_B({E})$, and an
unbounded self-adjoint, regular and odd operator $D$ such that for all $a\in \calA\subset A$, 
a dense $\ast$-subalgebra,
\begin{align*}
  &[D,\pi(a)]_\pm \in \End_B({E}),   &&\pi(a)(1+D^2)^{-1/2}\in \End_B^0({E}).
\end{align*}
\end{defn}
For a complex Kasparov module, one simply replaces all spaces and algebras 
with complex ones.
Where unambiguous, we will omit the representation $\pi$ 
and write unbounded Kasparov modules as $(\calA, {E}_B, D)$.
The results of Baaj and Julg~\cite{BJ83} continue to hold for real 
Kasparov modules, so given an unbounded module 
$(\calA,{E}_B,D)$ we apply the bounded transformation 
to obtain the real Kasparov module $(A,{E}_B,D(1+D^2)^{-1/2})$.

\subsection{Semifinite theory}

An unbounded $A$-$\C$ or $A$-$\R$ Kasparov module is precisely a complex or 
real spectral triple as defined by Connes. Complex spectral triples satisfying 
additional regularity properties have the advantage 
that the local index formula by Connes and Moscovici~\cite{CoM} gives computable expressions 
for the index pairing with $K$-theory, a special case of the Kasparov product
$$
  K_j(A) \times KK^j(A,\C) \to K_0(\C)\cong \Z.
$$
We can extend this general framework by working with semifinite spectral triples.

Let $\tau$ be a fixed faithful, normal, semifinite trace on a von Neumann algebra 
$\calN$. We let $\calK_\calN$ be the $\tau$-compact
operators in $\calN$ (that is, the norm closed ideal generated by the 
projections $P\in\calN$ with $\tau(P)<\infty$).
\begin{defn}
A semifinite spectral triple $(\calA,\calH,D)$ relative to $(\calN,\tau)$ is given by a $\Z_2$-graded 
Hilbert space $\calH$, 
a graded $\ast$-algebra $\calA\subset\calN$ with (graded) representation on 
$\calH$ and a densely defined odd 
unbounded self-adjoint operator $D$ affiliated to $\calN$ such that
\begin{enumerate}
  \item $[D,a]_\pm$ is well-defined on $\Dom(D)$ and extends to a bounded operator on $\calH$ for 
  all $a\in\calA$,
  \item $a(1+D^2)^{-1/2}\in\calK_\calN$ for all $a\in\calA$.
\end{enumerate}
For complex algebras and spaces, we can also remove the gradings, in 
which case the semifinite spectral triple is called odd (otherwise even).
\end{defn}

If we take $\calN=\calB(\calH)$ and $\tau = \Tr$, then we recover the 
usual definition of a spectral triple.

\begin{thm}[\cite{KNR,CGRS2}] \label{thm:semifinite_to_KK}
Let $(\calA,\calH,D)$ be a complex semifinite spectral triple associated to 
$(\calN,\tau)$ with $A$ the $C^*$-completion of $\calA$. Then $(\calA,\calH,D)$ 
determines a class in $KK(A,C)$ with $C$ a subalgebra of $\calK_\calN$. If $A$ 
is separable, we can take $C$ to be separable. If 
$(\calA,\calH,D)$ is odd, then the triple determines a class in $KK^1(A,C)$.
\end{thm}

\subsubsection{The semifinite index pairing} \label{subsec:Index_pairing_defn}
Semifinite spectral triples $(\calA,\calH,D)$ with $\calA$ separable 
and ungraded can be paired with $K$-theory elements 
by the following composition
\begin{equation} \label{eq:complex_semifinite_index_pairing}
  K_j(A) \times KK^j(A, C) \to K_0(C) \xrightarrow{\tau_\ast} \R,
\end{equation}
with the class in $KK^j(A,C)$ coming from Theorem \ref{thm:semifinite_to_KK}. 
The image of the semifinite index pairing is a countably generated subset of 
$\R$ and, as such, can 
potentially detect finer invariants than the usual pairing of $K$-theory 
with $K$-homology. 
We call the map from Equation \eqref{eq:complex_semifinite_index_pairing} 
the \emph{semifinite index pairing} of a $K$-theory class with a 
semifinite spectral triple and use the notation $\langle [e],[(\calA,\calH,D)]\rangle$ 
for $[e]\in K_j(A)$ to represent this pairing.

We can also describe the semifinite index pairing analytically
using the semifinite Fredholm index.
Given a semifinite von Neumann algebra $(\calN,\tau)$, an operator 
$T\in \calN$ that is invertible modulo $\calK_\calN$ has semifinite Fredholm 
index
$$
   \Index_\tau(T) = \tau(P_{\Ker(T)}) - \tau(P_{\Ker(T^*)}).
$$

We can use the semifinite index to write down an analytic formula for the 
semifinite pairing. 
Let $\calA^{\sim}=\calA\oplus\C$ be the minimal unitisation of $\calA$. 
Given $b\in M_n(\calA^\sim)$ 
we let
\begin{equation*} 
\hat{b} = \begin{pmatrix} b & 0 \\ 0 & 1_b \end{pmatrix},
\end{equation*}
where $1_b = \pi^n(b)$ and $\pi^n : M_n(\calA^\sim) \to M_n(\C)$ is the 
quotient map coming from the unitisation.

\begin{prop}[\cite{CGRS2}, Proposition 2.13] \label{prop:semifinite_index_is_kas_prod}
Let $(\calA,\calH,D)$ be a complex semifinite spectral triple relative to 
$(\calN,\tau)$ with $\calA$ separable and $D$ invertible. Let $e$ be a projector in $M_n(\calA^\sim)$, which 
represents $[e]\in K_0(\calA)$ and $u$ a unitary in $M_n(\calA^\sim)$ 
representing $[u]\in K_1(\calA)$. In the even case, 
define $T_\pm = \frac{1}{2}(1\mp \gamma)T\frac{1}{2}(1\pm \gamma)$ 
with $\gamma$ the grading on $\calH$.
Then with $F = D|D|^{-1}$
and $P=(1+F)/2$, the semifinite index pairing is 
represented by
\begin{align*}
  \langle [e] - [1_e], (\calA,\calH,D)\rangle &= \Index_{\tau\otimes \Tr_{\C^{2n}}}\!(  \hat{e}(F \otimes 1_{2n})_+ \hat{e}), \qquad \textrm{even case}, \\
  \langle [u], (\calA,\calH,D) \rangle &= \Index_{\tau\otimes \Tr_{\C^{2n}}}\!\left( (P \otimes 1_{2n})\hat{u}(P \otimes 1_{2n}) \right), \qquad \textrm{odd case}.
\end{align*}
\end{prop}

If $D$ is not invertible, we take $m>0$ and define the double spectral triple
$(\calA,\calH\oplus\calH,D_m)$ 
relative to $(M_2(\calN),\tau\otimes\Tr_{\C^2})$, 
where the operator $D_m$ and the action of $\calA$ is given by
\begin{align*}
  &D_m = \begin{pmatrix} D & m \\ m & -D \end{pmatrix},  &&a\mapsto  \begin{pmatrix} a & 0 \\ 0 & 0 \end{pmatrix}
\end{align*}
for all $a\in\calA$. If $(\calA,\calH,D)$ is graded by $\gamma$, 
then the double is graded by $\hat{\gamma} = \gamma \oplus (-\gamma)$. 
Doubling the spectral triple does not change the $K$-homology class 
and ensures that the unbounded operator $D_m$ is 
invertible~\cite{Connes85}. 

\begin{remark}[Semifinite spectral triples and torsion invariants]
The pairing of Equation \eqref{eq:complex_semifinite_index_pairing} is 
valid in both the complex and real setting. The pairing is, however, 
unhelpful in the case of torsion invariants as if $[x]\in K_0(C)$ has 
finite order, then $\tau_\ast([x]) = 0$. In particular, torsion 
invariants are common in real $K$-theory and play 
an important role in, for example, characterising the topological 
phase of a free-fermionic system~\cite{Kitaev09}. In order to 
access torsion invariants we do not take the induced trace and 
consider the more general product
\begin{equation} \label{eq:torsion-semifinite-pairing}
  K_j(A) \times KK^d(A, C) \to K_{j-d}(C)
\end{equation}
for real or complex algebras. Hence 
we work with Kasparov modules directly. Unbounded Kasparov theory is 
a  useful method for computing internal Kasparov products,
as one must when 
the product represents a torsion class.

The image of the $K$-theoretic pairing from 
Equation \eqref{eq:torsion-semifinite-pairing} can be 
interpreted as a Clifford module, a finitely generated subspace of a 
countably generated $C^*$-module $E_C$ with graded Clifford action.
We can associate an analytic index to elements in 
$KO_{j-d}(C)$ or $K_{j-d}(C)$ via an analogue of Atiyah--Bott--Shapiro theory of Clifford 
modules, 
see \cite{ABS64, BCR15}. 
\end{remark}

\subsubsection{Kasparov modules to semifinite spectral triples}
We can associate a Kasparov module to a semifinite 
spectral triple by Theorem \ref{thm:semifinite_to_KK}.
One may ask if the converse is true. Given an unbounded 
Kasparov $A$-$B$ module with $B$ containing a faithful semifinite 
norm lower-semicontinuous trace (or tracial weight), 
we can often construct semifinite spectral 
triples using the dual 
trace construction 
(see Section \ref{sec:semifinite_spec_trip} or~\cite{PR06} for a simple example). 

The dual-trace method of constructing semifinite spectral triples has the advantage that 
the algebra $B$ is often more closely related to the problem under consideration 
than the algebra $C$ from Theorem \ref{thm:semifinite_to_KK}. In particular,
the semifinite index pairing from Equation \eqref{eq:complex_semifinite_index_pairing} 
can be rewritten with $B$ in the place of $C$.

Given a sufficiently regular (complex) semifinite 
spectral triple from an unbounded Kasparov module, we may use the semifinite 
local index formula to compute the $K$-theoretic semifinite index pairing. As the local index 
formula is a cyclic expression involving traces and derivations, semifinite 
spectral triples and index theory can be employed in order to 
more easily
compute pairings of $K$-theory classes with unbounded Kasparov modules 
as in Equation \eqref{eq:complex_semifinite_index_pairing}.

\subsection{Summability of non-unital spectral triples} \label{subsec:cgrs2_preliminaries}
Spectral triples often contain more than just $K$-homological data. 
Hence we introduce extra structure on spectral triples that have the 
interpretation of a differential structure and measure theory. 
If the algebra is non-unital and non-local in the sense of~\cite{RennieSmoothness}, 
then we require the 
noncommutative measure theory developed in~\cite{CGRS1, CGRS2}.
Our brief exposition follows~\cite[Section 2]{vdDPR13}. In order to 
discuss smoothness and summability for non-unital spectral triples, 
we need to introduce an analogue of $L^p$-spaces for operators and weights 
over a semifinite von Neumann algebra $(\calN, \tau)$. 

\begin{defn}
Let $D$ be a densely defined self-adjoint operator affiliated to $\calN$. 
Then for each $p\geq 1$ and $s>p$ we define a weight $\varphi_s$ 
on $\calN$ by
$$  
\varphi_s(T) = \tau\!\left((1+D^2)^{-s/4} T (1+D^2)^{-s/4}\right) 
$$
for $T$ a positive element in $\calN$. We define the subspace 
$\calB_2(D,p)$ of $\calN$ by
$$  
\calB_2(D,p) = \bigcap_{s>p}\left(\Dom(\varphi_s)^{1/2}\cap (\Dom(\varphi_s)^{1/2})^*\right), \quad
  \Dom(\varphi_s)^{1/2} = \big\{ T\in \calN\,:\, \varphi_s(T^*T) < \infty \big\}.
$$
\end{defn}
Take $T\in \calB_2(D,p)$. The norms
$$  
\calQ_n(T) = \left( \|T\|^2 + \varphi_{p+1/n}(|T|^2) + \varphi_{p+1/n}(|T^*|^2)\right)^{1/2}  
$$
for $n=1,2,\ldots$ take finite values on $\calB_2(D,p)$ and provide a 
topology on $\calB_2(D,p)$ stronger than the norm topology.
The space $\calB_2(D,p)$ is a Fr\'{e}chet algebra~\cite[Proposition 1.6]{CGRS2} 
and can be interpreted as the bounded square integrable operators.

To introduce the bounded integrable operators, first take 
the span of products, $\calB_2(D,p)^2$, 
and define the norms
$$  
\calP_{n}(T) = \inf\left\{\sum_{i=1}^k \calQ_n(T_{1,i})\calQ_n(T_{2,i})\,:\,
  T= \sum_{i=1}^k T_{1,i}T_{2,i},\ T_{1,i},\,T_{2,i}\in\calB_2(D,p) \right\}, 
$$
where the sums are finite and the infimum is over all possible such 
representations of $T$. It is shown in~\cite[p12--13]{CGRS2} that $\calP_n$ 
are norms on $\calB_2(D,p)^2$.

\begin{defn}
Let $D$ be a densely defined and self-adjoint operator and $p\geq 1$. 
We define $\calB_1(D,p)$ to be the completion of $\calB_2(D,p)^2$ with 
respect to the family of norms $\{\calP_n\,:\, n=1,2,\ldots\}$.
\end{defn}

\begin{defn} \label{def:finitely_summ_spec_trip_nonunital}
A semifinite spectral triple $(\calA,\calH,D)$ relative to 
$(\calN, \tau)$ is said to be finitely summable if 
there exists $s>0$ such that for all $a\in\calA$, $a(1+D^2)^{-s/2}\in\calL^1(\calN,\tau)$. 
In such a case we let 
$$  
p = \inf\{s>0\,:\, \forall a\in\calA, \,\tau(|a|(1+D^2)^{-s/2}) < \infty \} 
$$
and call $p$ the spectral dimension of $(\calA,\calH,D)$.
\end{defn}

Note that $|a|(1+D^2)^{-s/2}\in\calL^1(\calN,\tau)$ by the polar decomposition 
$a = v|a|$, which does not require $|a|$ to be in $\calA$. For the definition of 
spectral dimension to have meaning, we require that $\tau(a(1+D^2)^{-s/2})\geq 0$ 
for $a\geq 0$, a fact that follows from~\cite[Theorem 3]{Bikchentaev98}. 
For a semifinite spectral triple $(\calA,\calH,D)$ to 
be finitely summable with 
spectral dimension $p$, it is a necessary condition that 
$\calA\subset \calB_1(D,p)$~\cite[Proposition 2.17]{CGRS2}.

\begin{defn}
Given a densely-defined self-adjoint operator $D$, 
set $\calH_\infty = \bigcap_{k\geq 0}\Dom(D^k)$. For an operator 
$T:\calH_\infty\to\calH_\infty$, we define
\begin{align*}
   &\delta(T) = [|D|,T],  &&L(T) = (1+D^2)^{-1/2}[D^2,T],  &&R(T) = [D^2,T](1+D^2)^{-1/2}.
\end{align*}
\end{defn}

One has that (cf.~\cite{CoM, CPRS2})
$$  
\bigcap_{n\geq 0} \Dom(L^n) = \bigcap_{n\geq 0} \Dom(R^n) = \bigcap_{k,l\geq 0}\Dom(L^k\circ R^l) = \bigcap_{n\geq 0} \Dom(\delta^n). 
$$
We see that to define $\delta^k(T)$, we require that $T: \calH_k \to \calH_k$ 
for $\calH_k = \bigcap_{l=0}^k \Dom(D^l)$.
\begin{defn}
Let $D$ be a densely defined self-adjoint operator affiliated to $\calN$ and $p\geq 1$. 
Then define for $k=0,1,\ldots$
$$  
\calB_1^k(D,p) = \left\{ T\in\calN \,\left\vert \, T:\calH_l\to\calH_l 
  \text{ and }\delta^l(T)\in\calB_1(D,p)\,\, \forall l=0,\ldots,k \right. \right\}, 
$$
$$  
\calB_2^k(D,p) = \left\{ T\in\calN \,\left\vert \, T:\calH_l\to\calH_l 
  \text{ and }\delta^l(T)\in\calB_2(D,p)\,\, \forall l=0,\ldots,k \right. \right\} 
$$
as well as
$$  
\calB_1^\infty(D,p) = \bigcap_{k=0}^\infty \calB_1^k(D,p), \qquad 
  \calB_2^\infty(D,p) = \bigcap_{k=0}^\infty \calB_2^k(D,p). 
$$
\end{defn}
For any $k$ (including $\infty$), we equip $\calB_1^k(D,p)$ with 
the topology induced by the seminorms 
$$ 
\calP_{n,l}(T) = \sum_{j=0}^l \calP_n (\delta^j(T)) 
$$
for $T\in\calN$, $l= 0,\ldots,k$ and $n\in\N$.

If we are interested in index theory in the non-compact setting, we need 
to control the integrability of both functions and their derivatives. 
The noncommutative analogue of this regularity turns out to be a finitely 
summable spectral triple but with additional smoothness properties.

\begin{defn} \label{def:smoothly_summable}
Let $(\calA,\calH,D)$ be a semifinite spectral triple relative to $(\calN,\tau)$. 
We say that $(\calA,\calH,D)$ 
is $QC^k$-summable if it is finitely summable with spectral dimension $p$ and 
$$ 
 \calA \cup [D,\calA] \subset \calB_1^k(D,p). 
$$
We say that $(\calA,\calH,D)$ is smoothly summable if it is $QC^k$-summable 
for all $k\in\N$, that is 
$$  
\calA\cup [D,\calA] \subset \calB_1^\infty(D,p). 
$$
\end{defn}

For a smoothly summable spectral triple $(\calA,\calH,D)$, we can 
introduce the $\delta$-$\varphi$ topology on $\calA$ by the seminorms
\begin{equation}  \label{eq:delta_phi_topology_seminorms}
   \calA \ni a \mapsto \calP_{n,k}(a) + \calP_{n,k}([D,a])   
\end{equation}
for $n,k\in\N$. The completion of $\calA$ in the $\delta$-$\varphi$ 
topology is Fr\'{e}chet and closed under the holomorphic functional 
calculus~\cite[Proposition 2.20]{CGRS2}.
We finish this section with a sufficient and checkable condition 
of finite summability of spectral triples.

\begin{prop}[\cite{CGRS2}, Proposition 2.16] 
\label{prop:finite_summ_condition}
Let $(\calA,\calH,D)$ be a semifinite spectral triple. If $\calA \subset \calB_1^\infty(D, p)$ 
for some $p\geq 1$, then $(\calA,\calH,D)$ is finitely 
summable with spectral dimension given by the infimum of such $p$'s. 
More generally, if $\calA\subset \calB_2(D,p)\calB_2^{\lfloor p\rfloor+1}(D,p)$ for $p\geq 1$, 
then $(\calA,\calH,D)$ is finitely
summable with spectral dimension given by the infimum of such $p$'s.
\end{prop}
Lemma \ref{lem:tight-sum} offers a slight variation of this result in order to get sharp results on localisation.
Lemma \ref{lem:tight-sum}
has essentially the same conclusion, though different statement and proof, as \cite[Proposition 6.6]{CGRS1}.

\subsection{The local index formula} \label{subsec:local_index_formula}
We now briefly recall the semifinite local index formula from~\cite{CGRS2}, which is an 
extension of previous formulas, 
\cite{CoM, RennieSummability, CPRS2,CPRS3}, to non-unital and non-local 
semifinite (complex) spectral triples. We note that the local index formula 
requires a smoothly summable semifinite spectral triple of finite spectral 
dimension. This may seem restrictive, 
but turns out to be satisfied in our examples.

To define the resolvent cocycle, we first establish the notation 
$R_s(\lambda) = (\lambda-(1+s^2+D^2))^{-1}$.
\begin{defn}[\cite{CGRS2,CPRS2,CPRS3}] \label{def:resolvent_cocycle_defn}
Let $(\calA,\calH,D)$ be a smoothly summable complex semifinite spectral triple 
relative to $(\calN,\tau)$
with spectral dimension $p$. For $a\in(0,1/2)$, let $\ell$ 
be the vertical line $\ell=\{a+iv\,:\,v\in\R\}$. We define the 
resolvent cocycle $(\phi_m^r)_{m=0}^M$ for $\Re(r)>(1-m)/2$ as 
\begin{align*} 
   \phi_m^r(a_0,\ldots,a_m)  = \frac{\eta_m}{2\pi i} \int_0^\infty\! s^m\,
     \tau\!\left( \gamma \int_\ell \lambda^{-p/2-r} a_0 R_s(\lambda) [D,a_1] R_s(\lambda)
       \cdots [D,a_m]R_s(\lambda) \,\mathrm{d}\lambda\right)\mathrm{d}s, 
\end{align*}
where
$$ \eta_m = \left(-\sqrt{2i}\right)^{\bullet} 2^{m+1} \frac{\Gamma(m/2+1)}{\Gamma(m+1)} $$
with $\bullet =0,1$ depending on whether the spectral triple is even or odd.
\end{defn}

The integral over $\ell$ is well-defined by~\cite[Lemma 3.3]{CGRS2}. 
The index formula is a pairing of a cocycle with an algebraic chain. 
If $e\in\calA^\sim$ is a projection, we define $\mathrm{Ch}^0(e) =e$ 
and for $k\geq 1$,
$$ 
\mathrm{Ch}^{2k}(e) = (-1)^k\frac{(2k)!}{k!}\,(e-1/2)\otimes e\otimes\cdots\otimes e\in(\calA^\sim)^{\otimes (2k+1)}. 
$$
If $u\in\calA^\sim$ is a unitary, then we define for $k\geq 0$
$$ 
\mathrm{Ch}^{2k+1}(u) = (-1)^k k!\, u^*\otimes u\otimes \cdots\otimes u^*\otimes u \in (\calA^\sim)^{\otimes (2k+2)}. 
$$
We split up the theorem into odd and even cases.

\begin{thm}[\cite{CoM, RennieSummability, CGRS2}] \label{thm:local_index_formula_odd}
Let $(\calA, \calH, D)$ be an odd smoothly summable complex semifinite 
spectral triple relative to $(\calN,\tau)$ and with spectral dimension $p$. 
Let $N = \lfloor\frac{p}{2}\rfloor+1$, where $\lfloor\cdot\rfloor$ 
is the floor function, and let $u$ be a unitary in the unitisation 
of $\calA$. The semifinite index pairing can be computed with the resolvent cocycle
$$  
\langle [u],[(\calA,\calH,D)]\rangle = \frac{-1}{\sqrt{2\pi i}} \res_{r=(1-p)/2} \sum_{m=1,\mathrm{odd}}^{2N-1} \phi_m^r(\mathrm{Ch}^m(u)) 
$$
and the function
$r\mapsto \sum\limits_{m=1,\mathrm{odd}}^{2N-1} \!\phi_m^r(\mathrm{Ch}^m(u))$ 
analytically continues to a deleted neighbourhood of $r=(1-p)/2$.
\end{thm}

\begin{thm}[\cite{CoM, RennieSummability, CGRS2}] \label{thm:local_index_even}
Let $(\calA, \calH, D)$ be an even smoothly summable complex semifinite 
spectral triple relative to $(\calN,\tau)$ and with spectral dimension $p$. 
Let $N = \lfloor\frac{p+1}{2}\rfloor$ and 
$e\in\calA^\sim$ be a self-adjoint projection. The semifinite index pairing 
can be computed by the resolvent cocycle
$$  
\langle [e]-[1_e], [(\calA,\calH,D)]\rangle = \res_{r=(1-p)/2} \sum_{m=0,\mathrm{even}}^{2N}\! \phi_m^r(\mathrm{Ch}^m(e) - \mathrm{Ch}^m(1_e)) 
$$
and the function 
$r\mapsto \sum\limits_{m=0,\mathrm{even}}^{2N}\! \phi_m(\mathrm{Ch}^m(e))$ 
analytically continues to a deleted neighbourhood of $r=(1-p)/2$.
\end{thm}


\end{document}